%% file: Thiele.tex
\def\year{2024}\relax
\documentclass[letterpaper]{article} % DO NOT CHANGE THIS
\usepackage{aaai24}  % DO NOT CHANGE THIS
\usepackage{times}  % DO NOT CHANGE THIS
\usepackage{helvet}  % DO NOT CHANGE THIS
\usepackage{courier}  % DO NOT CHANGE THIS
\usepackage[hyphens]{url}  % DO NOT CHANGE THIS
\usepackage{graphicx} % DO NOT CHANGE THIS
\usepackage{natbib}  % DO NOT CHANGE THIS AND DO NOT ADD ANY OPTIONS TO IT
\usepackage{caption} % DO NOT CHANGE THIS AND DO NOT ADD ANY OPTIONS TO IT
\DeclareCaptionStyle{ruled}{labelfont=normalfont,labelsep=colon,strut=off} % DO NOT CHANGE THIS
\usepackage{microtype}
\nocopyright
\title{Refined Characterizations of Approval-Based Committee Scoring Rules}
\author {
    Chris Dong, % Note that the comma should be placed BEFORE the superscript for optimum readability
    Patrick Lederer
}
\affiliations {
    School of Computation, Information, and Technology, Technical University of Munich \\
    chris.dong@tum.de, ledererp@in.tum.de % email address must be in roman text type, not monospace or sans serif
}

\usepackage{fix-cm}
	\usepackage{xspace}
	\usepackage{amsmath}
	\usepackage{amsthm}
    \usepackage{amsfonts}
    \usepackage{amssymb}

	\usepackage{nicefrac}
	\usepackage{tikz}
	\usetikzlibrary{positioning,arrows,shapes,calc}
	\usepackage{colortbl}
	\usepackage{booktabs}
	\usepackage{cleveref}
	\usepackage{array}
	\usepackage{graphicx}
	\usepackage{xcolor}
	\usepackage{enumitem}
	\usepackage{tabularx}
	\usepackage{thm-restate}
	\usepackage{ragged2e}
    \usepackage{microtype}
    \usepackage{cleveref}
    
    \def\multiset#1#2{\ensuremath{\left(\kern-.3em\left(\genfrac{}{}{0pt}{}{#1}{#2}\right)\kern-.3em\right)}}

	\theoremstyle{definition}

	\theoremstyle{plain}
	
	\newtheorem{lemma}{Lemma}
	
	\newtheorem{proposition}{Proposition}

\begin{document}

\maketitle

\begin{abstract}
In approval-based committee (ABC) elections, the goal is to select a fixed size subset of the candidates, a so-called committee, based on the voters' approval ballots over the candidates. One of the most popular classes of ABC voting rules are ABC scoring rules, which have recently been characterized by \citet{LaSk21a}. However, this characterization relies on a model where the output is a ranking of committees instead of a set of winning committees and no full characterization of ABC scoring rules exists in the latter standard setting. We address this issue by characterizing two important subclasses of ABC scoring rules in the standard ABC election model, thereby both extending the result of \citet{LaSk21a} to the standard setting and refining it to subclasses. In more detail, by relying on a consistency axiom for variable electorates, we characterize \emph{(i)} the prominent class of Thiele rules and \emph{(ii)} a new class of ABC voting rules called ballot size weighted approval voting. Based on these theorems, we also infer characterizations of three well-known ABC voting rules, namely multi-winner approval voting, proportional approval voting, and satisfaction approval voting.
\end{abstract}

\section{Introduction}

An important problem for multi-agent systems is collective decision making: given the voters' preferences over a set of alternatives, a common decision has to be made. This problem has traditionally been studied by economists for settings where a single candidate is elected \citep{ASS02a}, but there is also a multitude of applications where a fixed number of the candidates needs to be elected. The archetypal example for this is the election of a city council, but there are also technical applications such as recommender systems \citep{SFL16a,GaFa22a}. In social choice theory, this type of elections is typically called \emph{approval-based committee (ABC) elections} and has recently attracted significant attention \citep[e.g.,][]{ABC+16a,FSST17a,LaSk22b}. In more detail, the research on these elections focuses on \emph{ABC voting rules}, which are functions that choose a set of winning committees (i.e., fixed size subsets of the candidates) based on the voters' approval ballots (i.e., the sets of candidates that the voters approve).

One of the most important classes of ABC voting rules are ABC scoring rules \citep[see, e.g.,][]{LaSk21a}. These rules generalize the idea of single-winner scoring rules to ABC elections: each voter assigns points to each committee according to some scoring function and the winning committees are those with the maximal total score. There are many well-known examples of ABC scoring rules, such as multi-winner approval voting (\texttt{AV}), satisfaction approval voting (\texttt{SAV}), Chamberlin-Courant approval voting (\texttt{CCAV}), and proportional approval voting (\texttt{PAV}). Moreover, ABC scoring rules are a superset of the prominent class of Thiele rules. 

In a recent breakthrough result, \citet{LaSk21a} have formalized the relation between ABC scoring rules and single-winner scoring rules by characterizing ABC scoring rules with almost the same axioms as \citet{Youn75a} uses for his influential characterization of single-winner scoring rules. In more detail, \citet{LaSk21a} show that ABC scoring rules are the only ABC ranking rules that satisfy the axioms of anonymity, neutrality, continuity, weak efficiency, and consistency. However, this result discusses ABC ranking rules, which return transitive rankings of committees, whereas the literature on ABC elections typically focuses on sets of winning committees as output. Hence, this theorem does not allow for characterizations of ABC scoring rules in the standard ABC voting setting.

While \citet{LaSk21b} also present a result for the standard ABC election setting, the proof of this result is incomplete.\footnote{Roughly, the proof of \citet{LaSk21b} works by constructing an ABC ranking rule $g$ based on an ABC voting rule $f$ that satisfies the given axioms. Then, \citet{LaSk21b} show that $g$ is an ABC scoring rule, which implies that $f$ is an ABC scoring rule in the choice setting. However, the authors never show that $g$ returns transitive rankings, which is required by definition of ABC ranking rules. Closing this gap seems surprisingly difficult.} Moreover, even when the proof could be fixed, this result is not a full characterization of ABC scoring rules as it needs a technical axiom called $2$-non-imposition. This axiom is, e.g., violated by \texttt{AV} and \texttt{SAV} and hence, characterizations of important ABC voting rules---and more generally tools to easily infer such results---are still missing. \citeauthor{LaSk21b} (\citeyear{LaSk21b}, p. 16) also acknowledge this shortcoming by writing that ``a full characterization of ABC scoring rules within the class of ABC choice rules remains as important future work''. 

\paragraph{Our contribution.} We address this problem by presenting full axiomatic characterizations of two important subclasses of ABC scoring rules, namely Thiele rules and ballot size weighted approval voting (BSWAV) rules, in the standard ABC election setting. Hence, our results refine the result of \citet{LaSk21a} to subclasses and extend it to the standard ABC voting setting. {Thiele rules} are ABC scoring rules that do not depend on the ballot size and have attracted significant attention \citep[e.g.,][]{ABC+16a,SFL16a,BLS18a}. On the other hand, {BSWAV rules} generalize multi-winner approval voting by weighting voters depending on the size of their ballots. So far, the class of BSWAV rules has only been studied for single-winner elections \citep{AlVo09a} but not for ABC elections. For example, \texttt{PAV} and \texttt{CCAV} are Thiele rules, \texttt{SAV} is a BSWAV rule, and \texttt{AV} is in both classes. Moreover, every ABC scoring rule that has been studied in the literature is in one of our two classes.

For our results, we mainly rely on the axioms of \citet{LaSk21a}: anonymity, neutrality, continuity, weak efficiency, and consistency. The first four of these axioms are mild standard conditions that are satisfied by every reasonable ABC voting rule. By contrast, consistency is central for our proofs. This axiom requires that if some committees are chosen for two disjoint elections, then precisely these committees should win in a joint election, and it features in several prominent results in social choice theory \citep[e.g.,][]{Youn75a,YoLe78a,Fish78d}. 

To characterize Thiele rules, we need one more axiom called {independence of losers}. This condition requires that a winning committee $W$ stays winning if some voters change their ballot by disapproving ``losing'' candidates outside of $W$ as, intuitively, the quality of $W$ should only depend on its members.
Similar conditions are well-known for single-winner elections \citep[e.g.,][]{BrPe19a} and this axiom has recently been adapted to ABC voting by \citet{DoLe22a}. We then show that \emph{an ABC voting rule is a Thiele rule if and only if it satisfies anonymity, neutrality, consistency, continuity, and independence of losers (\Cref{thm:Thiele}).} 

For our characterization of BSWAV rules, we introduce a new axiom called {choice set convexity}. This condition requires that if two committees are chosen, then all committees ``in between'' those committees are chosen, too: if $W$ and $W'$ are chosen, then all committees $W''$ with $W\cap W'\subseteq W''\subseteq W\cup W'$ are also chosen. We believe that this axiom is reasonable for excellence-based elections (where only the individual quality of the candidates matters) as a tie between committees indicates that they are equally good and the candidates in $W\setminus W'$ and $W'\setminus W$ are thus exchangeable.
We then prove that \emph{an ABC voting rule is a BSWAV rule if and only if it satisfies anonymity, neutrality, consistency, continuity, weak efficiency, and choice set convexity (\Cref{thm:BSWAV}).} 

While our theorems are intuitively related to the results of \citet{LaSk21b,LaSk21a}, they are logically independent. In particular, in contrast to their results, our theorems allow for simple characterizations of \emph{all} Thiele rules and BSWAV rules in the \emph{standard ABC voting model}. We also demonstrate this point in \Cref{sec:specific} by axiomatizing \texttt{AV}, \texttt{SAV}, and \texttt{PAV}. In more detail, we obtain full characterizations of these rules by analyzing axioms for party-list profiles (where candidates are partitioned into parties and each voter approves all candidates of a single party) that formalize when all candidates of a party are chosen. To the best of our knowledge, the result for \texttt{SAV} is the first full characterization of this rule. An overview of our results is given in \Cref{fig:overview}.

\begin{figure}[t]
		\centering
	      \begin{tikzpicture}
		  \tikzstyle{arrow}=[->,>=angle 60, shorten >=1pt,draw]
		  \tikzstyle{mynode}=[align=center, anchor=north]

		  		\node[mynode] (Scoring) at (0,-0.2)   {ABC scoring rules};
		  		\node[mynode] (Thiele) at (-2.1, -1){Thiele rules\\\emph{\small Anonymity, Neutrality, }\\\emph{\small Consistency, Continuity,}\\\emph{\small Independence of Losers}};
				\node[mynode] (BSW) at (2.1, -1) {BSWAV rules\\\emph{\small Anonymity, Neutrality, Consistency,}\\\emph{\small Continuity, Choice Set Convexity,}\\\emph{\small Weak Efficiency}};
				
				\node[mynode] (SAV)  at (1.6, -4.1){\texttt{SAV}};
				\node[mynode] (AV)  at (-0.9, -4.1){\texttt{AV}};
				\node[mynode] (PAV)  at (-3.3, -4.1){\texttt{PAV}};
		
		  		\draw[-latex] (Scoring) to (-1.8,-1);
		  		\draw[-latex] (Scoring) to (1.8,-1);
[]				\draw[-latex] (-0.9,-2.7) to node[midway, right,align=center]{\emph{\small Excellence}\\\emph{\small criterion}} (AV);
				\draw[-latex] (-3.3,-2.7) to node[midway, right,align=center]{\emph{\small Party-propor-}\\\emph{\small tionality}} (PAV);
				\draw[-latex] (1.6, -2.7) to node[midway, right,align=center]{\emph{\small Party-proportionality,}\\\emph{\small aversion to unanimous}\\\emph{\small committees}} (SAV);
	      \end{tikzpicture}
		  \caption{Overview of our results. An arrow from $X$ to $Y$ means that $Y$ is a subset or an element of $X$. The axioms written on an arrow from $X$ to $Y$ characterize the rule $Y$ within the class $X$. The axioms written below Thiele rules and BSWAV rules characterize these classes of ABC voting rules.}
		  \label{fig:overview}
\end{figure}
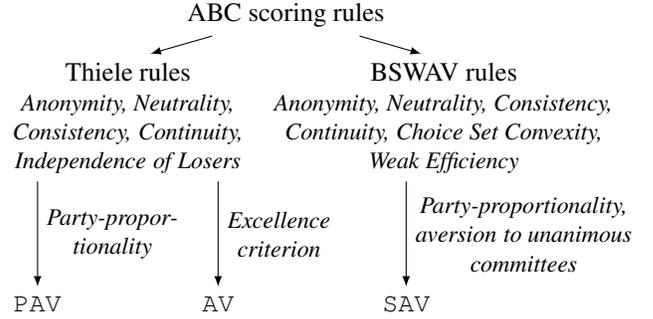 

\paragraph{Related work.}
The lack of axiomatic characterizations is one of the major open problems in the field of ABC voting \citep[see, e.g.,][Q1]{LaSk22b}, and there are thus only few closely related papers. Maybe the most important one is due to \citet{LaSk21a} who characterize ABC scoring rules in the context of ABC ranking rules; however, this result does not allow for characterizations of ABC scoring rules in the standard setting. The follow-up paper by \citet{LaSk21b} tries to fix this issue, but its proof is incomplete and the main result requires a technical auxiliary condition that rules out important rules such as \texttt{AV} and \texttt{SAV}. Moreover, \citet{DoLe22a} characterize committee monotone ABC voting rules, which can be seen as greedy approximations of ABC scoring rules. Finally, committee scoring rules have also been analyzed for the case that voters report ranked ballots, but the results for this setting are also restricted to characterizations of committee ranking rules \citep{SFS19a} or partial characterizations within the class of committee scoring rules \citep{EFSS17a,FSST19a}. 

Furthermore, a large amount of papers studies axiomatic properties of ABC scoring rules \citep[e.g.,][]{LaSk18a,ABC+16a,SaFi19a,BLS18a,LaSk20a}. For instance, \citet{ABC+16a} investigate Thiele rules with respect to how fair they represent groups of voters with similar preferences, and \citet{SaFi19a} study monotonicity conditions for several ABC scoring rules. Another important aspect of these rules is their computational complexity. In particular, it is known that all Thiele rules but \texttt{AV} are $\mathtt{NP}$-hard to compute on the full domain \citep{AGG+15a,SFL16a}. There is thus significant work on how to compute these rules by, e.g., restricting the domain of preference profiles \citep{ElLa15a,Pete18b}, studying approximation algorithms \citep{DMMS20a,BFGG22a}, or designing FPT algorithms \citep{BFK+20a}. For a more detailed overview on ABC scoring rules, we refer to the survey by \citet{LaSk22b}. 

Finally, in the broader realm of social choice, there are numerous conceptually related results as consistency features in many prominent theorems: for instance, \citet{Youn75a} has characterized scoring rules for single-winner elections based on this axiom \citep[see also][]{Smit73a,Myer95b,Piva13a}, numerous characterizations of single-winner approval voting rely on consistency \citep{Fish78d,BrPe19a}, \citet{YoLe78a} have characterized Kemeny's rule with the help of this axiom, and \citet{Bran13a} characterize a randomized voting rule called maximal lotteries based on this condition.

\section{Preliminaries}

Let $\mathbb{N}=\{1,2,\dots\}$ denote an infinite set of voters and let $\mathcal{C}=\{c_1,\dots, c_m\}$ denote a set of $m\geq 2$ candidates. Intuitively, we interpret $\mathbb{N}$ as the set of all possible voters and a concrete electorate $N$ is a finite and non-empty subset of $\mathbb{N}$. We thus define $\mathcal{F}(\mathbb{N})=\{N\subseteq\mathbb{N} \colon N \textit{ is non-empty and finite}\}$ as the set of all possible electorates. Given an electorate $N\in\mathcal{F}(\mathbb{N})$, we assume that each voter $i\in N$ reports her preferences over the candidates as \emph{approval ballot} $A_i$, i.e., as a non-empty subset of $\mathcal{C}$. $\mathcal{A}$ is the set of all possible approval ballots. An \emph{(approval) profile} $A$ is a mapping from $N$ to $\mathcal{A}$, i.e., it assigns an approval ballot to every voter in the given electorate. Next, we define $\mathcal{A}^*=\bigcup_{N\in\mathcal{F}(\mathbb{N})}\mathcal{A}^N$ as the set of all approval profiles. For every profile $A\in\mathcal{A}^*$, $N_A$ denotes the set of voters that submit a ballot in $A$. Finally, two approval profiles $A,A'$ are called \emph{disjoint} if $N_A\cap N_{A'}=\emptyset$ and for disjoint profiles $A,A'$, we define the profile $A''=A+A'$ by $N_{A''}=N_A\cup N_{A'}$, $A''_i=A_i$ for $i\in N_A$, and $A''_i=A'_i$ for $i\in N_{A'}$.

Given an approval profile, our aim is to elect a \emph{committee}, i.e., a subset of the candidates of predefined size. We denote the target committee size by $k\in\{1,\dots,m-1\}$ and the set of all size-$k$ committees by $\mathcal{W}_k=\{W\subseteq\mathcal{C}\colon |W|=k\}$. For determining the winning committees for a given preference profile, we use \emph{approval-based committee (ABC) voting rules} which are mappings from $\mathcal A^* $ to $2^{\mathcal{W}_k}\setminus\{\emptyset\}$. Note that we define ABC voting rules for a fixed committee size and may return multiple committees. The first condition is for notational convenience and the second one is necessary to satisfy basic fairness conditions.

\subsection{ABC Voting Rules}\label{subsec:rules}

We focus in this paper on two classes of ABC voting rules, namely Thiele rules and BSWAV rules, which are both refinements of the class of ABC scoring rules. 

\paragraph{ABC scoring rules.} ABC scoring rules rely on a scoring function according to which voters assign points to committees and choose the committees with maximal total score. Formally, a \emph{scoring function} $s(x,y)$ is a mapping from $\{0,\dots,k\}\times \{1,\dots, m\}$ to $\mathbb{R}$ such that $s(x,y)\geq s(x',y)$ for all $x,x'\in \{\max(0, k+y-m), \dots, \min(k,y)\}$ with $x\geq x'$. We define the score of a committee $W$ in a profile $A$ as $\hat s(A,W)=\sum_{i\in N_A} s(|A_i\cap W|, |A_i|)$. Then, an ABC voting rule $f$ is an \emph{ABC scoring rule} if there  is a scoring function $s$ such that $f(A)=\{W\in\mathcal{W}_k\colon \forall W'\in\mathcal{W}_k\colon \hat s(A,W)\geq \hat s(A,W')\}$ for all profiles $A\in\mathcal{A}^*$. The set $\{\max(0,k+y-m), \dots, \min(k,y)\}$ contains all ``active'' intersection sizes: a committee of size $k$ and a ballot of size $y$ intersect at least in $\max(0, k+y- m)$ candidates and at most in $\min(k,y)$ candidates.

\paragraph{Thiele rules.} Arguably the most prominent subclass of ABC scoring rules are Thiele rules. These rules, which have first been suggested by their namesake \citet{Thie95a}, are ABC scoring rules that ignore the ballot size. Hence, Thiele rules are defined by a non-decreasing \emph{Thiele scoring function} $s:\{0,\dots, k\}\rightarrow \mathbb{R}$ with $s(0)=0$, and choose the committees that maximize the total score. Formally, an ABC voting rule $f$ is a \emph{Thiele rule} if there is a Thiele scoring function $s$ such that $f(A)=\{W\in\mathcal{W}_k\colon \forall W'\in\mathcal{W}_k\colon \hat s(A,W)\geq \hat s(A,W')\}$ for all profiles $A\in\mathcal{A}^*$, where $\hat s(A,W)=\sum_{i\in N_A} s(|A_i\cap W|)$. There are numerous important Thiele rules such as multi-winner approval voting (\texttt{AV}; defined by $s_{\texttt{AV}}(x)=x$), proportional approval voting (\texttt{PAV}; defined by $s_{\texttt{PAV}}(x)=\sum_{z=1}^x \frac{1}{z}$ for $x>0$), and Chamberlin-Courant approval voting (\texttt{CCAV}; defined by $s_{\texttt{CCAV}}(x)=1$ for $x>0$).

\paragraph{BSWAV rules.} Ballot size weighted approval voting rules form a new subclass of ABC scoring rules which generalize \texttt{AV} by weighting voters based on their ballot size. Formally, a \emph{ballot size weighted approval voting (BSWAV) rule} $f$ is defined by a weight vector $\alpha\in \mathbb{R}^m_{\geq 0}$ and chooses for every profile $A$ the committees $W$ that maximize $\hat s(A,W)=\sum_{i\in N_A} \alpha_{|A_i|}|A_i\cap W|$. The score of a committee $W$ for a BSWAV rule can be represented as the sum of the scores of individual candidates $c\in W$ since $\sum_{i\in N_A} \alpha_{|A_i|}|A_i\cap W|=\sum_{c\in W}\sum_{i\in N_A\colon c\in A_i} \alpha_{|A_i|}$. Clearly, \texttt{AV} is the BSWAV rule defined by $\alpha_x = 1$ for all $x\in \{1,\dots,m\}$. Another well-known BSWAV rule is satisfaction approval voting (\texttt{SAV}) defined by $\alpha_x=\frac{1}{x}$ for $x\in \{1,\dots,m\}$. This rule has been popularized by \citet{BrKi14a} for ABC elections, but it has been studied before by, e.g., \citet{AlVo09a} and \citet{KiMa12a}.
\medskip

We note that Thiele rules and BSWAV rules are diametrically opposing subclasses of ABC scoring rules: Thiele rules do not depend on the ballot size at all, whereas BSWAV rules only depend on this aspect. Consequently, if $k<m-1$, the sets of BSWAV rules and Thiele rules only intersect in \texttt{AV} and the trivial rule \texttt{TRIV} (which always chooses all size $k$ committees). So, \texttt{AV} is the only non-trivial ABC voting rule that is in both classes; \emph{non-triviality} means here that there is a profile $A$ such that $f(A)\neq \texttt{TRIV}(A)$. Moreover, both classes are proper subsets of the set of ABC scoring rules if $1<k<m-1$. By contrast, the set of BSWAV rules is equivalent to the set of ABC scoring rules if $k\in \{1, m-1\}$. 

\subsection{Basic Axioms}\label{subSec:Axioms}

Next, we introduce the axioms used for our characterizations.

\paragraph{Anonymity.} Anonymity is one of the most basic fairness properties and requires that all voters should be treated equally. Formally, we say an ABC voting rule $f$ is \emph{anonymous} if $f(A)=f(\pi(A))$ for all profiles $A\in\mathcal{A}^*$ and permutations $\pi:\mathbb{N}\rightarrow\mathbb{N}$. Here, we denote by $A'=\pi(A)$ the profile with $N_{A'}=\{\pi(i)\colon i\in N_A\}$ and $A'_{\pi(i)}=A_i$ for all $i\in N_A$. 

\paragraph{Neutrality.} Similar to anonymity, \emph{neutrality} is a fairness property for the candidates. This axiom requires of an ABC voting rule $f$ that $f(\tau(A))=\{\tau(W)\colon W\in f(A)\}$ for all profiles $A\in\mathcal{A}^*$ and permutations $\tau:\mathcal{C}\rightarrow\mathcal{C}$. This time, $A'=\tau(A)$ denotes the profile with $N_{A'}=N_A$ and $A'_{i}=\tau(A_i)$ for all $i\in N_A$.

\paragraph{Weak Efficiency.} Weak efficiency requires that unanimously unapproved candidates can never be ``better'' than approved ones. Formally, we say an ABC voting rule $f$ is \emph{weakly efficient} if $W\in f(A)$ for a committee $W\in\mathcal{W}_k$ with $c\in W\setminus (\bigcup_{i\in N_A} A_i)$ implies that $(W\cup \{c'\})\setminus \{c\}\in f(A)$ for all candidates $c'\in \mathcal{C}\setminus W$.

\paragraph{Continuity.} The intuition behind continuity is that a large group of voters should be able to enforce that some of its desired outcomes are chosen. Hence, an ABC voting rule $f$ is \emph{continuous} if for all profiles $A,A'\in\mathcal{A}^*$, there is $\lambda\in \mathbb{N}$ such that $f(\lambda A+A')\subseteq f(A)$. Here, $\lambda A$ denotes the profile consisting of $\lambda$ copies of $A$; the names of the voters in $N_{\lambda A}$ will not matter as we will focus on anonymous rules. 

\paragraph{Consistency.} The central axiom for our results is consistency. This condition states that if some committees are chosen for two disjoint profiles, then precisely those committees are chosen in the joint profile. Formally, an ABC voting rule $f$ is \emph{consistent} if $f(A+A')=f(A)\cap f(A')$ for all disjoint profiles $A, A'\in\mathcal{A}^*$ with $f(A)\cap f(A')\neq \emptyset$. Consistency and the previous four axioms have been introduced by \citet{LaSk21b} for ABC elections. Moreover, except consistency, all these axioms are very mild and satisfied by almost all commonly considered ABC voting rules. 

\paragraph{Independence of Losers.} Independence of losers has been adapted to ABC elections by \citet{DoLe22a} and requires of an ABC voting rule $f$ that a winning committee $W$ should still be a winning committee if voters disapprove candidates outside of $W$. Or, put differently, whether a committee $W$ wins should not depend on the  voters' approvals of ``losing'' candidates not in $W$. We hence say an ABC voting rule $f$ is \emph{independent of losers} if $W\in f(A)$ implies that $W\in f(A')$ for all profiles $A,A'\in\mathcal{A}^*$ and committees $W\in\mathcal{W}_k$ such that $N_A=N_{A'}$ and $W\cap A_i=W\cap A_i'$ and $A_i'\subseteq A_i$ for all voters $i\in N_A$. The motivation for this axiom is that the quality of $W$ should only depend on the candidates in $W$. So, if some voters disapprove candidates $x\not\in W$, the quality of $W$ is not affected and a chosen committee $W$ should stay chosen. All commonly studied ABC voting rules that are independent of the ballot size (e.g., Thiele rules, Phragm\'en's rule, and sequential Thiele rules) satisfy this axiom, whereas all BSWAV rules except \texttt{AV} fail~it.

\paragraph{Choice Set Convexity.} Finally, we introduce a new condition called choice set convexity: an ABC voting rule $f$ is \emph{choice set convex} if $W,W'\in f(A)$ implies that $W''\in f(A)$ for all committees $W, W', W''\in\mathcal{W}_k$ and profiles $A\in\mathcal{A}^*$ such that $W\cap W'\subseteq W''\subseteq W\cup W'$. More informally, this axiom states that if a rule chooses two committees $W$ and $W'$, then all committees ``between'' $W$ and $W'$ are also chosen. We believe that choice set convexity is reasonable in elections in which only the individual quality of the elected candidates matters. For example, if we want to hire $3$ applicants for independent jobs based on the interviewers' preferences, it seems unreasonable that the sets $\{c_1,c_2,c_3\}$ and $\{c_1,c_4,c_5\}$ are good enough to be hired but $\{c_1, c_2, c_4\}$ is not. More generally, we can interpret the membership of a candidate in a chosen committee as certificate for its quality and all candidates $c\in (W\setminus W')\cup (W'\setminus W)$ are then equally good. Many commonly considered voting rules fail this axiom, but one can always compute the ``convex hull'' of a choice set.

\section{Results}\label{sec:classes}

We are now ready to state our results. In particular, we discuss the characterizations of Thiele rules and BSWAV rules in \Cref{subsec:Thiele,subsec:BSWAV}, respectively. Moreover, we present characterizations of \texttt{AV}, \texttt{PAV}, and \texttt{SAV} in \Cref{sec:specific}. Due to space constraints, we defer most proofs to the appendix and give proof sketches instead. 

\subsection{Characterization of Thiele Rules}
\label{subsec:Thiele}

We now turn to our first characterization: Thiele rules are the only ABC voting rules that are anonymous, neutral, consistent, continuous, and independent of losers. We thus turn the result of \citet{LaSk21a} into a characterization of Thiele rules in the standard ABC voting model by replacing weak efficiency with independence of losers. 

\begin{restatable}{theorem}{Thiele}\label{thm:Thiele}
    An ABC voting rule is a Thiele rule if and only if it satisfies anonymity, neutrality, consistency, continuity, and independence of losers. 
\end{restatable}
\begin{proof}[Proof Sketch]
    First, suppose that $f$ is a Thiele rule and let $s(x)$ denote its Thiele scoring function. Clearly, $f$ is anonymous, neutral, consistent, and continuous as all ABC scoring rules satisfy these axioms. So, we will only show that $f$ is independent of losers. For this, consider two profiles $A,A'\in\mathcal{A}^*$ and a committee $W\in f(A)$ such that $N_A=N_{A'}$ and $A_i'\subseteq A_i$ and $W\cap A_i'=W\cap A_i$ for all $i\in N_A$. It holds that $\hat s(A',W)=\hat s(A,W)$ since $W\cap A_i'=W\cap A_i$ for all $i\in N_A$. Moreover, $\hat s(A,W)\geq \hat s(A,W')$ for all $W'\in\mathcal{W}_k$ because $W\in f(A)$. Finally, $\hat s(A, W')\geq \hat s(A',W')$ for all $W'\in\mathcal{W}_k$ as $s(x)$ is non-decreasing and $A_i'\subseteq A_i$ for all $i\in N_A$. By chaining the inequalities, we conclude that $\hat s(A',W)\geq \hat s(A',W')$ for all committees $W'\in\mathcal{W}_k$, so $W\in f(A')$ and $f$ satisfies independence of losers.

    For the other direction, we suppose that $f$ is an ABC voting rule that satisfies all axioms of the theorem and aim to show that $f$ is a Thiele rule. For this, we will use the separating hyperplane theorem for convex sets similar to the works of, e.g., \citet{Youn75a} and \citet{SFS19a}. For this, we note first that, if $f$ is trivial, it is the Thiele rule defined by $s(x)=0$ for all $x$. So, we suppose that $f$ is non-trivial and show that for every committee $W\in\mathcal{W}_k$, there is a profile $A\in\mathcal{A}^*$ such that $f(A)=\{W\}$. To apply the separating hyperplane theorem for convex sets, we next extend $f$ to a function $\hat g$ of the type $\mathbb{Q}^{|\mathcal{A}|}\rightarrow 2^{\mathcal{W}_k}\setminus \{\emptyset\}$ while keeping all its properties intact. We then define the sets $R_i^f=\{v\in\mathbb{Q}^{|\mathcal{A}|}\colon W^i\in \hat g(v)\}$ for all $W^i\in\mathcal{W}_k$ and let $\bar R_i^f$ denote the closure of $R_i^f$ with respect to $\mathbb{R}^{|\mathcal{A}|}$. It follows from the properties of $\hat g$ that the sets $\bar R_i^f$ are convex and have disjoint interiors. The separating hyperplane theorem for convex sets thus shows that there are non-zero vectors $\hat u^{i,j}\in\mathbb{R}^{|\mathcal{A}|}$ such that $v\hat u^{i,j}\geq 0$ if $v\in \bar R_i$ and ${v\hat u^{i,j}\leq 0}$ if $v\in \bar R_j$. Moreover, we will show that $\bar R_i^f=\{v\in\mathbb{R}^{|\mathcal{A}|}\colon \forall W^j\in\mathcal{W}_k\setminus \{W^i\}\colon v\hat u^{i,j}\geq 0\}$, so we study the vectors $\hat u^{i,j}$ next.

    For this, we first infer from neutrality and independence of losers that there is a function $s^1(x,y)$ such that $\hat u^{i,j}_\ell=s^1(|W^i\cap A_\ell|, |W^j\cap A_\ell|)$ for all ballots $A_\ell$ and committees $W^i, W^j$ with $|W^i\setminus W^j|=1$. If $k\in \{1,m-1\}$, this insight is already enough for the proof. By contrast, if $k\in \{2,\dots,m-2\}$, we need to analyze the vectors $\hat u^{i,j}$ for committees $W^i, W^j$ with $|W^i\setminus W^j|=t>1$. To this end, we construct a sequence of committees $W^{j_0},\dots, W^{j_t}$ by replacing the candidates in $W^i\setminus W^j$ one after another with those in $W^j\setminus W^i$. By studying the linear (in)dependence of the vectors $\hat u^{i,j}$ and $\hat u^{j_{x-1}, j_{x}}$ for $x\in \{1,\dots, t\}$, we then show that $\hat u^{i,j}=\delta \sum_{x=1}^{t} \hat u^{j_{x-1}, j_{x}}$ for some $\delta>0$. Based on this insight, we can now define the score function $s$ of $f$: $s(0)=0$ and $s(x)=s(x-1)+s^1(x,x-1)$ for $x\geq 1$. By our previous observations, it follows that $\hat u^{i,j}_\ell=\delta(s(|W^i\cap A_\ell|)-s(|W^j\cap A_\ell|))$, so $\bar R_i^f=\{v\in \mathbb{R}^{|\mathcal{A}|}\colon \forall W^j\in\mathcal{W}_k\colon \hat s(v,W^i)\geq \hat s(v,W^j)\}$. From this, we infer that $\hat g(v)\subseteq \{W^i\in\mathcal{W}_k\colon v\in \bar R_i^f\}=\{W^i\in\mathcal{W}_k\colon \forall W^j\in \mathcal{W}_k\setminus \{W^i\}\colon\hat s(v,W^i)\geq \hat s(v,W^j)\}$ for all $v\in \mathbb{Q}^{|\mathcal{A}|}$. Thus, $f(A)\subseteq \{W^i\in\mathcal{W}_k\colon \forall W^j\in \mathcal{W}_k\setminus \{W^i\}\colon\hat s(A,W^i)\geq \hat s(A,W^j)\}$ and, as the last step, continuity shows $f$ is the Thiele rule induced by $s$.
\end{proof}

\paragraph{Remark 1.} All axioms are required for \Cref{thm:Thiele}. If we omit independence of losers, \texttt{SAV} satisfies all remaining axioms. If we omit continuity, we can refine Thiele rules by applying a second Thiele rule as tie-breaker in case of multiple chosen committees. If we only omit consistency, sequential Thiele rules satisfy all given axioms. These rules compute the winning committees iteratively by always adding the candidate to a winning committee which increases the score the most. If we omit neutrality or anonymity, biased Thiele rules that double the points of every committee that contains a specific candidate or the points assigned by specific voters to the committees satisfy all other axioms.

\subsection{Characterization of BSWAV Rules}\label{subsec:BSWAV}

Next, we discuss the characterization of BSWAV rules: these are the only ABC voting rules that satisfy anonymity, neutrality, consistency, continuity, choice set convexity, and weak efficiency. The central axiom for this characterization (aside of consistency) is choice set convexity as it enforces that candidates can be exchanged between chosen committees.  

\begin{restatable}{theorem}{BSWAV}\label{thm:BSWAV}
    An ABC voting rule is a BSWAV rule if and only if it satisfies anonymity, neutrality, consistency, continuity, choice set convexity, and weak efficiency. 
\end{restatable}
\begin{proof}[Proof Sketch]
    First, we assume that $f$ is a BSWAV rule and let $\alpha=(\alpha_1, \dots, \alpha_m)$ denote its weight vector. It is simple to verify that $f$ is neutral, anonymous, continuous, and consistent. Moreover, $f$ is weakly efficient as the weights $\alpha_i$ are all non-negative. Finally, we show that $f$ is choice set convex.
    %BEGIN added Stuff
    For this, we consider a profile $A$ and two distinct committees $W,W'\in f(A)$. Next, we choose two candidates $a\in W\setminus W'$ and $b\in W'\setminus W$ and let $W'' = (W\setminus \{a\})\cup \{b\}$. The central observation is now that BSWAV scores are additive, i.e., $\hat s(A,W)=\sum_{x\in W} \hat s(A,x)$ for $\hat s(A,x)=\sum_{i\in N_A\colon x\in A_i} \alpha_{|A_i|}$. Since $W\in f(A)$, $0\leq \hat s(A,W)-\hat s(A,W'') = \hat s(A,a)- \hat s(A,b)$. By applying this argument also to $W'$ and $W'''=(W'\setminus \{b\})\cup \{a\}$, we obtain $0\leq \hat s(A,b)- \hat s(A,a)$, so $\hat s(A,a) = \hat s(A,b)$ and $\hat s(A,W)=\hat s(A,W'')$. This proves that $W''\in f(A)$ and by repeating the argument, we infer that $\bar W\in f(A)$ for all $\bar W$ with $W\cap W'\subseteq \bar W\subseteq W\cup W'$.

    For the converse direction, we give again only a rough proof sketch and note that the outline of this proof is very similar to the one of \Cref{thm:Thiele} as mainly the technical details differ. In more detail, we first extend $f$ to a function $\hat g$ on $\mathbb{Q}^{|\mathcal{A}|}$ and then use the same hyperplane argument as for \Cref{thm:Thiele}. Hence, we will again analyze the sets $R_i^f=\{v\in\mathbb{Q}^{|\mathcal{A}|}\colon W^i\in \hat g(v)\}$ and the vectors $\hat u^{i,j}$ with $v\hat u^{i,j}\geq 0$ if $v\in \bar R_{i}^f$ and $v\hat u^{i,j}\leq 0$ if $v\in \bar R_j^f$. In particular, based on choice set convexity, we show for every ballot size $r\in \{1,\dots, m\}$ that there is a constant $\alpha_r\geq 0$ such that $\hat u^{i,j}_\ell=\alpha_r$ for all ballots $A_\ell\in\mathcal{A}$ with $|A_\ell|=r$ and committees $W^i,W^j\in\mathcal{W}_k$ with $|W_i\cap A_\ell|=|W_j\cap A_\ell|+1$. Based on this insight, it is simple to complete the proof if $k\in \{1,m-1\}$. On the other hand, if $k\in \{2,\dots, m-2\}$, we again consider committees $W^i,W^j$ such that $|W^i\setminus W^j|=t>1$. Just as for \Cref{thm:Thiele}, we consider a sequence of committees $W^{j_0},\dots, W^{j_t}$ such that $W^{j_0}=W^i$, $W^{j_t}=W^j$, and $|W^{j_{x-1}}\setminus W^{j_{x}}|=1$ for $x\in \{1,\dots, t\}$, and show that $\hat u^{i,j}=\delta \sum_{x=1}^{t} \hat u^{j_{x-1}, j_x}$ for some $\delta>0$. This implies that $\hat u^{i,j}_\ell=\alpha_r(|W^i\cap A_\ell|-|W^j\cap A_\ell|)$ for all committees $W^i, W^j\in\mathcal{W}_k$ and ballots $A_\ell\in\mathcal{A}$ with $|A_\ell|=r$. Finally, we can now prove that $f$ is the BSWAV rule defined by the score function $s(|W\cap A_\ell|, |A_\ell|)=\alpha_{|A_\ell|} |W\cap A_\ell|$.
\end{proof}

\paragraph{Remark 2.} All axioms are required for \Cref{thm:BSWAV}. For anonymity, neutrality, and continuity, we can define examples similar to the ones given for Thiele rules. When omitting consistency, the ``convex hull'' of Phragm\'en's rule satisfies all remaining axioms and is no BSWAV rule. The rule that elects the $k$ candidates with minimal approval scores satisfies all given axioms but weak efficiency. Finally, every Thiele rule other than \texttt{AV} only fails choice set convexity.

\paragraph{Remark 3.} \texttt{AV} is the only non-trivial ABC voting rule that is both a BSWAV rule and a Thiele rule if $k\leq m-2$. \Cref{thm:Thiele,thm:BSWAV} thus characterize \texttt{AV} as the only non-trivial ABC voting rule that is anonymous, neutral, continuous, consistent, independent of losers, and choice set convex if $k\le m-2$.

\paragraph{Remark 4.}
   We define ABC voting rules for a fixed committee size $k$, but in the literature $k$ is often considered as part of the input. For such rules, \Cref{thm:Thiele,thm:BSWAV} imply that for every $k\in \{1,\dots, m-1\}$, $f(A,k)$ is a Thiele rule or a BSWAV rule, respectively, if it satisfies the required axioms. However, our conditions do not enforce consistency with respect to the committee size, so we can, e.g., use \texttt{AV} for $k=2$ and \texttt{PAV} for $k=3$. It is not difficult to exclude such rules. For instance, the well-known axiom of committee monotonicity \citep[][]{EFSS17a} entails for every BSWAV rule that it must use the same weight vector for every committee size $k$. Similar, committee separability, an axiom introduced by \citet{DoLe22a}, can be used to enforce that non-imposing Thiele rules use the same Thiele scoring function for every committee size. Thus, our results can be easily extended to the setting where the committee size is part of the input.

\subsection{Characterizations of \texttt{AV}, \texttt{PAV}, and \texttt{SAV}}\label{sec:specific}

Finally, we demonstrate in this section how \Cref{thm:Thiele,thm:BSWAV} can be used to characterize specific ABC voting rules. To this end, we first note that there are numerous characterizations of ABC voting rules within the class of Thiele rules in the literature, and \Cref{thm:Thiele} can typically be used to extend these results to full characterizations. For instance, \citet{LaSk18a} characterize \texttt{AV} among the class of Thiele rules based on a strategyproofness notion and it is easy to extend this result to a full characterization of \text{AV} based on \Cref{thm:Thiele}. Similar claims are true for, e.g., the characterization of \texttt{AV} based on committee monotonicity \citep[][]{Jans16a}, the characterization of \texttt{PAV} based on D'Hondt proportionality \citep[e.g.,][]{BLS18a}, or characterizations of \texttt{CCAV} \citep[e.g.,][]{DDE+22a}. In this paper, we will, however, give characterizations of three ABC scoring rules (namely \texttt{AV}, \texttt{PAV}, and \texttt{SAV}) that are largely independent of the literature. The reason for this is that our technique seems rather universal and may thus also be used to characterize further Thiele rules or BSWAV rules. Finally, we will state our results only within the class of Thiele rules and BSWAV rules, respectively; \Cref{thm:Thiele,thm:BSWAV} then generalize these results to full characterizations.

In more detail, for all three results in this subsection, we study axioms defined for special profiles. To this end, we say a profile $A\in\mathcal{A}^*$ is a \emph{party-list profile} if there is a partition $\mathcal{P}_A=\{P_1,\dots,P_\ell\}$ of the candidates such that every voter approves all candidates in one set $P_j$, i.e., for every voter $i\in N_A$, there is a set $P_j\in\mathcal{P}_A$ such that $A_i=P_j$. Less formally, in a party-list profile, the candidates are grouped into disjoint parties and every voter supports a single party by approving all of its members. We denote by $n_j$ the number of voters who support party $P_j$ in a party-list profile $A$. For these profiles, we will investigate the question when a voting rule elects all members of a party. The reason for this design choice is twofold: firstly, this will lead to rather mild axioms which makes our characterizations only stronger. Secondly, on party-list profiles, BSWAV rules typically elect one party after another by first electing all members of the first party, then electing all members of the second party, and so on. Hence, axioms describing when all candidates of a party are elected are well-suited for characterizing these rules.

Clearly, any justification for when all members of a party should be chosen needs to consider the purpose of the election. For instance, if the goal of an election is to find the best $k$ candidates only based on their individual quality (a setting known as excellence-based elections), the main criterion for deciding whether to choose a candidate is the number of voters supporting it. Hence, if a party $P_i$ is approved by more voters than another party $P_j$, then every candidate in $P_i$ seems better than every candidate in $P_j$. Thus, if all candidates of party $P_j$ are chosen, all candidates of party $P_i$ should also be chosen. We formalize this idea as follows: An ABC voting rule $f$ satisfies the \emph{excellence criterion} if for all party-list profiles $A$, committees $W\in f(A)$, and parties $P_i,P_j\in\mathcal{P}_A$ with $n_i<n_j$, it holds that $P_i\subseteq W$ implies that $P_j\subseteq W$. As we show next, this condition characterizes \texttt{AV} among Thiele rules.

\begin{proposition}\label{prop:AV}
    \texttt{AV} is the only Thiele rule that satisfies the excellence criterion.
\end{proposition}
\begin{proof}
    Clearly, \texttt{AV} satisfies the excellence criterion and we thus focus on the converse direction. For this, let $f$ denote a Thiele rule that satisfies the excellence criterion and let $s$ denote its Thiele scoring function. Our first goal is to show that $s(1)>0$ and we consider for this the party-list profile $A$ in which $2$ voters approve $P_1=\{c_1\}$ and $1$ voter approves $P_2=\{c_2,\dots, c_{k+1}\}$. Now, if $s(1)=0$, then $P_2\in f(A)$ as $s$ is non-decreasing. This, however, violates the excellence criterion as there is a winning committee that contains all members of $P_2$ but none of $P_1$, even though $n_1>n_2$. Hence, $s(1)>0$ and we subsequently suppose that $s(1)=1$ as Thiele rules are invariant under scaling the scoring function.

    Next, we assume for contradiction that there is an index $\ell\in \{2,\dots, k\}$ such that $s(\ell)\neq \ell$ and $s(x)=x$ for all $x<\ell$. Moreover, we define $\Delta=|s(\ell)-\ell|\neq 0$ and let $t\in\mathbb{N}$ such that $t\geq 2$ and $t\Delta>k$. We now use a case distinction with respect to $s(\ell)$ and first suppose that $s(\ell)=\ell+\Delta$. In this case, consider the party-list profile $A$ where $t$ voters approve $P_1=\{c_1,\dots, c_\ell\}$ and each other candidate $c\in \mathcal{C}\setminus P_1$ is uniquely approved by $t+1$ voters. It is easy to verify that every committee $W$ with $P_1\subseteq W$ has a score of $\hat s(A,W)=ts(\ell)+(k-\ell)(t+1)=t\Delta+t\ell +(k-\ell)(t+1)>k+tk$. By contrast, every committee $W'$ with $\ell'=|P_1\cap W'|<\ell$ has a score of $\hat s(A,W')=t s(\ell')+(k-\ell') (t+1)=t\ell'+(k-\ell')(t+1)\leq tk+k$. Thus, $f(A)=\{W\in\mathcal{W}_k\colon P_1\subseteq W\}$. However, this contradicts the excellence criterion since $P_1\subseteq W$ for every $W\in f(A)$ and there is a party $P_j=\{c\}$ with $c\not\in W$ and $n_j>n_1$. 
    
    For the second case, we suppose that $s(\ell)=\ell-\Delta$ and consider the profile $A$ in which $t$ voters approve the party $P_1=\{c_1,\dots, c_\ell\}$ and each candidate $c\in \mathcal{C}\setminus P_1$ is uniquely approved by $t-1$ voters. We compute again the scores of committees $W\in\mathcal{W}_k$: if $P_1\subseteq W$, then $\hat s(A,W)=t s(\ell)+(k-\ell)(t-1)=t\ell-t\Delta+(k-\ell)(t-1)<tk-k$, and if $|W'\cap P_1|=\ell-1$, then $\hat s(A,W')=t s(\ell-1)+(k-\ell+1)(t-1)=t(\ell-1)+(k-\ell+1)(t-1)\geq kt-k$. Hence, $P_1\not\subseteq W$ for all $W\in f(A)$. However, this violates the excellence criterion since for every $W\in f(A)$, there is a party $P_j=\{c\}$ with $P_j\subseteq W$ and $n_j<n_1$. We thus have a contradiction in both cases, so $s(\ell)= \ell$ and $f$ is \texttt{AV}.
\end{proof}

Another frequent goal in committee elections is proportional representation: the chosen committee should proportionally represent the voters' preferences. To this end, we note that if a party $P_i$ with $n_i$ votes gets $x_i$ seats in the chosen committee, then each of the elected candidates in $P_i$ represents on average $\nicefrac{n_i}{x_i}$ voters. Hence, if $\nicefrac{n_i}{x_i}<\nicefrac{n_j}{x_j+1}$, then reassigning one seat from party $P_i$ to party $P_j$ intuitively results in a more representative outcome. We will formalize this intuition with a new proportionality notion since we aim to show that \texttt{SAV} is more proportional than \texttt{AV}, but \texttt{SAV} violates all commonly considered proportionality axioms. In more detail, we say that an ABC voting rule $f$ is \emph{party-proportional} if for all party-list profiles $A$, committees $W\in f(A)$, and parties $P_i, P_j\in\mathcal{P}_A$ with $\nicefrac{n_i}{|P_i|}<\nicefrac{n_j}{|P_j|}$, it holds that $P_i\subseteq W$ implies $P_j\subseteq W$. Intuitively, this axiom states that we can only choose all members of a party if there is no party that represents on average more voters and is not fully chosen yet. Hence, this axiom combines the idea of proportionality with the native behavior of BSWAV rules. Even though party-proportionality is a rather weak axiom as it is, e.g., implied by D'Hondt proportionality \citep{LaSk21a}, we show next that this condition characterizes \texttt{PAV} within the class of Thiele rules. This demonstrates that our new axiom is indeed a reasonable and non-trivial proportionality notion. 

\begin{restatable}{proposition}{PAV}\label{prop:PAV}
    \texttt{PAV} is the only Thiele rule that satisfies party-proportionality. 
\end{restatable}
\begin{proof}[Proof Sketch]
    First, we show that $\texttt{PAV}$ is party-prop\-or\-tio\-nal. To this end, let $A$ denote a party-list profile, consider two parties $P_i$, $P_j\in\mathcal{P}_A$ with $\frac{n_i}{|P_i|}<\frac{n_j}{|P_j|}$, and suppose for contradiction that there is a committee $W\in\texttt{PAV}(A)$ such that $P_i\subseteq W$, $P_j\not\subseteq W$. In this case, exchanging a candidate $x\in W\cap P_i$ with a candidate $y\in P_j\setminus W$ leads to a committee $W'$ with higher \texttt{PAV}-score than $W$, which contradicts that $W\in\texttt{PAV}(A)$. Thus, \texttt{PAV} is party-proportional. For the other direction, we proceed similarly to the proof of \Cref{prop:AV} and let $f$ denote a Thiele rule that is party-proportional and $s$ its Thiele scoring function. First, we show that $s(1)>0$ by the same construction as in the proof of \Cref{prop:AV} and rescale $s$ such that $s(1)=1$. Then, we construct two profiles showing that $f$ fails party-proportionality if $s(\ell)\neq \sum_{x=1}^\ell \nicefrac{1}{x}$ for some $\ell\in \{2,\dots, k\}$. So, $f$ is indeed \texttt{PAV}. 
\end{proof}

It is easy to see that \texttt{SAV} satisfies---in contrast to \texttt{AV}---party-pro\-por\-tiona\-li\-ty, so \texttt{SAV} is more proportional than \texttt{AV}. Even more, party-proportionality characterizes \texttt{SAV} within the class of BSWAV rules when only allowing voters to approve at most $k$ candidates. However, if there is a party $P_i$ with $|P_i|>k$, this is no longer true as not all member of such parties can be elected. We thus introduce another axiom to characterize \texttt{SAV}: an ABC voting rule $f$ satisfies \emph{aversion to unanimous committees} if for all party-list profiles $A$ and parties $P_i\in\mathcal{P}_A$, it holds that $W\subseteq P_i$ for all $W\in f(A)$ implies that $\frac{n_i}{|P_i|}>n_j$ for all other parties $P_j\in\mathcal{P}_A$ with $|P_j|=1$. Intuitively, this axiom is a mild diversity criterion which requires that a single party can only get all seats in the chosen committee if it is approved by a sufficient number of voters when compared to singleton parties. We next characterize \texttt{SAV} based on this this axiom and party-proportionality.

\begin{restatable}{proposition}{SAV}\label{prop:SAV}
    \texttt{SAV} is the only BSWAV rule that satisfies party-proportionality and aversion to unanimous committees.
\end{restatable}
\begin{proof}[Proof Sketch]
    First, it follows immediately from the definition of \texttt{SAV} that it satisfies party-proportionality and aversion to unanimous committees. For the other direction, we consider a BSWAV rule $f$ that satisfies the given axioms and let $\alpha\in\mathbb{R}^{m}_{\geq 0}$ denote its weight vector. From here on, the proof proceeds again just as the one of \Cref{prop:AV}: we first show that $\alpha_1>0$, rescale such that $\alpha_1=1$, and then use a similar construction to infer that $\alpha_\ell=\frac{1}{\ell}$ for all $\ell\in \{1,\dots,m\}$.
\end{proof}

\paragraph{Remark 5.} 
We note that \texttt{PAV} fails aversion to unanimous committees as the ratio is chosen too restrictive: there are party-list profiles $A$ with a party $P_i$ such that $W\subseteq P_i$ for all $W\in\texttt{PAV}(A)$ even though $n_j>\nicefrac{n_i}{|P_i|}$ for a singleton party $P_j$. However, for all such profiles, it holds that $\nicefrac{n_i}{k}>n_j$. In the context of proportional representation, this bound seems more reasonable as it states that each elected member of $P_i$ represents more voters than the single member of $P_j$. Interestingly, party-proportionality together with a variant of this condition (for all party-list profiles $A$ and parties $P_i$, it holds that $W\subseteq P_i$ for all $W\in f(A)$ if and only if $n_j<\nicefrac{n_i}{k}$ for all singleton parties $P_j$) characterize the BSWAV rule defined by the weight vector $\alpha_\ell=\max(\nicefrac{1}{\ell},\nicefrac{1}{k})$ for all $\ell$. This rule is known as modified satisfaction approval voting \citep{KiMa12a} and this observation shows that it might be more desirable than \texttt{SAV}.

\section{Conclusion}

In this paper, we axiomatically characterize two important classes of approval-based committee (ABC) voting rules, namely Thiele rules and BSWAV rules. Thiele rules choose the committees that maximize the total score according to a score function that only depends on the intersection size of the considered committee and the ballots of the voters. On the other hand, BSWAV rules are a new generalization of multi-winner approval voting which weight voters depending on the size of their ballot. For both of our characterizations, the central axiom is consistency which has famously been used by \citet{Youn75a} for a characterization of single-winner scoring rules or by \citet{LaSk21a} for a characterization of ABC scoring rules in the context of committee ranking rules. In particular, our results allow for simple characterizations of all important ABC scoring rules as all such rules belong to one of our classes. We also demonstrate this point by characterizing the well-known ABC voting rules \texttt{AV}, \texttt{SAV}, and \texttt{PAV}. In particular, the result for \texttt{SAV} is, to the best of our knowledge, the first full characterization of this rule. \Cref{fig:overview} shows a more detailed overview of our results. 

Our paper offers several directions for future work. Firstly, our main results allow, of course, to characterize further ABC scoring rules. Secondly, characterizations of many important ABC voting rules (e.g., Phragm\'en's rule and the method of equal shares) are still missing and some of our ideas might be helpful to derive such results. Finally, even though all relevant ABC scoring rules belong to one of our classes, we would find a full characterization of the set of ABC scoring rules still interesting.

\section*{Acknowledgements}
We thank Felix Brandt and Dominik Peters for helpful feedback. This work was supported by the Deutsche Forschungsgemeinschaft under grants \mbox{BR 2312/11-2} and \mbox{BR 2312/12-1}.

\newpage
\include{Thiele_Appendix_AAAI}

\end{document}

%% file: Thiele_Appendix_AAAI.tex
\newpage
\appendix

\section{Omitted Proofs from Section 3}

In this appendix, we will prove \Cref{thm:Thiele,thm:BSWAV}. Since the proofs of these theorems are rather involved, we organize them in subsections: first, we discuss the hyperplane argument for both theorems jointly in \Cref{app:hyperplane}, then consider the (sligthly simpler) proof of \Cref{thm:BSWAV} in \Cref{app:BSWAV}, and finally present the proof of \Cref{thm:Thiele} in \Cref{app:Thiele}.

Moreover, for our proofs, we use additional notation. In particular, we suppose that there is a bijection $B:\{1,\dots,|\mathcal{A}|\}\rightarrow \mathcal{A}$ that enumerates all our ballots. This function allows us to represent profiles $A$ by vectors $v$ such that the $\ell$-th entry of $v$ states how often the ballot $B(\ell)$ is reported in $A$. When specifying the vector of a specific profile $A$, we typically write $v(A)$, but we also consider arbitrary vectors $v\in\mathbb{R}^{|\mathcal{A}|}$ which usually have the same interpretation. Finally, we define the permutation of vectors as follows: $\tau(v)_{\ell_1}=v_{\ell_2}$ for all permutations $\tau:\mathcal{C}\rightarrow\mathcal{C}$, vectors $v$, and indices $\ell_1,\ell_2\in \{1,\dots,|\mathcal{A}|\}$ such that $B(\ell_1)=\tau(B(\ell_2))$. Put differently, if there are $v_2$ ballots of type $B(\ell_2)$ in $v$, then there are $\tau(v)_{\ell_1}=v_{\ell_2}$ ballots of type $B(\ell_1)=\tau(B(\ell_2))$ in $\tau(v)$.

\subsection{The Hyperplane Argument}\label{app:hyperplane}

In this subsection, we will show the hyperplane argument sketched in the proof of \Cref{thm:Thiele} and additionally investigate some of its consequences. Note that our subsequent arguments do not rely on independence of losers, weak efficiency, or choice set convexity and thus form the basis of the proofs of both \Cref{thm:BSWAV,thm:Thiele}. However, we will focus in this section on \emph{non-imposing} ABC voting rules, which requires that for every committee $W\in\mathcal{W}^k$, there is a profile $A\in\mathcal{A}^*$ such that $f(A)=\{W\}$.

Now, let $f$ denote an anonymous, neutral, consistent, and non-imposing ABC voting rule. As first step, we aim to extend the domain of $f$ from approval profiles to $\mathbb{Q}^{|\mathcal{A}|}$. To this end, we recall that, since $f$ is anonymous, there is a function $g:\mathbb{N}^{|\mathcal{A}|}\rightarrow \mathcal{W}_k$ such that $f(A)=g(v(A))$ for all profiles $A\in\mathcal{A}^*$. We show next how to extend this function to $\mathbb{Q}^{|\mathcal{A}|}$ while preserving its desirable properties.

\begin{restatable}{lemma}{domain}\label{lem:domain}
    Let $f$ denote a non-imposing ABC voting rule that satisfies anonymity, neutrality, and consistency. There is a function $\hat g:\mathbb{Q}^{|\mathcal{A}|}\rightarrow 2^{\mathcal{W}_k}\setminus \{\emptyset\}$ that satisfies neutrality, consistency, and $\hat g(v(A))=f(A)$ for all $A\in\mathcal{A}^*$. 
\end{restatable}

\begin{proof}
    Let $f$ denote a non-imposing ABC voting rule satisfying anonymity, neutrality, and consistency. Moreover, let  $g:\mathbb{N}^{|\mathcal{A}|}\rightarrow \mathcal{W}_k$ denote a neutral and consistent function such that $g(v(A))=f(A)$ for all $A\in\mathcal{A}^*$; $f$ uniquely defines such a function since it is anonymous. We will subsequently extend the domain of $g$. For doing so, we will heavily rely on the profile $A^*$ in which every ballot is reported once. Moreover, let $v^*=v(A^*)$ and observe that $v^*_\ell=1$ for all $\ell\in \{1,\dots, |\mathcal{A}|\}$. Clearly, anonymity and neutrality require that $f(A^*)=g(v^*)=\mathcal{W}_k$ as all committees are symmetric to each other in $A^*$.\medskip

    \textbf{Step 1: Extension to $\mathbb{Z}^{|\mathcal{A}|}$}
    
    First, we define a function $\bar g:\mathbb{Z}^{|\mathcal{A}|}\rightarrow\mathcal{W}_k$ that extends $g$ to negative numbers: $\bar g(v-\ell v^*)=g(v)$ for all $v\in \mathbb{N}^{|\mathcal{A}|}$ and $\ell\in\mathbb{N}_0$. First, note that $\bar g$ is well-defined: for every two integers $\ell,\ell'\in\mathbb{N}_0$ and vectors $v, v'\in \mathbb{N}^{|\mathcal{A}|}$ such that $v - \ell v^* = v' - \ell' v^*$, it holds that $v'=v+(\ell'-\ell)v^*$ and $v=v'+(\ell-\ell')v^*$. Assuming that $\ell>\ell'$, we thus infer from consistency that $g(v)=g(v')\cap g((\ell-\ell')v^*)=g(v')$ as $g(v^*)=g((\ell-\ell')v^*)=\mathcal{W}_k$. Because $g(v)=g(v')$, we have by definition that $\bar g(v-\ell v^*)=\bar g(v'-\ell'v^*)$. Moreover, $\bar g$ is defined for all $v\in\mathbb{Z}^{|\mathcal{A}|}$ since we can always find $v'\in\mathbb{N}^{|\mathcal{A}|}$ and $\ell\in\mathbb{N}_0$ with $v=v'-\ell v^*$. Finally, note that $\bar g(v(A))=\bar g(v(A)-0v^*)=g(v(A))=f(v(A))$ for all profiles $A\in\mathcal{A}^*$.

    Next, it is easy to verify that $\bar g$ inherits neutrality and consistency from $g$. For showing the neutrality of $\bar g$, consider a vector $v\in \mathbb{Z}^{|\mathcal{A}|}$ and let $W\in \bar g(v)$. By the definition of $\bar g$, there are $v'\in\mathbb{N}^{|\mathcal{A}|}$ and $\ell\in\mathbb{N}_0$ such that $v=v'-\ell v^*$ and $\bar g(v)=g(v+\ell v^*)=g(v')$. Since $\tau(v^*)=v^*$, it is easy to see that $\tau(v)+\ell v^*=\tau(v')$ for all permutations $\tau:\mathcal{C}\rightarrow\mathcal{C}$. Hence, $\tau(W)\in \bar g(\tau(v))=g(\tau(v)+\ell v^*)=g(\tau(v'))$ due to the neutrality of $g$. 
    
    Finally, for proving that $\bar g$ is consistent, consider two vectors $v^1, v^2\in\mathbb{Z}^{|\mathcal{A}|}$. By definition of $\bar g$, there are $\bar v^1, \bar v^2\in\mathbb{N}^{|\mathcal{A}|}$ and $\ell_1,\ell_2\in\mathbb{N}_0$ such that $v^1=\bar v^1-\ell_1 v^*$, $v^2=\bar v^2-\ell_2 v^*$, $\bar g(v^1)=g(v^1+\ell_1 v^*)=g(\bar v^1)$, and $\bar g(v^2)=g(v^2+\ell_2 v^*)=g(\bar v^2)$. Clearly, this implies that $\bar g(v^1+v^2)=g(v^1+v^2+\ell_1 v^*+\ell_2 v^*)=g(\bar v^1+\bar v^2)$. Hence, if $\bar g(v^1)\cap \bar g(v^2)=g(\bar v^1)\cap g(\bar v^2)\neq \emptyset$, then $\bar g(v^1+v^2)=g(\bar v^1+\bar v^2)=g(\bar v^1)\cap g(\bar v^2)=\bar g(v^1)\cap \bar g(v^2)$ because $g$ is consistent.\medskip
    
    \textbf{Step 2: Extension to $\mathbb{Q}^{|\mathcal{A}|}$}
    
    As second step, we extend $\bar g$ to the rational numbers. For doing so, we define $\hat g(\frac{v}{\ell})=\bar g(v)$ for all $v\in\mathbb{Z}^{|\mathcal{A}|}$ and $\ell\in\mathbb{N}$. Clearly, $\hat g$ is defined for all $v\in\mathbb{Q}^{|\mathcal{A}|}$. Next, the consistency of $\bar g$ shows that $\hat g$ is well-defined: if there are $v,v'\in\mathbb{Z}^{|\mathcal{A}|}$ and $\ell, \ell'\in \mathbb{N}$ such that $\frac{v}{\ell}=\frac{v'}{\ell'}$, then $\ell' v= \ell v'$. By consistency of $\bar g$, we hence infer that $\bar g(v)=\bar g(\ell' v)= \bar g(\ell v')=\bar g(v')$, which proves that $\hat g$ is well-defined. Moreover, observe that $\hat g(v(A))=\hat g(\frac{v(A)}{1})=\bar g(v(A))=f(A)$ for all $A\in\mathcal{A}^*$.
    
    Next, it is simple to show that $\hat g$ is neutral and consistent. For proving neutrality, let $v\in\mathbb{Q}^{|\mathcal{A}|}$ be an arbitrary vector and $W\in\hat g(W)$. By definition, there are $v'\in\mathbb{Z}^{|\mathcal{A}|}$ and $\ell\in\mathbb{N}$ such that $v=\frac{v'}{\ell}$ and $\hat g(v)=\bar g(v')$. It holds for every permutation $\tau$ that $\tau(v)=\frac{\tau(v')}{\ell}$ and thus, we have that $\tau(W)\in \hat g(\tau(v))=\bar g(\tau(v'))$ because $\bar g$ is neutral. 
    
    Similarly, for showing that $\hat g$ is consistent, consider two vectors $v^1,v^2\in\mathbb{Q}^{|\mathcal{A}|}$ such that $\hat g(v^1)\cap \hat g(v^2)\neq \emptyset$. By definition of $\hat g$, there are $\hat v^1,\hat v^2\in\mathbb{Z}^{|\mathcal{A}|}$ and $\ell_1, \ell_2\in\mathbb{N}$ such that $v^1=\frac{\hat v^1}{\ell_1}$, $v^2=\frac{\hat v^2}{\ell_2}$, $\hat g(v^1)=\bar g(\hat v^1)$, and $\hat g(v^2)=\bar g(\hat v^2)$. Moreover, it holds by definition of $\hat g$ that $\hat g(v^1+v^2)=\hat g(\frac{\ell_2\hat v^1+\ell_1\hat v^2}{\ell_1\ell_2})=\bar g(\ell_2\hat v^1+\ell_1\hat v^2)$. Since $\bar g$ is consistent, we thus infer that $\hat g(v^1+v^2)=\bar g(\ell_2 \hat v^1+\ell_1\hat v^2)=\bar g(\ell_2\hat v^1)\cap \bar g(\ell_1\hat v^2)=\bar g(\hat v^1)\cap \bar g(\hat v^2)=\hat g(v^1)\cap \hat g(v^2)$. This proves that $\hat g$ is consistent.
\end{proof}

Since $\hat g$ fully describes $f$, we will next investigate this function. To this end, we introduce some additional notation. In particular, we suppose that the committees in $\mathcal{W}_k$ are arranged in an arbitrary order $W^1, \dots, W^{|\mathcal{W}_k|}$ and define $R_i^f=\{v\in\mathbb{Q}^{|\mathcal{A}|}\colon W_i\in \hat g(v)\}$ as the set of vectors for which $\hat g$ chooses $W^i$. First, we note that the sets $R_i^f$ are symmetric: if a permutation $\tau:\mathcal{C}\rightarrow\mathcal{C}$ maps $W_i$ to $W_j$ (i.e., $W_j=\tau(W_i)$), then $v\in R_i^f$ if and only if $\tau(v)\in R_j^f$ because $W_i\in \hat g(v)$ if and only if $\tau(W_i)\in \hat g(\tau(v))$. Moreover, since $\hat g$ is consistent, all $R_i^f$ are $\mathbb{Q}$-convex (i.e., if $v,v'\in R_i^f$, then $\lambda v+(1-\lambda)v'\in R_i^f$ for all $\lambda \in \mathbb{Q}\cap [0,1]$). Consequently, the closure of $R_i^f$ with respect to $\mathbb{R}^{|\mathcal{A}|}$, $\bar R_i^f$, is a convex cone. Furthermore, we observe that $g(v)=\{W^i\in\mathcal{W}_k\colon v\in R_i^f\}\subseteq \{W^i\in\mathcal{W}_k\colon v\in \bar R_i^f\}$ for all $v\in\mathbb{Q}^{|\mathcal{A}|}$. Hence, we will subsequently analyze the sets $\bar R_i^f$ and show next that these sets can be separated by hyperplanes. In the subsequent lemma, we use $vu$ for the standard scalar product between two vectors $v,u\in\mathcal{R}^{|\mathcal{A}|}$. 

\begin{lemma}\label{lem:jointHyperplane}
    Let $f$ denote a non-imposing ABC voting rule that satisfies anonymity, neutrality, and consistency. Furthermore, consider two distinct committees $W^i, W^j\in\mathcal{W}_k$. There is a non-zero vector $u^{i,j}\in\mathbb{R}^{|\mathcal{A}|}$ such that $v u^{i,j}\geq 0$ for all $v\in \bar R_i^f$ and $v u^{i,j}\leq 0$ for all $v\in \bar R_j^f$.
\end{lemma}
\begin{proof}
    Let $f$ denote a non-imposing ABC voting rule satisfying anonymity, neutrality, and consistency. Furthermore, let $\hat g$ denote the extension of $f$ to $\mathbb{Q}^{|\mathcal{A}|}$ as defined in \Cref{lem:domain}. Finally, we consider two arbitrary committees $W^i, W^j \in \mathcal{W}_k$. We will first show that the interiors of the sets $\bar R_i^f$ and $\bar R_j^f$ are disjoint, i.e., $\text{int } \bar R_i^f\cap \text{int } \bar R_j^f=\emptyset$. Assume for contradiction that this is not the case, which means that there is $v\in \text{int } \bar R_i^f\cap \text{int } \bar R_j^f\cap \mathbb{Q}^{|\mathcal{A}|}$. By the definition of $\bar R_i^f$ and $\bar R_j^f$, this means that $W^i\in \hat g(v)$, $W^j\in \hat g(v)$. On the other hand, $f$ is non-imposing, so there is a profile $A$ such that $f(A)=\{W^i\}$. By the definition of $\hat g$, $\hat g(v(A))=\{W^i\}$. Finally, since $v$ is in the interior of $\bar R_j^f\cap \mathbb{Q}^{|\mathcal{A}|}$, there must be $\lambda\in (0,1)\cap \mathbb{Q}$ such that $(1-\lambda v) + \lambda v(A)\in R_j^f$. However, by consistency, we have that $\hat g((1-\lambda) v+ \lambda v(A))= \hat g(v)\cap \hat g(v(A))=\{W^i\}$. This is a contradiction and thus, the interiors of $\bar R_i^f$ and $\bar R_j^f$ must be disjoint.

    Next, we observe that the interiors of $\bar R_i^f$ and $\bar R_j^f$ are non-empty. This follows from the observation that the sets $R_i^f$, and thus also their closures $\bar R_i^f$, are symmetric and that $\mathbb{R}^{|\mathcal{A}|}=\bigcup_{\ell\in \{1,\dots,|\mathcal{W}_k|\}} \bar R_\ell^f$. Since there is only a finite number of committees, this entails that the sets $\bar R_i^f$ have full dimension and thus have indeed non-empty interiors. Finally, we can now use the separating hyperplane theorem for convex sets to derive that there is a non-zero vector $u^{i,j}\in\mathbb{R}^{|\mathcal{A}|}$ that satisfies the conditions of the lemma.  
\end{proof}

For an easy notation, we say that a non-zero vector $u$ separates $\bar R_i^f$ from $\bar R_j^f$ if $vu\geq 0$ for all $v\in \bar R_i^f$ and $vu\leq 0$ for all $v\in \bar R_j^f$. In particular, the vectors derived in \Cref{lem:jointHyperplane} are such separating vectors. 
Moreover, if a vector $u$ separates $\bar R_i^f$ from $\bar R_j^f$, then $-u$ separates $\bar R_j^f$ from $\bar R_i^f$. 
We show next that the sets $\bar R_i^f$ are fully described by \emph{every} set of separating vectors. 

\begin{lemma}\label{lem:jointPolyhedron}
Let $f$ denote a non-imposing ABC voting rule that satisfies anonymity, neutrality, and consistency.
For all distinct $i,j\in \{1,\dots, |\mathcal{W}_k|\}$, let $u^{i,j}\in \mathbb{R}^{|\mathcal{A}|}$ denote a non-zero vector such that $u^{i,j}$ separates $\bar R_i^f$ from $\bar R_j^f$. It holds for all $i\in \{1,\dots,|\mathcal{W}_k|\}$ that $\bar R_i^f=S_i^f=\{x\in\mathbb{R}^{|\mathcal{A}|}\colon x u^{i,j}\geq 0 \text{ for all } j\in \{1,\dots,|\mathcal{W}_k|\}\setminus\{i\}\}$.
\end{lemma}
\begin{proof}
Let $f$ denote a non-imposing ABC voting rule that satisfies all given axioms, let the vectors $u^{i,j}$ be defined as in the lemma, and fix an index $i\in \{1,\dots,|\mathcal{W}_k|\}$. By definition, it holds that $v u^{i,j}\geq 0$ for all $j\in \{1,\dots,|\mathcal{W}_k|\}\setminus \{i\}$ if $v\in \bar R_i^f$, so $v\in S_i^f$. This proves that $\bar R_i^f\subseteq S_i^f$. For the other direction, note that the sets $\bar R_i^f$ are fully dimensional since they are symmetric and $\mathbb{R}^{|\mathcal{A}|}=\bigcup_{j\in \{1,\dots,|\mathcal{W}_k|\}} \bar R_j^f$. Since $\bar R_i^f\subseteq S_i^f$, we thus also have that $\text{int } S_i^f\neq \emptyset$. Now, let $v\in \text{int } S_i^f$, which means that $vu^{i,j}>0$ for all $j\in \{1,\dots,|\mathcal{W}_k|\}$, $j\neq i$. 
In turn, this implies that $v\not \in \bar R_j^f$ for all $j\in \{1,\dots,|\mathcal{W}_k|\}\setminus \{i\}$ because $v\in \bar R_j^f$ entails that $v u^{i,j}\leq 0$. 
Since $\mathbb{R}^{|\mathcal{A}|}=\bigcup_{j\in \{1,\dots,|\mathcal{W}_k|\}} \bar R_j^f$ and all $\bar R_j^f$ are closed and convex, we now infer that $v\in \text{int } \bar R_i^f$. Hence, $\text{int } S_i^f\subseteq \text{int } \bar R_i^f$, so we deduce that $S_i^f\subseteq \bar R_i^f$.
\end{proof}

As a consequence of \Cref{lem:jointPolyhedron}, it suffices to understand the separating vectors $u^{i,j}$ for characterizing $\hat g$. Hence, we now aim to derive such vectors that are additionally symmetric. For this, we start with the simple but helpful observation that the symmetry of the sets $\bar R_i^f$ entails some symmetry for the hyperplanes. 

\begin{lemma}\label{lem:jointEasySymmetry}
    Let $f$ denote a non-imposing ABC voting rule that satisfies anonymity, neutrality, and consistency. Moreover, consider committees $W^i, W^j, W^{i'}, W^{j'}\in\mathcal{W}_k$ and a permutation $\tau:\mathcal{C}\rightarrow\mathcal{C}$ such that $W^i\neq W^j$, $W^{i'} \neq W^{j'}$, $|W^i\cap W^j|=|W^{i'}\cap W^{j'}|$, $\tau(W^i\cap W^j)= W^{i'}\cap W^{j'}$, $\tau(W^i\setminus W^j)= W^{i'}\setminus W^{j'}$, and $\tau(W^j\setminus W^i)= W^{j'}\setminus W^{i'}$. If a vector $u$ separates $\bar R_{i}^f$ from $\bar R_{j}^f$, then $\tau(u)$ separates $\bar R_{i'}^f$ from $\bar R_{j'}^f$.
\end{lemma}
\begin{proof}
    Let $f$ denote an ABC voting rule satisfying all given axioms, and consider committees $W^i, W^j, W^{i'}, W^{j'}\in\mathcal{W}_k$ as defined in the lemma. Moreover,  let $u$ denote a vector that separates $\bar R_i^f$ from $\bar R_j^f$, and let $\tau:\mathcal{C}\rightarrow\mathcal{C}$ be a permutation that satisfies the conditions of the lemma. Now, consider a vector $v'\in\bar R_{i'}^f$. By the neutrality of $\hat g$, there is a vector $v\in \bar R_{i}^f$ such that $\tau(v)=v'$ because $\tau(W^i)=W^{i'}$. Since $u$ separates $\bar R_{i}^f$ and $\bar R_{j}^f$, we have $v u\geq 0$. Now, it is straightforward that $ v'\tau(u)=\tau(v)\tau(u) = v u \geq 0$ because the scalar product does not change if we permute both vectors. Hence, it holds that $v'\tau(u)\geq 0$ for all $v'\in\bar R_{i'}^f$. An analogous argument also works for vectors $v'\in\bar R_{j'}^f$ and $\tau(u)$ thus separates $\bar R_{i'}^f$ from $\bar R_{j'}^f$.  
\end{proof}

Based on \Cref{lem:jointHyperplane,lem:jointPolyhedron,lem:jointEasySymmetry}, we show next that there are highly symmetric vectors that fully specify the sets $\bar R_i^f$. 

\begin{restatable}{lemma}{hyperplanes}\label{lem:hyperplanes}
    Let $f$ denote a non-imposing ABC voting rule that satisfies anonymity, neutrality, and consistency. There are non-zero vectors $\hat u^{i,j}$ that satisfy the following conditions for all $W^i, W^j\in\mathcal{W}_k$:
    \begin{enumerate}[itemsep=0em, topsep=2pt]
        \item $\bar R_i^f=\{v\in \mathbb{R}^{|\mathcal{A}|}\colon \forall j'\in \{1,\dots, |\mathcal{W}_k|\}\setminus \{i\}\colon \hat u^{i,j'} v\geq 0\}$.
        \item $\hat u^{i,j}=-\hat u^{j,i}$.
        \item $\hat u^{i',j'}=\tau(\hat u^{i,j})$ for all permutations $\tau:\mathcal{C}\rightarrow\mathcal{C}$ with $\tau(W^i)=W^{i'}$ and $\tau(W^j)=W^{j'}$.
    \end{enumerate}
\end{restatable}\begin{proof}
    Let $f$ denote a non-imposing ABC voting rule that satisfies all given axioms. By \Cref{lem:jointHyperplane}, there are non-zero vectors $u^{i,j}$ that separate $\bar R_i^f$ from $\bar R_j^f$ for all pairs of committees $W^i, W^j\in\mathcal{W}_k$ and suppose that $u^{i,j}=-u^{j,i}$. Our main goal is to make these vectors symmetric and we will heavily rely on \Cref{lem:jointEasySymmetry} for this. To this end, we define $z=\max\limits_{W^i,W^j\in\mathcal{W}_k} |W^i\setminus W^j|$ as the maximal distance between two committees. Moreover, we fix $z+1$ committees $W^{i_0},\dots, W^{i_z}$ such that $|W^{i_0}\setminus W^{i_x}|=x$ for all $x\in \{1,\dots, z\}$.

    Next, we will derive the symmetric separating vectors $\hat u^{i,j}$. 
    For this, consider an arbitrary index $x\in \{1,\dots, z\}$ and let $u^{i_0, i_x}$ be the vector that separates $\bar R_{i_0}^f$ from $\bar R_{i_x}^f$. 
    Moreover, we define the sets $X^{i_0\setminus i_x}=W^{i_0}\setminus W^{i_x}$, $X^{i_0\cap i_x}=W^{i_0}\cap W^{i_x}$, $X^{i_x\setminus i_0}=W^{i_x}\setminus W^{i_0}$, and $\mathcal{T}=\{\tau\in \mathcal{C}^{\mathcal{C}}\colon \tau(X^{i_0\cap i_x})=X^{i_0\cap i_x},\tau(X^{i_0\setminus i_x})=X^{i_0\setminus i_x}, \tau(X^{i_x\setminus i_0})=X^{i_x\setminus i_0} \}$. 
    In particular, it holds for every $\tau\in\mathcal{T}$ that $\tau(W^{i_0})=W^{i_0}$ and $\tau(W^{i_x})=W^{i_x}$. \sloppy{Consequently, \Cref{lem:jointEasySymmetry} shows that $\tau(u^{i_0, i_x})$ also separates $\bar R_{i_0}^f$ from $\bar R_{i_x}^f$ and the same follows for $\tilde{u}^{i_0, i_x}=\sum_{\tau\in\mathcal{T}} \tau(u^{i_0,i_x})$.}
    This also means that the vector $\tilde{u}^{i_x, i_0}=-\tilde{u}^{i_0, i_x}$ separates $\bar R_{i_x}^f$ from $\bar R_{i_0}^f$. Next, let $\tau^*$ denote a permutation such that $\tau^*(X^{i_0\cap i_x})=X^{i_0\cap i_x}$, $\tau^*(X^{i_0\setminus i_x})=X^{i_x\setminus i_0}$, $\tau^*(X^{i_x\setminus i_0})=X^{i_0\setminus i_x}$, and $\tau^*(\tau^*(c))=c$ for all candidates $c\in\mathcal{C}$. 
    It is easy to verify that $\tau^*(W^{i_0})=W^{i_x}$ and $\tau^*(W^{i_x})=W^{i_0}$ and \Cref{lem:jointEasySymmetry} thus shows that $\tau^*(\tilde{u}^{i_x, i_0})$ separates $\bar R_{i_0}^f$ from $\bar R_{i_x}^f$. 
    Finally, we define the vector $\hat u^{i_0, i_x}$ by $\hat u^{i_0, i_x}=\tilde{u}^{i_0, i_x}+\tau^*(\tilde u^{i_x, i_0})$ and note that this vector separates $\bar R_{i_0}^f$ from $\bar R_{i_x}^f$. Moreover, we generalize these vectors to arbitrary committees $W^i, W^j\in\mathcal{W}_k$ as follows: we first determine $x=|W^i\setminus W^j|$ and choose a permutation $\tau$ such that $\tau(W^{i_0})=W^i$ and $\tau(W^{i_x})=W^j$. Then, $\hat u^{i,j}=\tau(\hat u^{i_0, i_x})$. This vector separates $\bar R_i^f$ from $\bar R_j^f$ by \Cref{lem:jointEasySymmetry}.  

    It remains to show that these vectors satisfy our conditions. In more detail, we first prove Claim (1) and the discuss an auxiliary claim establishing some symmetry properties of the vectors $\hat u^{i_0, i_x}$ for every $x\in \{1,\dots, z\}$. Based on this auxiliary claim, we then show Claims (2) and (3).\medskip

     \textbf{Claim (1): $\bar R_i^f=\{v\in\mathbb{R}^{|\mathcal{A}|}\colon \forall j\in \{1,\dots, |\mathcal{W}_k|\}\setminus \{i\}\colon v\hat u^{i,j}\geq 0\}$}
    
    For proving this claim, we only need to show that all vectors $\hat u^{i,j}$ are non-zero because it has already been proven that these vectors separate separate $\bar R_i^f$ from $\bar R_j^f$. Hence, once it is established that the vectors $\hat u^{i,j}$ are non-zero, the claim follows from \Cref{lem:jointPolyhedron}. Now, consider two committees $W^i, W^j\in\mathcal{W}_k$ and let $x=|W^i\setminus W^j|$.  
    Since we derive $\hat u^{i,j}$ from $\hat u^{i_0, i_x}$ by permuting the latter vector, $\hat u^{i,j}$ is non-zero if $\hat u^{i_0, i_x}$ is non-zero. Hence, it only remains to show that $\hat u^{i_0,i_x}$ is a non-zero vector. For this, we use that $\bar R_{i_0}^f=\{x\in\mathcal{R}^{|\mathcal{A}|}\colon \forall j\in \{1,\dots,|\mathcal{W}_k|\}\setminus\{i_0\}\colon x u^{i_0,\ell}\geq 0\}$ due to \Cref{lem:jointPolyhedron}, where the vectors $u^{i_0, \ell}$ denote the hyperplanes given by \Cref{lem:jointHyperplane}. Now, let $v$ denote a point in the interior of $\bar R_{i_0}^f$; such a point exists as $\bar R_{i_0}^f$ is fully dimensional and thus has a non-empty interior. Since $v\in\text{int} \bar R_{i_0}^f$, it holds that $v u^{i_0, \ell}>0$ for all $\ell\in\{1,\dots,|\mathcal{W}_k|\}\setminus \{i\}$, in particular that $vu^{i_0, i_x}>0$. Next, we note that also the vectors $\tau(u^{i_0, i_x})$ for $\tau\in \mathcal{T}$ are non-zero and separate $\bar R_{i_0}^f$ from $\bar R_{i_x}^f$. So, we can exchange $u^{i_0, i_x}$ with $\tau(u^{i_0, i_x})$ in the presentation of $R_{i_0}^f$ and infer that $v\tau(u^{i_0, i_x})>0$ since $v$ is still in the interior of $\bar R_{i_0}^f$. Hence, $v \tilde{u}^{i_0, i_x}=v\sum_{\tau\in\mathcal{T}} \tau (u^{i_0, i_x})>0$, so $\tilde{u}^{i_0, i_x}$ is a non-zero vector. This implies that $\tau^*(\tilde{u}^{i_x, i_0})$ also is a non-zero vector and we can thus also represent $\bar R_{i_0}^f$ by replacing $u^{i_0, i_x}$ with $\tau^*(\tilde{u}^{i_x, i_0})$. This implies again that $v\tau^*(\tilde{u}^{i_x, i_0})>0$ and we therefore conclude that $v\hat u^{i_0, i_x}=v (\tilde{u}^{i_0, i_x}+\tau^*(\tilde{u}^{i_x, i_0}))>0$, so $\hat u^{i_0, i_x}$ is indeed a non-zero vector.\medskip
    
    \textbf{Auxiliary Claim: Symmetry of $\hat u^{i_0, i_x}$}

    We will first prove that the vectors $\hat u^{i_0, i_x}$ are rather symmetric. In more detail, we will show that $\hat u^{i_0, i_x}_{\ell_1}=\hat u^{i_0, i_x}_{\ell_2}$ for all ballots $B(\ell_1), B(\ell_2)\in \mathcal{A}$ and all $x\in \{1,\dots, z\}$ such that $|B(\ell_1)|=|B(\ell_2)|$ and $|B(\ell_1)\cap X|=|B(\ell_2)\cap X|$ for all $X\in\{X^{i_0\cap i_x},X^{i_0\setminus i_x}, X^{i_x\setminus i_0}\}$. For this, we fix such ballots $B(\ell_1)$, $B(\ell_2)$ and an index $x\in \{1,\dots z\}$. By our assumptions, there is a bijection $\tilde\tau:\mathcal{C}\rightarrow\mathcal{C}$ such that $B(\ell_2)=\tilde\tau(B(\ell_1))$ and $\tilde\tau(X)=X$ for all $X\in \{X^{i_0\cap i_x},X^{i_0\setminus i_x}, X^{i_x\setminus i_0}\}$. By the latter insight, it follows that $\tilde\tau\circ\tau\in\mathcal{T}$ for every permutation $\tau\in\mathcal{T}$. Moreover, for distinct permutations $\tau_1,\tau_2\in\mathcal{T}$, it holds that $\tilde\tau\circ\tau_1\neq\tilde\tau\circ\tau_2$ because there is a candidate $c$ such that $\tau_1(c)\neq\tau_2(c)$. This shows that $\{\tilde\tau\circ\tau\colon \tau\in \mathcal{T}\}=\mathcal{T}$. Since $B(\ell_2)=\tilde\tau(B(\ell_1))$, we derive for every vector $u\in\mathbb{R}^{|\mathcal{A}|}$ that $\tilde\tau(u)_{\ell_2}=u_{\ell_1}$. Consequently, $\tilde u^{i_0, i_x}_{\ell_1}=\sum_{\tau\in\mathcal{T}} \tau(u^{i_0,i_x})_{\ell_1}=\sum_{\tau\in\mathcal{T}} \tilde\tau(\tau(u^{i_0, i_x}))_{\ell_2}=\sum_{\tau\in\mathcal{T}} \tau(u^{i_0, i_x})_{\ell_2}=\tilde u^{i_0, i_x}_{\ell_2}$. This proves that the vector $\tilde u^{i_0, i_x}$ satisfies our symmetry condition, and clearly the vector $\tilde u^{i_x, i_0}=-\tilde{u}^{i_0, i_x}$ satisfies this condition, too. Finally, recall that we choose $\tau^*$ such that $\tau^*(X^{i_0\cap i_x})=X^{i_0\cap i_x}$, $\tau^*(X^{i_0\setminus i_x})=X^{i_x\setminus i_0}$, and $\tau^*(X^{i_x\setminus i_0})=X^{i_0\setminus i_x}$. This means that $|B(\ell)\cap X^{i_0\cap i_x}|=|\tau^*(B(\ell))\cap X^{i_0\cap i_x}|$, $|B(\ell)\cap X^{i_0\setminus i_x}|=|\tau^*(B(\ell))\cap X^{i_x\setminus i_0}|$, and $|B(\ell)\cap X^{i_x\setminus i_0}|=|\tau^*(B(\ell))\cap X^{i_0\setminus i_x}|$. \sloppy{Since $|B(\ell_1)\cap X|=|B(\ell_2)\cap X|$ for all $X\in\{X^{i_0\cap i_x}, X^{i_0\setminus i_x}, X^{i_x\setminus i_0}\}$, the same holds for $\tau^*(B(\ell_1))$ and $\tau^*(B(\ell_2))$ and we infer that $\tau^*(\tilde u^{i_x, i_0})_{\ell_1}=\tau^*(\tilde u^{i_x, i_0})_{\ell_2}$. Finally, this means that $\hat u^{i_0, i_x}_{\ell_1}=\tilde u^{i_0, i_x}_{\ell_1} + \tau^*(\tilde u^{i_x, i_0})_{\ell_1}=\tilde u^{i_0, i_x}_{\ell_2} + \tau^*(\tilde u^{i_x, i_0})_{\ell_2}=\hat u^{i_0, i_x}_{\ell_2}$, which proves our auxiliary claim.}\medskip

    \textbf{Claim (2): $\hat u^{i,j}=-\hat u^{j,i}$}

    Consider two committees $W^i, W^j\in\mathcal{W}_k$, let $x=|W^i\setminus W^j|=|W^j\setminus W^i|$, and fix a ballot $B(\ell)$. We will show that $\hat u^{i,j}_\ell=-\hat u^{j,i}_\ell$ to prove this claim. For this, let $\tau$ denote the permutation such that $\tau(W^{i_0})=W^i$, $\tau(W^{i_x})=W^j$, and $\hat u^{i,j}=\tau(\hat u^{i_0, i_x})$. Similarly, we define $\tau'$ as the permutation with $\tau'(W^{i_0})=W^j$, $\tau'(W^{i_x})=W^i$, and $\hat u^{j,i}=\tau'(\hat u^{i_0, i_x})$. By definition, it holds that $\hat u^{i,j}_\ell=\hat u^{i_0, i_x}_{\ell_1}$ and $\hat u^{j,i}_\ell=\hat u^{i_0, i_x}_{\ell_2}$ for the indices $\ell_1$, $\ell_2$ with $B(\ell)=\tau(B(\ell_1))$ and $B(\ell)=\tau'(B(\ell_2))$. Hence, the claim follows by proving that $\hat u^{i_0, i_x}_{\ell_1}=-\hat u^{i_0, i_x}_{\ell_2}$. For this, we first observe that the condition on $\tau$ and $\tau'$ require that $\tau(X^{i_0\cap i_x})=\tau'(X^{i_0\cap i_x})=X^{i\cap j}$, $\tau(X^{i_0\setminus i_x})=X^{i\setminus j}$, $\tau(X^{i_x\setminus i_0})=X^{j\setminus i}$, $\tau'(X^{i_0\setminus i_x})=X^{j\setminus i}$, and $\tau'(X^{i_x\setminus i_0})=X^{i\setminus j}$. Hence, we infer that $|B(\ell_1)\cap X^{i_0\cap i_x}|=|B(\ell)\cap X^{i\cap j}|=|B(\ell_2)\cap X^{i_0\cap i_x}|$, $|B(\ell_1)\cap X^{i_0\setminus i_x}|=|B(\ell)\cap X^{i\setminus j}|=|B(\ell_2)\cap X^{i_x\setminus i_0}|$, and $|B(\ell_1)\cap X^{i_x\setminus i_0}|=|B(\ell)\cap X^{j\setminus i}|=|B(\ell_2)\cap X^{i_0\setminus i_x}|$. Moreover, it clearly holds that $|B(\ell_1)|=|B(\ell)|=|B(\ell_2)|$. Now, we consider again the permutation $\tau^*$ used in the definition of $\hat u^{i_0, i_x}$ and recall that $\tau^*(W^{i_0})=W^{i_x}$, $\tau^*(W^{i_x})=W^{i_0}$, and $\tau^*(\tau^*(c))=c$ for all $c\in\mathcal{C}$. Furthermore, let $\ell_3$ denote the index such that $B(\ell_3)=\tau^*(B(\ell_1))$ and note that $|B(\ell_2)\cap X|=|B(\ell_3)\cap X|$ for all $X\in\{X^{i_0\cap i_x}, X^{i_0\setminus i_x}, X^{i_x\setminus i_0}, \mathcal{C}\}$. Hence, our auxiliary claim entails that $\hat u^{i_0, i_x}_{\ell_2}=\hat u^{i_0, i_x}_{\ell_3}$. On the other hand, we have by definition that $\tau^*(\tilde u^{i_x,i_0})_{\ell_3}=\tilde u^{i_x, i_0}_{\ell_1}=-\tilde u^{i_0, i_x}_{\ell_1}$ and $\tau^*(\tilde u^{i_x,i_0})_{\ell_1}=\tilde u^{i_x, i_0}_{\ell_3}=-\tilde u^{i_0, i_x}_{\ell_3}$. It is now easy to compute that $\hat u^{i_0, i_x}_{\ell_1}=\tilde u^{i_0, i_x}_{\ell_1}+\tau^*(\tilde u^{i_x,i_0})_{\ell_1}=\tilde u^{i_0, i_x}_{\ell_1}-\tilde u^{i_0, i_x}_{\ell_3}=-(\tilde u^{i_0,i_x}_{\ell_3} - \tilde u^{i_0, i_x}_{\ell_1})=-(\tilde u^{i_0,i_x}_{\ell_3} + \tau^*(\tilde u^{i_x, i_0})_{\ell_3})=-\hat u^{i_0, i_x}_{\ell_3}$. We therefore conclude that $\hat u^{i_0, i_x}_{\ell_1}=-\hat u^{i_0, i_x}_{\ell_3}=-\hat u^{i_0, i_x}_{\ell_2}$, which proves this claim.\medskip

    \textbf{Claim (3): $\hat u^{i',j'}=\hat\tau(\hat u^{i,j})$ if $\tau(W^i)=W^{i'}$ and $\hat \tau(W^j)=W^{j'}$}

    For this claim, we consider four committees $W^i, W^j, W^{i'}, W^{j'}$ and a permutation $\hat \tau$ such that $\hat \tau(W^i)=W^{i'}$ and $\hat \tau(W^j)=W^{j'}$. Moreover, consider two ballots $B(\ell_1), B(\ell_2)\in\mathcal{A}$ such that $B(\ell_1)=\hat \tau(B(\ell_2))$. We will show that $\hat u^{i', j'}_{\ell_1}=\hat u^{i,j}_{\ell_2}$, which implies that $\hat u^{i',j'}=\hat \tau(\hat u^{i,j})$. For this, let $x=|W^i\setminus W^j|=|W^{i'}\setminus W^{j'}|$ and let $\tau$ and $\tau'$ denote the permutations such that $\tau(W^{i_0})=W^{i}$, $\tau(W^{i_x})=W^{j}$, $\tau'(W^{i_0})=W^{i'}$, $\tau'(W^{i_x})=W^{j'}$, $\hat u^{i,j}=\tau(\hat u^{i_0, i_x})$, and $\hat u^{i',j'}=\tau'(\hat u^{i_0, i_x})$. \sloppy{Clearly, there are integers $\ell_3$, $\ell_4$ such that $B(\ell_1)=\tau'(B(\ell_3))$ and $B(\ell_2)=\tau(B(\ell_4))$}. By definition, this means that $\hat u^{i',j'}_{\ell_1}=\tau'(\hat u^{i_0, i_x})_{\ell_1}=\hat u^{i_0, i_x}_{\ell_3}$ and $\hat u^{i,j}_{\ell_2}=\tau(\hat u^{i_0, i_x})_{\ell_2}=\hat u^{i_0, i_x}_{\ell_4}$. Hence, our equality follows by showing that $\hat u^{i_0, i_x}_{\ell_3}=\hat u^{i_0, i_x}_{\ell_4}$. For this, we note that $\hat \tau(X^{i\cap j})=X^{i'\cap j'}$, $\hat \tau(X^{i\setminus j})=X^{i'\setminus j'}$, and $\hat \tau(X^{j\setminus i})=X^{j'\setminus i'}$. Moreover, analogous claims hold for $\tau$ (between $W^{i_0},W^{i_x}$ and $W^{i}, W^j$) and for $\tau'$ (between $W^{i_0}, W^{i_x}$ and $W^{i'}, W^{j'}$). Thus, we can derive the following equalities since permuting sets does not change the size of their set intersection.
    \begin{align*}
        &|B(\ell_3)\!\cap\! X^{i_0\cap i_x}|\!=\!|B(\ell_1)\!\cap\! X^{i'\cap j'}|\\
        &\qquad\qquad\qquad\qquad\!=\!|B(\ell_2)\!\cap\! X^{i\cap j}|\!=\!|B(\ell_4)\!\cap\! X^{i_0\cap i_x}|\\
        &|B(\ell_3)\!\cap\! X^{i_0\setminus i_x}|\!=\!|B(\ell_1)\!\cap\! X^{i'\setminus j'}|\\
        &\qquad\qquad\qquad\qquad\!=\!|B(\ell_2)\!\cap\! X^{i\setminus j}|\!=\!|B(\ell_4)\!\cap\! X^{i_0\setminus i_x}|\\
        &|B(\ell_3)\!\cap\! X^{i_x\setminus i_0}|\!=\!|B(\ell_1)\!\cap\! X^{j'\setminus i'}|\\
        &\qquad\qquad\qquad\qquad\!=\!|B(\ell_2)\!\cap\! X^{j\setminus i}|\!=\!|B(\ell_4)\!\cap\! X^{i_x\setminus i_0}|
    \end{align*}
    
    Finally, we clearly have that $|B(\ell_3)|=|B(\ell_1)|=|B(\ell_2)|=|B(\ell_4)|$, so our auxiliary claim implies that $\hat u^{i_0,i_x}_{\ell_3}=\hat u^{i_0,i_x}_{\ell_4}$. This concludes the proof of this claim.
\end{proof}

After proving \Cref{lem:hyperplanes}, we will next investigate its consequences as we will heavily rely on this lemma. In more detail, as explained in the proof sketches of \Cref{thm:Thiele,thm:BSWAV}, we will frequently consider the hyperplanes $\hat u^{i,j}$ for committees $W^i, W^j$ with $|W^i\setminus W^j|=1$. We thus show in the next lemma that there is a compact representation of these hyperplanes.

\begin{lemma}\label{lem:s1}
    Let $f$ denote a non-imposing ABC voting rule that satisfies anonymity, neutrality, and consistency. For every ballot size $r\in\{1,\dots, m\}$, there is a functions $s_r^1(x,y)$ that satisfies the following claims for all ballots $B(\ell)$ with $|B(\ell)|=r$ and committees $W^i, W^j\in\mathcal{W}_k$ with $|W^i\setminus W^j|=1$. 
    \begin{enumerate}
        \item $\hat u^{i,j}_\ell=s^1_{r}(|B(\ell)\cap W^i|, |B(\ell)\cap W^j|)$.
        \item $s^1_r(|B(\ell)\cap W^i|, |B(\ell)\cap W^j|)=0$ if $|B(\ell)\cap W^i|=|B(\ell)\cap W^j|$.
        \item $s^1_r(|B(\ell)\cap W^i|, |B(\ell)\cap W^j|)=-s^1_r(|B(\ell)\cap W^j|, |B(\ell)\cap W^i|)$.
    \end{enumerate}
\end{lemma}
\begin{proof}
    Let $f$ denote an ABC scoring rule that satisfies all given conditions and let $\hat u^{i,j}$ denote the non-zero vectors given by \Cref{lem:hyperplanes}. Our main goal is to show that $\hat u^{i,j}_\ell=\hat u^{i',j'}_{\ell'}$ for all committees $W^i, W^j, W^{i'}, W^{j'}\in\mathcal{W}_k$ and ballots $B(\ell), B(\ell')\in\mathcal{A}$ such that $|W^i\setminus W^j|=|W^{i'}\setminus W^{j'}|=1$, $|B(\ell)|=|B(\ell')|$, $|B(\ell)\cap W^i|=|B(\ell')\cap W^{i'}|$, and $|B(\ell)\cap W^j|=|B(\ell')\cap W^{j'}|$. Clearly, this implies the existence of the functions $s^1_r$ as we can just define $s^1_r(x,y)=\hat u^{i,j}_\ell$ for arbitrary committees $W^i, W^j\in\mathcal{W}_k$ and a ballot $B(\ell)$ with $|W^i\setminus W^j|=1$, $|B(\ell)|=r$, $|B(\ell)\cap W^i|=x$, and $|B(\ell) \cap W^j|=y$. For proving our claim, we define $\{a\}=W^i\setminus W^j$, $\{b\}=W^j\setminus W^i$, $\{a'\}=W^{i'}\setminus W^{j'}$, and $\{b'\}=W^{j'}\setminus W^{i'}$. Moreover, we use a case distinction with respect to whether $|B(\ell)\cap W^i|=|B(\ell)\cap W^j|$ or not. 

    First, we suppose that $|B(\ell)\cap W^i|=|B(\ell)\cap W^j|$ and consequently also $|B(\ell')\cap W^{i'}|=|B(\ell')\cap W^{j'}|$. In this case, we claim that $\hat u^{i,j}_\ell=\hat u^{i',j'}_{\ell'}=0$ and prove this statement only for $\hat u^{i,j}_\ell$ as the argument for $\hat u^{i',j'}_{\ell'}$ is symmetric. The key insight here is that if $|B(\ell)\cap W^i|=|B(\ell)\cap W^j|$, then either $\{a,b\}\subseteq B(\ell)$ or $\{a,b\}\cap B(\ell)=\emptyset$. Now, let $\tau$ denote the permutation defined by $\tau(a)=b$, $\tau(b)=a$, and $\tau(x)=x$ for all $x\in\mathcal{C}\setminus \{a,b\}$. It is easy to see that $\tau(W^i)=W^j$, $\tau(W^j)=W^i$, and $\tau(B(\ell))=B(\ell)$. Therefore, we can use Claims (2) and (3) of \Cref{lem:hyperplanes} to infer that $-\hat u^{i,j}_{\ell}=\hat u^{j,i}_{\ell}=\tau(\hat u^{i,j})_{\ell}=\hat u^{i,j}_{\ell}$. Clearly, this is only possible if $\hat u^{i,j}_{\ell}=0$, so our claim follows.

    A second case, suppose that $|B(\ell)\cap W^i|\neq|B(\ell)\cap W^j|$. Without loss of generality, we suppose that $|B(\ell)\cap W^i|>|B(\ell)\cap W^j|$. This implies that $|B(\ell)\cap W^i|=|B(\ell')\cap W^{i'}|=|B(\ell')\cap W^{j'}|+1=|B(\ell)\cap W^j|+1$, so $a\in B(\ell)$, $b\not\in B(\ell)$, $a'\in B(\ell')$, and $b'\not\in B(\ell')$. Now, let $\tau$ denote the permutation such that $\tau(a)=a'$, $\tau(b)=b'$, $\tau(W^i\cap W^j)=W^{i'}\cap W^{j'}$, and $\tau(B(\ell))=B(\ell')$; by our assumptions such a permutation exists. Clearly, $\tau(W^i)=W^{i'}$, $\tau(W^j)=W^{j'}$, and $\tau(B(\ell))=B(\ell')$, so Claim (3) of \Cref{lem:hyperplanes} entails that $\hat u^{i',j'}_{\ell'}=\tau(\hat u^{i,j})_{\ell'}=\hat u^{i,j}_{\ell}$. This proves the desired equality. 
    
    By the insights of the last two paragraphs, it follows that there are functions $s_r^1(x,y)$ with $\hat u^{i,j}_\ell=s_{|B(\ell)|}^1(|B(\ell)\cap W^i|,|B(\ell)\cap W^j|)$ for all ballots $B(\ell)$ and committees $W^i$, $W^j$ with $|W^i\setminus W^j|=1$. Moreover, the analysis in the second paragraph immediately implies that $s^1_r(|B(\ell)\cap W^i|, |B(\ell)\cap W^j|)=0$ for all committees $W^i, W^j$ and ballots $B(\ell)$ with $|W^i\setminus W^j|=1$, $|B(\ell)|=r$, and $|B(\ell)\cap W^i|=|B(\ell)\cap W^j|$ because $\hat u^{i,j}_{\ell}=0$. Hence, our functions $s_r^1$ satisfy Claims (1) and (2). Moreover, Claim (2) implies Claim (3) in the case that $|B(\ell)\cap W^i|=|B(\ell)\cap W^j|$.

    Hence, it remains to show that $s^1_r(|B(\ell)\cap W^i|, |B(\ell)\cap W^j|)=-s^1_r(|B(\ell)\cap W^j|, B(\ell)\cap W^i|)$ if $|B(\ell)\cap W^i|\neq |B(\ell)\cap W^j|$. For this, consider again two committees $W^i, W^j\in\mathcal{W}_k$ with $W^i\setminus W^j=\{a\}$, $W^j\setminus W^i=\{b\}$. Moreover, consider a ballot $B(\ell)$ such that $a\in B(\ell)$, $b\not\in B(\ell)$ and let $\tau$ denote the permutation defined by $\tau(a)=b$, $\tau(b)=a$, and $\tau(x)=x$ for all $x\in\mathcal{C}\setminus \{a,b\}$. By Claims (2) and (3) of \Cref{lem:hyperplanes}, we have for the ballot $B(\ell')=\tau(B(\ell))$ that $-\hat u^{i,j}_{\ell'}=\hat u^{j,i}_{\ell'}=\tau(\hat u^{i,j})_{\ell'}=\hat u^{i,j}_\ell$. On the other hand, it is easy to see that $|B(\ell)\cap W^i|=|B(\ell')\cap W^j|$ and $|B(\ell)\cap W^j|=|B(\ell')\cap W^i|$. Hence, we infer that $s^1_r(|B(\ell)\cap W^i|, |B(\ell)\cap W^j|)=\hat u^{i,j}_\ell = -\hat u^{i,j}_{\ell'}=-s^1_r(|B(\ell')\cap W^i|, |B(\ell')\cap W^j|)=-s^1_r(|B(\ell)\cap W^j|, |B(\ell)\cap W^i|)$, which proves this claim.
\end{proof}

Based on the last insight, we can already fully characterize the set of ABC voting rules that are ABC scoring rules if $k=1$ or $k=m-1$. This turns out rather helpful for the proofs of \Cref{thm:Thiele} and \Cref{thm:BSWAV} because it is straightforward to adapt the proof below to show these results when $k\in \{1,m-1\}$. 

\begin{proposition}\label{prop:Bordercase}
    Assume $k=1$ or $k=m-1$. An ABC voting rule is an ABC scoring rule if and only if it satisfies anonymity, neutrality, consistency, continuity, and weak efficiency.
\end{proposition}
\begin{proof}
    It is easy to check that ABC scoring rules satisfy all given axioms. So, we focus on the converse and let $f$ denote an ABC voting rule that satisfies all given axioms for $k=1$; the case that $k=m-1$ follows from similar arguments. First, if $f$ is trivial, it is the ABC scoring rule induced by the score function $s(x,y)=0$. We hence suppose that $f$ is non-trivial. We will first show that $f$ is non-imposing. For this, we note that there is a ballot $A\in\mathcal{A}$ such that $f(A)\neq \mathcal{W}_k$ because of non-triviality and consistency. Let $c,d$ denote candidates such that $\{c\}\in f(A)$, $\{d\}\not \in f(A)$ and consider a permutation $\tau:\mathcal{C}\rightarrow\mathcal{C}$ with $\tau(c)=c$. By neutrality, $\{c\}\in f(\tau(A))$, $\{\tau(d)\}\not\in f(\tau(A))$. Next, consider the profile $A^*$ that consists of a ballot $\tau(A)$ for every permutation $\tau$ with $\tau(c)=c$. By consistency, we infer that $f(A^*)=\bigcap_{\tau:\mathcal{C}\rightarrow\mathcal{C}\colon \tau(c)=c} f(\tau(A))=\{\{c\}\}$. Neutrality implies now that $f$ is non-imposing
 
    Next, we use \Cref{lem:domain} to obtain the function $\hat g:\mathbb{Q}^{|\mathcal{A}|}\rightarrow 2^{\mathcal{W}_k}\setminus \{\emptyset\}$ and define the sets $R_i^f=\{v\in\mathbb{Q}^{|\mathcal{A}|}\colon W^i\in \hat g(v)\}$. In turn, \Cref{lem:hyperplanes} entails the existence of symmetric non-zero vectors $\hat u^{i,j}$ such that $\bar R_i^f=\{v\in \mathbb{R}^{|\mathcal{A}|}\colon \forall j\in \{1,\dots, |\mathcal{W}_k|\}\setminus \{i\}\colon \hat u^{i,j} v\geq 0\}$. Moreover, since $k=1$, it follows for all distinct committees $W^i, W^j\in\mathcal{W}_k$ that $|W^i\setminus W^j|=1$. So, \Cref{lem:s1} applies and shows that there are functions $s_r(x,y)$ such that $\hat u^{i,j}_\ell=s_{|B(\ell)|}(|W^i\cap B(\ell)|, |W^j\cap B(\ell)|)$ for all ballots $B(\ell)\in\mathcal{A}$ and committees $W^i, W^j\in\mathcal{W}_k$. Based on these functions, it is simply to infer the score function $s(x,z)$ of $f$: we define $s(0,z)=0$ and $s(1,z)=s^1_z(1,0)$. It is now easy to check that $s(|W^i\cap B(\ell)|, |B(\ell)|)-s(|W^j\cap B(\ell)|, |B(\ell)|)=s_{|B(\ell)|}^1(|B(\ell)\cap W^i|, |B(\ell)\cap W^j|)$ due to the properties of $s^1_r$ discussed in \Cref{lem:s1}. Consequently, it holds that $\hat u^{i,j} v=\sum_{\ell\in \{1,\dots,|\mathcal{A}|\}} v_\ell s^1(|W^i\cap B(\ell)|, |W^j\cap B(\ell)|, |B(\ell)|)=\sum_{\ell\in \{1,\dots,|\mathcal{A}|\}}  v_\ell (s(|W^i\cap B(\ell)|, |B(\ell)|) - s(|W^j\cap B(\ell)|, |B(\ell)|))$ for all $W^i$, $W^j\in\mathcal{W}_k$ and $v\in\mathbb{R}^{|\mathcal{A}|}$. We thus define $\hat s(v,W)=\sum_{\ell\in \{1,\dots,|\mathcal{A}|\}}  v_\ell s(|W^i\cap B(\ell)|, |B(\ell)|)$ and infer that $\bar R_i^f=\{v\in\mathbb{R}^{|\mathcal{A}|}\colon \forall j\in \{1,\dots, |\mathcal{W}_k|\}\setminus \{i\}\colon \hat u^{i,j} v\geq 0\}=\{v\in \mathbb{R}^{|\mathcal{A}|}\colon\forall j\in \{1,\dots, |\mathcal{W}_k|\}: \hat s(v,W^i)\geq \hat s(v,W^j)\}$. Hence, $f(A)=\hat g(v(A))\subseteq \{W\in\mathcal{W}_k\colon \forall W'\in\mathcal{W}_k: \hat s(A,W)\geq \hat s(A,W')\}:=f'(A)$ for all $A\in\mathcal{A}^*$. 
    
    Next, we will show that this subset relation is an equality. Suppose for this that there is a profile $A$ such that $f(A)\subsetneq f'(A)$ and let $\{d\}\in f'(A)\setminus f(A)$. We note that $f'$ is consistent and non-trivial, so an analogous argument as for $f$ shows that it is non-imposing. Thus, there is a profile $A'$ such that $f'(A')=\{\{d\}\}$. By the consistency of $f'$ and the above subset relation, we have that $f(\lambda A+A')=f'(\lambda A+A')=\{\{d\}\}$ for all $\lambda\in \mathbb{N}$. However, this contradicts the continuity of $f$, which requires that there is $\lambda\in \mathbb{N}$ such that $f(\lambda A+A')\subseteq f(A)$. So, $f$ is the ABC scoring rule induced by $s$. Finally, we show that $s$ is non-decreasing. Otherwise, there is a ballot size $y\in \{1,\dots,m-1\}$ such that $0=s(0,y)>s(1,y)$. Now, consider a single ballot $A$ of size $y$. By definition of $s$ and $f$, $f(A)= \{W\in\mathcal{W}_k\colon W\not\subseteq A\}$. However, this outcome violates weak efficiency, so $s$ needs to be non-decreasing in its first argument.
\end{proof}

\subsection{Proof of \Cref{thm:BSWAV}}\label{app:BSWAV}
We will next turn to the proof of \Cref{thm:BSWAV}: BSWAV rules are the only ABC voting rules that satisfy anonymity, neutrality, consistency, continuity, choice set convexity, and weak efficiency. Since the proof that every BSWAV rule satisfies these axioms is in the main body, we focus here on the converse direction. Unfortunately, this direction is rather involved and we thus introduce several auxiliary lemmas before proving \Cref{thm:BSWAV}. In more detail, we first construct several important auxiliary profiles in \Cref{lem:weakEfficientTriviality} to show that every non-trivial ABC voting rule $f$ that satisfies all of our axioms is non-trivial. This allows us to access the vectors $\hat i^{i,j}$ derived in \Cref{lem:hyperplanes}. By investigating the linear independence of these  vectors, we can then show that $f$ is a BSWAV rule. 

To ease the outlay of our lemmas, we introduce some additional notation. Firstly, we define $\mathcal{F}^1$ as the set of ABC voting rules that satisfy anonymity, neutrality, consistency, continuity, choice set convexity, and weak efficiency (i.e., the axioms required for \Cref{thm:BSWAV}). Secondly, we define the \emph{convex hull} of two committees $W^i, W^j\in\mathcal{W}_k$ as $[W_i, W_j]= \{W\in \mathcal{W}_k: W_i\cap W_j\subseteq W \subseteq W_i \cup W_j \}$.

We start the proof of \Cref{thm:BSWAV} by constructing profiles in which a single candidate is either guaranteed to be chosen or to be not chosen. In more detail, given a candidate $x\in \mathcal{C}$ and a ballot size $r$, we consider the profile $A^{x,r}$ which contains each ballot $A$ with $|A|=r$ and $x\in A$ once, and the profile $A^{-x,r}$ which contains each ballot $A$ with $|A|=r$ and $x\not\in A$ once. 

\begin{lemma}\label{lem:weakEfficientTriviality}
Let $f\in\mathcal{F}^1$ denote a non-imposing ABC voting rule. It holds for all candidates $x\in \mathcal{C}$ and ballot sizes $r\in \{1,\dots,m\}$ that $f(A^{x,r})=\{W\in\mathcal{W}_k\colon x\in W\}$ and $f(A^{-x,r})=\{W\in\mathcal{W}_k\colon x\not\in W\}$ if there is a ballot $A\in\mathcal{A}$ with $|A|=r$ and $f(A)\neq\mathcal{W}_k$.
\end{lemma}
\begin{proof}
Consider an ABC voting rule $f\in\mathcal{F}^1$, a ballot size $r\in\{1,\dots, m\}$, and a candidate $x\in\mathcal{C}$. Moreover, suppose that there is ballot $A$ such that $|A|=r$ and $f(A)\neq\mathcal{W}_k$. First, this implies that $r\neq m$ because otherwise $A=\mathcal{C}$ and neutrality requires that all committees are chosen. Next, by anonymity and neutrality, there are only three possible outcomes for $f(A^{x,r})$ and $f(A^{-x,r})$: for both of these profiles, either $\{W\in\mathcal{W}_k\colon x\in W\}$, $\{W\in\mathcal{W}_k\colon x\not\in W\}$, or $\mathcal{W}_k$ has to be chosen as all committees in the first two sets are symmetric to each other. Moreover, weak efficiency excludes that $f(A^{-x,r})=\{W\in\mathcal{W}_k\colon x\in W\}$ as this axiom allows us to replace $x$ with any other candidate $y\in\mathcal{C}\setminus \{x\}$. Finally, we note that anonymity and neutrality require that $f(A^{x,r}+A^{-x,r})=\mathcal{W}_k$ because the profile $A^{x,r}+A^{-x,r}$ consists of all ballots of size $r$. Hence, by consistency and our previous observations, we either have that $f(A^{-x,r})=f(A^{x,r})=\mathcal{W}_k$, or $f(A^{-x,r})=\{W\in\mathcal{W}_k\colon x\not\in W\}$ and $f(A^{x,r})=\{W\in\mathcal{W}_k\colon x\in W\}$. Indeed, for all other possible combinations, it holds that $f(A^{x,r})\cap f(A^{-x,r})\neq \emptyset$ and $f(A^{x,r})\cap f(A^{-x,r})\neq\mathcal{W}_k$, so consistency would be violated.

Now suppose for contradiction that $f(A^{-x,r})=f(A^{x,r})=\mathcal{W}_k$. If $r=1$ or $r=m-1$, this conflicts with the assumption that there is a ballot $A$ of size $r$ such that $f(A)\neq\mathcal{W}_k$. The reason for this is that either $A^{x,r}$ or $A^{-x,r}$ only consist of a single ballot and neutrality between $A$ and $A^{x,r}$ (resp. $A^{-x,r}$) then requires that not all committees of size $k$ are chosen. Hence, we assume that $1<r<m-1$. For this case, let $X^+$, $X^-$ denote two disjoint and possibly empty sets of candidates. Moreover, we define $A^{X^+, X^-}$ as the profile containing each ballot $A$ with $|A|=r$, $X^+\subseteq A$, and $X^-\cap A=\emptyset$ once. Note that $A^{X^+, X^-}$ is not the empty profile if $|X^+|\leq r$ and $|X^-|\leq m-r$. Our goal is to prove that $f(A^{X^+, X^-})=\mathcal{W}_k$ for all disjoint sets $X^+, X^-$ by an induction over $t=|X^+\cup X^-|\in \{1,\dots, \min(r, m-r)\}$. When $t=\min(r, m-r)$, then $A^{X^+,X^-}$ consists of a single ballot and thus, this insight conflicts again with neutrality and the assumption that there is a ballot $A$ of size $r$ with $f(A)\neq\mathcal{W}_k$.

Now, the induction basis $t=1$ of our claim follows from our assumptions since $f(A^{-x,r})=f(A^{x,r})=\mathcal{W}_k$, and neutrality allows us to rename $x$ to any other candidate. We therefore assume that the induction hypothesis holds up to some $t\in \{1,\dots, \min(r,m-r)-1\}$ and will prove it for $t+1$. For this, we will first show an auxiliary claim: given two disjoint sets of candidates $X^+$, $X^-$ with $|X^+\cup X^-|=t-1$ and two candidates $x,y\in \mathcal{C}\setminus (X^+\cup X^-)$,
it holds that $f(A^{X^+\cup \{x,y\}, X^-})=f(A^{X^+, X^-\cup \{x,y\}})$ and 
$f(A^{X^+\cup \{x\}, X^-\cup \{y\}})=f(A^{X^+\cup \{y\}, X^-\cup \{x\}})$. We prove here only the first claim as the second one follows analogously. The central observation for the proof is that $A^{X^+\cup \{x,y\}, X^-} + A^{X^+, X^-\cup \{x\}} + A^{X^+, X^-\cup \{y\}} = A^{X^+, X^-} +  A^{X^+, X^-\cup \{x,y\}}$. Moreover, by the induction hypothesis, we know that $f(A^{X^+, X^-\cup \{y\}})=f(A^{X^+, X^-\cup \{x\}})=f(A^{X^+, X^-})=\mathcal{W}_k$. Hence, we infer from consistency that 

\begin{align*}
    &f(A^{X^+\cup \{x,y\}, X^-})\\
    &=f(A^{X^+\cup\{x,y\}, X^-})\!\cap\! f(A^{X^+, X^-\cup \{x\}})
    \!\cap\! f(A^{X^+, X^-\cup \{y\}})\\
    &=f(A^{X^+\cup \{x,y\}, X^-}+A^{X^+, X^-\cup \{x\}}+A^{X^+, X^-\cup \{y\}})\\
    &=f(A^{X^+, X^-}+A^{X^+, X^-\cup \{x,y\}})\\
    &=f(A^{X^+, X^-})\cap f(A^{X^+, X^-\cup \{x,y\}})\\
    &=f(A^{X^+, X^-\cup \{x,y\}}).
\end{align*}

Finally, consider an arbitrary set of candidates $X=\{x_1, \dots, x_{t+1}\}$. By weak efficiency, anonymity, and neutrality, we have that $W\in f(A^{\emptyset, X})$ for all committees $W$ that minimize $|X\cap W|$. Now, consider the profile $A^{\{y\}, X\setminus \{y\}}$ for $y\in X$. First, by our auxiliary claim, we have that $f(A^{\{y\}, X\setminus \{y\}})=f(A^{z, X\setminus \{z\}})$ for all $y,z \in X$. Now, if there is a committee $W\in f(A^{\{y\}, X\setminus \{y\}})$ that minimizes $|W\cap X|$, then $f(A^{\emptyset, X})\cap f(A^{\{y\}, X\setminus \{y\}})\neq \emptyset$, which means that $f(A^{\emptyset, X})\cap f(A^{\{y\}, X\setminus \{y\}})=f(A^{\emptyset, X\setminus \{y\}})=\mathcal{W}_k$ by consistency and the induction hypothesis. 

\sloppy{Hence, suppose next that there are only ballots $W$ in $f(A^{\{y\}, X\setminus \{y\}})$ that do not minimize $|W\cap X|$.} By weak efficiency, we know for every such committee $W$ that we can replace the candidates in $z \in W\cap (X\setminus \{y\})$ with a candidate $z'\in \mathcal{C}\setminus (W\cup X)$ and the resulting committee $W'$ must still be chosen for $A^{\{y\}, X\setminus \{y\}}$. Now, if $|W\cap X|\geq 1$ for each $W\in\mathcal{W}_k$, this means that $f$ chooses a committee with minimal intersection with $X$ as we can exchange all but one candidate in $X$ with candidates from outside $X$. Since this contradicts the assumption that $f$ does not choose such a committee, we suppose next that $|W\cap X|=0$ for some committee. In this case, we can in an committee $W\in f(A^{\{y\}, X\setminus \{y\}})$ first replace all candidates but $y$ by weak efficiency. Then, we look at a second candidate $z\in X$ and use the fact that $f(A^{\{y\}, X\setminus \{y\}})=f(A^{\{z\}, X\setminus \{z\}})$ to also replace $y$. Hence, we derive again that a committee minimizing $|W\cap X|$ is chosen for $f(A^{\{y\}, X\setminus \{y\}})$. So, we have in both cases that $f(A^{\emptyset, X})\cap f(A^{\{y\}, X\setminus \{y\}})\neq \emptyset$ and consistency requires therefore that $f(A^{\emptyset, X})=f(A^{\{y\}, X\setminus \{y\}})=\mathcal{W}_k$. By applying our auxiliary claim to these two profiles, we derive analogous claims for all profiles $A^{X^+, X^-}$ with $X^+\cap X^-=\emptyset$ and $X^+\cup X^-=X$. Finally, since $X$ is chosen arbitrarily, this proves the induction step and we can thus infer that $f(A)=\mathcal{W}_k$ for each ballot $A$ of size $r$. This contradicts our assumptions, so the lemma follows. 
\end{proof}

As a consequence of \Cref{lem:weakEfficientTriviality}, every non-trivial ABC voting rule $f\in\mathcal{F}^1$ is non-imposing. Indeed, for every such voting rule $f$, there is some ballot $A$ such that $f(A)\neq\mathcal{W}_k$; otherwise, consistency requires that $f(A)=\mathcal{W}_k$ for all profiles $A\in\mathcal{A}^*$. Consequently, we can use \Cref{lem:weakEfficientTriviality} to construct a profile $A^{x,r}$ for some ballot size $r$ and candidate $x$ such that $f(A^{x,r})=\{W\in\mathcal{W}_k\colon x\in W\}$. Finally, by consistency, it is easy to infer that $f(A^W)=\{W\}$, where $A^W$ is the profile that consists of all $A^{x,r}$ with $x\in W$. Since the trivial rule is clearly the BSWAV rule induced by $\alpha_r=0$ for all $r\in \{1,\dots,m\}$, we therefore focus on non-imposing ABC voting rules for the rest of the proof. Furthermore, we will restrict our attention to committee sizes $k\in \{2,\dots, m-2\}$. The reason for this is that all ABC voting rules satisfy choice set convexity and all scoring rules are BSWAV rules if $k\in \{1,m-1\}$. Hence, \Cref{prop:Bordercase} implies \Cref{thm:BSWAV} in this case.

Since we will focus on non-imposing rules from now on, we can access the normal vectors $\hat u^{i,j}$ from \Cref{lem:hyperplanes} and the functions $s_r^1$ defined in \Cref{lem:s1}. In particular, we will next investigate these functions $s_r^1$ in more detail and show that they are actually constant. Note that in the next lemma, the set $\mathcal{Q}(k,r)=\{\max(0,k+r-m),\dots, \min(k,r)-1\}$ contains all integers $x$ such that there is a ballot $A$ of size $r$ and two committees committee $W^i, W^j\in\mathcal{W}_k$ such that $|W^i\setminus W^j|=1$, $|A\cap W^i|=x+1$, and $|A\cap W^j|=x$.

\begin{lemma}\label{lem:convexEquidistance}
    Suppose $2\leq k\leq m-2$ and let $f\in\mathcal{F}^1$ denote a non-imposing ABC voting rule. For all ballot sizes $r\in \{1,\dots, m\}$, there is constant $\alpha_r\in\mathbb{R}$ such that $s^1_r(x+1,x)=\alpha_r$ if $x\in \mathcal{Q}(k,r)$. Moreover, if there is a ballot of size $r$ such that $f$ does not choose $\mathcal{W}_k$ on it, then $\alpha_r>0$.
\end{lemma}
\begin{proof}
    Let $f\in\mathcal{F}^1$ denote a non-imposing ABC voting rule and let $s^1_r$ denote the functions derived in \Cref{lem:s1}. Moreover, we consider an arbitrary ballot size $r\in \{1,\dots, m\}$. First, if $f(A)=\mathcal{W}_k$ for all ballots $A$ of size $r$, then $s^1_r(x+1,x)=0$ for all $x\in\mathcal{Q}(k,r)$. Otherwise, there are two committees $W^i, W^j$ and a ballot $B(\ell)$ such that $|B(\ell)|=r$, $|W^i\setminus W^j|=1$, $|A\cap W^i|=|A\cap W^j|+1$ and $s^1_r(|B(\ell)\cap W^i|, |B(\ell)\cap W^j|)\neq 0$. By Claim (1) of \Cref{lem:s1}, this means that $\hat u^{i,j}_{\ell}\neq 0$, where $\hat u^{i,j}$ is one of the vectors derived in \Cref{lem:hyperplanes}. Hence, by Claim (2) of \Cref{lem:hyperplanes}, we either have that $v(A) \hat u^{i,j}<0$ or $v(A)\hat u^{j,i}<0$ for the profile $A$ that only contains ballot $B(\ell)$. We suppose subsequently that $v(A) \hat u^{i,j}<0$. This means that $v(A)\not\in \bar R_i^f$ by Claim (1) of \Cref{lem:hyperplanes} and thus $W^i\not\in \hat g(v(A))=f(A)$ because of the definition of $\bar R_i^f$. However, this contradicts that $f(A)=\mathcal{W}_k$ for all ballots of size $r$ and thus, $s^1_r(x+1,x)=0$ must hold for all $x\in\mathcal{Q}(k,r)$.
    
    Hence, we suppose next that $f$ is non-trivial on ballot size $r$. In this case, we consider two committees $W^i, W^j\in \mathcal{W}_k$ with $|W^i\cap W^j|= k-2$; such committees exist since $2\leq k\leq m-2$. 
    For a simple notation, we further define $W^i\setminus W^j=\{a_1, a_2\}$ and $W^j\setminus W^i=\{b_1,b_2\}$. 
    The main goal for our proof is to show that $s^1_r(x+1,x)=s^1_r(x+2,x+1)$ for all $x, x+1\in \mathcal{Q}(k,r)$. 
    By repeatedly applying this argument, it follows that $s^1_r(x+1,x)=s^1_r(y+1,y)$ for all $x,y\in \mathcal{Q}(k,r)$.
    
    To prove this claim, fix some index $x$ such that $x,x+1\in \mathcal{Q}(k,r)$. In particular, this means that there is a ballot $A$ of size $r$ such that $|W^i\cap A|=x+2$. Now, since $f$ is non-trivial for ballot size $r$, \Cref{lem:weakEfficientTriviality} shows that $f(A^{x,r})=\{W\in\mathcal{W}_k\colon x\in W\}$ and $f(A^{-x,r})=\{W\in\mathcal{W}_k\colon x\in W\}$ for every candidate $x\in\mathcal{C}$. Next, consider the profile $A^1$ that consists of a copy of $A^{x,r}$ for every $x\in W^i\cap W^j$ and of a copy of $A^{-x,r}$ for every $x\in \mathcal{C}\setminus (W^i\cup W^j)$. By consistency, it is easy to verify that $f(A^1)=[W^i, W^j]$. As third step, consider the profile $A^2$ which consists of the following two ballots: the first voter in $A^2$ approves $a_1$, $a_2$, $x$ candidates of $W^i\cap W^j$, and $r-x-2$ candidates of $\mathcal{C}\setminus (W^i\cup W^j)$, and the second voter has the same ballot except that he replaces $a_1$ and $a_2$ with $b_1$ and $b_2$; such a ballot exists as $x+1\in \mathcal{Q}(k,r)$. 
    
    Now, by the continuity of $f$, there is $\lambda\in \mathbb{N}$ such that $f(\lambda A^1+A^2)\subseteq [W^i,W^j]$. Based on choice set convexity, anonymity, and neutrality, we will show that this subset relation is actually an equality. \sloppy{For this, we note that for every permutation $\tau$ with $\tau(\{a_1,a_2\})=\{b_1,b_2\}$, $\tau(\{b_1,b_2\})=\{a_1,a_2\}$, and $\tau(x)=x$ for $x\in\mathcal{C}\setminus \{a_1,a_2,b_1,b_2\}$, it holds that $\tau(\lambda A^1+A^2)=\lambda A^1+A^2$ (possibly after reordering voters).} Hence, anonymity and neutrality show that if $W^i\in f(\lambda A^1+A^2)$, then $W^j\in f(\lambda A^1+A^2)$, and if $W\in  f(\lambda A^1+A^2)$ for $W\in [W^i, W^j]\setminus \{W^i,W^j\}$, then $[W^i, W^j]\setminus \{W^i,W^j\}\subseteq f(\lambda A^1+A^2)$. Now, if $W^i, W^j\in f(\lambda A^1+A^2)$, the choice set convexity immediately shows that $f(\lambda A^1+A^2)= [W^i,W^j]$. On the other hand, if $[W^i, W^j]\setminus \{W^i,W^j\}\subseteq f(\lambda A^1+A^2)$, then the committees $W=\{a_1,b_2\}\cup (W^i\cap W^j)$ and $W'=\{b_1,a_2\}\cup (W^i\cap W^j)$ are chosen. By choice set convexity, we thus infer again that $W^i, W^j\in f(\lambda A^1+A^2)$ because $W^i, W^j\in [W, W']$. Thus, we indeed have $f(\lambda A^1+A^2)=[W^i,W^j]$.

    Now, let $v^1=v(A^1)$, $v^2=v(A^2)$, and $v^*=v(\lambda A^1+A^2)$ denote the vectors corresponding to the profiles $A^1$, $A^2$, and $\lambda A^1+A^2$, respectively. Moreover, consider the committee $W^{\ell}=\{a_1, b_2\}\cup (W^i\cap W^j)$ which lies strictly between $W^i$ and $W^j$. Finally, let $\hat u^{i',j'}$ denote the hyperplanes constructed in \Cref{lem:hyperplanes} and note that $\hat u^{i,\ell}$ can be described by the the functions $s_r^1$. In particular, it holds that $v^2\hat u^{i,\ell}=s_r^1(x+2,x+1)+s_r^1(x,x+1)$ as the first voter in $A^2$ approves $x+2$ members of $W^i$ and $x+1$ members of $W^\ell$, and the second voter approves $x$ members of $W^i$ and and $x+1$ members of $W^\ell$. Our goal is hence to show that $s_r^1(x+2,x+1)+s_r^1(x,x+1)=0$ since Claim (3) in \Cref{lem:s1} then implies that $s_r^1(x+2,x+1)=s_r^1(x+1,x)$. For doing this, we note that $v^*\hat u^{i,\ell}=(\lambda v^1+v^2)\hat u^{i,\ell}$, so it is enough to show that $v^1 \hat u^{i, \ell}=0$ and $v^*\hat u^{i,\ell}=0$. For the latter, we observe that $v^*\in \bar R_i^f$ and $v^*\in \bar R_\ell^f$ since $W^i, W^\ell\in f(\lambda A^1+A^2)=\hat g(v^*)$. Since Claim (1) of \Cref{lem:hyperplanes} shows that $\bar R_{i'}^f=\{v\in\mathbb{R}^{|\mathcal{A}|}\colon \forall j'\in \{1,\dots,|\mathcal{W}_k|\}\setminus \{i'\}\colon v\hat u^{i',j'}\geq 0\}$ for all $i'$, we derive that $v^*\hat u^{i,\ell}=0$. Finally, to show that $v^1\hat u^{i,\ell}=0$, let $\tau$ denote the permutation defined by $\tau(a_2)=b_2$, $\tau(b_2)=a_2$, and $\tau(x)=x$. It can be checked that $\tau(A^{x,r})=A^{x,r}$  and  $\tau(A^{-x,r})=A^{-x,r}$ (up to renaming voters) for every candidate $x\in \mathcal{C}\setminus \{a_2,b_2\}$. Hence, it also holds that $\tau(A^1)=A^1$ (up to renaming voters). On the other hand, we have that $\tau(W^i)=W^\ell$ and $\tau(W^\ell)=W^i$. Hence, we can use Claims (2) and (3) of \Cref{lem:hyperplanes} to compute that $2v^1\hat u^{i,\ell}=v^1\hat u^{i,\ell}+\tau(v^1)\tau(\hat u^{i,\ell})=v^1\hat u^{i,\ell}+v^1\hat u^{\ell,i}=v^1\hat u^{i,\ell}-v^1\hat u^{i,\ell}=0$. Clearly, this implies that $v^1\hat u^{i,\ell}=0$, so it indeed holds that $s_r^1(x+2,x+1)=s_r^1(x+1,x)$, which shows that there are constants $\alpha_r$ such that $s_r^1(x+1,x)=\alpha_r$ for all $x\in\mathcal{Q}(k,r)$.

    Finally, we need to show that $\alpha_r>0$ if there is a ballot $A$ of size $r$ with $f(A)\neq\mathcal{W}_k$. For doing so, we consider two committees $W^i, W^j\in\mathcal{W}_k$ with $|W^i\setminus W^j|=1$. Moreover, let $A$ denote the profile that consists of $A^{x,r}$ for every $x\in W^i$. By consistency and \Cref{lem:weakEfficientTriviality}, it is easy to derive that $f(A)=\{W^i\}$. Moreover, by continuity and consistency, there is $\lambda\in\mathbb{N}$ such that $f(\lambda A+ B(\ell))=\{W^i\}$ for every ballot $B(\ell)\in\mathcal{A}$. This implies for the vector $v(A)$ that it is in the interior of $\bar R_i^f$. By \Cref{lem:hyperplanes}, we hence infer that $v(A) \hat u^{i,j}>0$. Finally, since $|W^i\setminus W^j|=1$, we can represent $\hat u^{i,j}$ by $s_r^1$. Because all ballots in $A$ have size $r$ and the symmetry properties of $s_r^1$ identified in \Cref{lem:s1}, it thus holds that $v \hat u^{i,j}=\alpha_r c_1 - \alpha_r c_2$, where $c_1$ states how many voters in $A$ approve more candidates in $W^i$ then in $W^j$ and $c_2$ counts how many candidates prefer more candidates in $W^j$ than in $W^i$. Finally, by the construction of $A$, it is easy to see that $c_1>c_2$, so $v(A)\hat u^{i,j}>0$ implies that $\alpha_r>0$.
\end{proof}

Note that \Cref{lem:convexEquidistance} has strong consequences for the vectors $\hat u^{i,j}$ constructed in \Cref{lem:hyperplanes}, in particular if we consider committees $W^i$, $W^j$, $W^{i'}, W^{j'}\in\mathcal{W}_k$ such that $W^i\setminus W^j=W^{i'}\setminus W^{j'}=\{a\}$ and $W^j\setminus W^i=W^{j'}\setminus W^{i'}=\{b\}$. For these committees, the lemma shows that $\hat u^{i,j}_\ell=s_{|B(\ell)|}(|B(\ell)\cap W^i|, |B(\ell)\cap W^j|)=s_{|B(\ell)|}(|B(\ell)\cap W^{i'}|, |B(\ell)\cap W^{j'}|)=\hat u^{i',j'}_\ell$ for every ballot $B(\ell)\in\mathcal{A}$ with $a\in B(\ell)$, $b\not\in B(\ell)$. Hence, by \Cref{lem:s1}, it thus follows that $\hat u^{i,j}=\hat u^{i',j'}$. To reflect this insight, we change from now on the notation from $\hat u^{i,j}$ to $\hat u^{a,b}$, where $a$ and $b$ are the candidates such that $W^i\setminus W^j=\{a\}$, $W^j\setminus W^i=\{b\}$.

In the next lemma, we apply this insight to derive some auxiliary claims on the profiles $A^{x,r}$ and $A^{-x,r}$.

\begin{lemma}\label{lem:uabProperties}
     Suppose $2\leq k\leq m-2$ and let $f\in\mathcal{F}^1$ denote a non-imposing ABC voting rule. The following claims hold for all distinct candidates $a,b,c\in\mathcal{C}$ and ballot sizes $r\in \{1,\dots, m\}$:
    \begin{enumerate}
        \item $\hat u^{a,b}v(A^{c,r}) = \hat u^{a,b}v(A^{-c,r})=0$
        \item $\hat u^{a,b}v(A^{b,r})=-\hat u^{b,a} v(A^{b,r})$ and $\hat u^{a,b}v(A^{-b,r})=-\hat u^{b,a} v(A^{-b,r})$
        \item $\hat u^{a,b}v(A^{a,r})>0$ and $\hat u^{a,b}v(A^{-a,r})<0$ if there is a ballot $A$ of size $r$ with $f(A)\neq\mathcal{W}_k$
        \item $\hat u^{a,b}+ \hat u^{b,c} = \hat u^ {a,c}$
    \end{enumerate}
\end{lemma}
\begin{proof}
Let $f\in\mathcal{F}^1$ denote a non-imposing ABC voting rule and fix three candidates $a,b,c$ and a ballot size $r$. Moreover, we let $X$ denote a set of $k-1\leq m-3$ candidates with $\{a,b,c\}\cap X=\emptyset$ and define the committees $W^{a}=X\cup \{a\}$, $W^{b}=X\cup \{b\}$, and $W^{c}=X\cup \{c\}$.\medskip

\textbf{Claim (1):} \sloppy{The claim follows by considering the permutation $\tau$ with $\tau(a)=b$, $\tau(b)=a$, and $\tau(x)=x$ for all $x\in\mathcal{C}$. Then, it is easy to see that $\tau(A^{c,r})=A^{c,r}$ (up to renaming voters) and $\tau(W^a)=W^b$, $\tau(W^b)=W^a$. Hence, we get that $2v(A^{c,r})\hat u^{a,b}=v(A^{c,r})\hat u^{a,b} +\tau(v(A^{c,r}))\tau(\hat u^{a,b})=v(A^{c,r})\hat u^{a,b} -v(A^{c,r})\hat u^{a,b}=0$ by using \Cref{lem:s1}. This implies that $v(A^{c,r})\hat u^{a,b}=0$ and an analogous argument works for $A^{-c,r}$.}\medskip

\textbf{Claim (2):} The claim follows immediately from Claim (2) in \Cref{lem:hyperplanes} because $\hat u^{a,b}=-\hat u^{b,a}$.\medskip

\textbf{Claim (3):} We focus on the profile $A^{a,r}$ as the claim for $A^{-a,r}$ can be shown analogously. Hence, note that for every ballot $B(\ell)$ in the profile $A^{a,r}$, either $a,b\in B(\ell)$ or $a\in B(\ell)$, $b\not\in B(\ell)$. By Claim (2) of \Cref{lem:s1}, we infer that $\hat u^{a,b}_\ell=0$ if $a,b\in B(\ell)$ and by \Cref{lem:convexEquidistance}, we infer that $u^{a,b}_\ell=\alpha_r$ if $a\in B(\ell)$, $b\not\in B(\ell)$ for some constant $\alpha_r$. Hence, $v(A^{a,r})\hat u^{a,b}=n_r \alpha_r$, where $n_r>0$ states the number of ballots $B(\ell)$ in $A^{a,r}$ with $a\in B(\ell)$, $b\not\in B(\ell)$. Finally, if there is a ballot $A$ of size $r$ with $f(A)\neq\mathcal{W}_k$, \Cref{lem:convexEquidistance} shows that $\alpha_r>0$ and the claim follows.\medskip

\textbf{Claim (4):} For this claim, we consider a ballot $B(\ell)$ and note that $\hat u^{x,y}_\ell=s^1_{|B(\ell)|}(|B(\ell)\cap W^x|, |B(\ell)\cap W^y|)$ for all distinct $x,y\in \{a,b,c\}$. The statement now follows by considering the $8$ cases enumerating whether $a\in B(\ell)$, $b\in B(\ell)$, and $c\in B(\ell)$. For instance, if $a,c\in B(\ell)$, $b\not\in B(\ell)$, then $\hat u^{a,c}_\ell=0$ (by Claim (2) of \Cref{lem:s1}) and $\hat u^{a,b}_\ell=s_{|B(\ell)|}^1(|B(\ell)\cap W^a|, |B(\ell)\cap W^b|)=-s_{|B(\ell)|}^1(|B(\ell)\cap W^b|,|B(\ell)\cap W^c|)=-\hat u^{b,c}_\ell$ (by Claim (3) of \Cref{lem:s1}). The remaining cases work similar and we leave them to the reader. 
\end{proof}

We now turn to the central part of the proof of \Cref{thm:BSWAV}. For this, consider two committees $W^{i}$, $W^{j}$ with $|W^{i}\setminus W^{j}|=t>1$ and let $\{a_1,\dots, a_t\}=W^{i}\setminus W^j$ and $\{b_1,\dots, b_t\}=W^j\setminus W^i$. Our main goal is to show that the vector $\hat u^{i,j}$ can be represented as the sum of all vectors $\hat u^{a_x, b_x}$ for $x\in \{1,\dots, t\}$ as this will allow us to represent the underlying voting rule as BSWAV rule.

To this end, we first show the linear independence of a large set of vectors $\hat u^{a,b}$. 

\begin{lemma}\label{lem:convexLinearIndependence}
    Suppose $2\leq k\leq m-2$ and let $f\in\mathcal{F}^1$ denote a non-imposing ABC voting rule. Moreover, consider $2t$ distinct candidates $a_1, b_1, \dots, a_t, b_t$. The set $U=\{\hat u^{a_1,b_1}, \hat u^{a_2,b_2}, \dots, \hat u^{a_t,b_t}\}\cup \{\hat u^{a_1,a_2}, \hat u^{a_2,a_3}, \dots, \hat u^{a_{t-1},a_t}\}$ is linearly independent.
\end{lemma}
\begin{proof}
    Let $f\in\mathcal{F}^1$ denote a non-imposing ABC voting rule. Moreover, consider $2t$ candidates $a_1, b_1,\dots, a_t, b_t$ and let $M$ be defined as in the lemma. For the proof of the lemma, we put the vectors in $U$ as rows into a matrix $M\in\mathbb{R}^{2t-1\times |\mathcal A|}$. In more detail, let row $i$ of $M$ with $0<i\leq t$ be given by $\hat u^{a_i,b_i}$ and let row $t+i$ with $0<i<t$ be given by $\hat u^{a_i, a_{i+1}}$.
    We want to show that for each $i\leq 2t -1$ there is a vector $v\in\mathbb{R}^{|\mathcal{A}|}$ such that $M v = w$ satisfies that $w_i \neq 0$ and $w_j= 0$ for all $j\neq i$. Then, the dimension of the image of $M$ is $2t-1$, which is the column rank of the matrix. Since it is a basic fact in linear algebra that the column rank equals the row rank, this means that the vectors in $U$ are linearly independent.

    For showing this claim, we first observe that there is a ballot size $r$ such that $f(A^{x,r})=\{W\in\mathcal{W}_k\colon x\in W\}$ and $f(A^{-x,r})=\{W\in\mathcal{W}_k\colon x\in W\}$ by \Cref{lem:weakEfficientTriviality} and the non-imposition of $f$. Based on the claims in \Cref{lem:uabProperties}, we define now the vectors $v^i$ that satisfy our constraints for the rows $i\in \{1,\dots,t\}$. In more detail, it suffices to consider the profile $A^{-b_i,r}$ and its corresponding vector $v^{-b_i}$ for this. Indeed, we note here that for all vectors $\hat u\in U$ but $\hat u^{a_i, b_i}$, it follows that $v^{-b_i} \hat u=0$ by Claim (1) in \Cref{lem:uabProperties}. On the other hand, Claims (2) and (3) in this lemma show that $v^{-b_i} \hat u^{a_i, b_i}\neq 0$, so $v^{-b_i}$ satisfies our requirements. 
    
    \sloppy{For the other $t-1$ rows, we consider a sightly more complicated profile: let $A^i$ denote the profile that consists of $A^{a_j, r}$ and $A^{b_j, r}$ for all $j\leq i$ and let $v^i$ denote the corresponding vector. Using Claim (1) of \Cref{lem:uabProperties}, we infer for all $\hat u^{a_j,b_j}$ with $j\leq i$ that $v^i\hat u^{a_j,b_j}=
    \sum_{\ell=1}^i (v(A^{a_\ell,r})+v(A^{b_\ell,r}))\hat u^{a_j, b_j}=(v(A^{a_j,r})+v(A^{b_j,r}))\hat u^{a_j, b_j}$. Now by the choosing $\tau$ with $\tau(a_j)=b_j$, $\tau(b_j)=a_j$, and $\tau(x)=x$ for all other candidates, we get that $(v(A^{a_j,r})+v(A^{b_j,r}))\hat u^{a_j, b_j}=v(A^{a_j,r})\hat u^{a_j, b_j}+\tau(v(A^{b_j,r}))\tau(\hat u^{a_j, b_j})=v(A^{a_j,r})\hat u^{a_j, b_j}+v(A^{a_j,r}) \hat u^{b_j, a_j}=0$, where the last equality uses Claim (2) of \Cref{lem:uabProperties}.} An analogous argument also applies for the vectors $\hat u^{a_j, a_{j+1}}$ with $j<i$. Next, by Claim (1) of \Cref{lem:uabProperties}, we infer also $v^i\hat u^{a_j,b_j}=0$ and $v^i\hat u^{a_j, a_{j+1}}$ for $j>i$. Finally, consider the vector $\hat u^{a_i, a_{i+1}}$. By Claim (1), we have that $v^1\hat u^{a_i, a_{i+1}}=\sum_{j=1}^i (v(A^{a_j,r}) + (A^{b_j,r}))\hat u^{a_i, a_{i+1}}=v(A^{a_i,r})\hat u^{a_i, a_{i+1}}$. By the same argument as in the last paragraph, this is non-zero, thus proving the lemma.
\end{proof}

As next step, we show that the linear independence observed in \Cref{lem:convexLinearIndependence} turns into a linear dependence once we add a vector $\hat u^{i,j}$ with $W^i\setminus W^j=\{a_1,\dots, a_t\}$ and $W^j\setminus W^i=\{b_1,\dots, b_t\}$

\begin{lemma}\label{lem:convexLinearDependence}
    Suppose $2\leq k\leq m-2$ and let $f\in\mathcal{F}^1$ denote a non-imposing ABC voting rule. Moreover, consider $2t$ distinct candidates $a_1, b_1, \dots, a_t, b_t$, 
    and two committees $W^{i}, W^{j}\in\mathcal{W}_k$ with $W^{i}\setminus W^{j} = \{a_1, \dots, a_t\}$ and $W^{j}\setminus W^{i} = \{b_1, \dots, b_t\}$.
    \sloppy{The set $U=\{\hat u^{a_1,b_1}, \hat u^{a_2,b_2}, \dots, \hat u^{a_t,b_t}\}\cup \{\hat u^{a_1,a_2}, \hat u^{a_2,a_3}, \dots, \hat u^{a_{t-1},a_t}\}\cup \{\hat u^{i, j}\}$ is linearly dependent.} 
\end{lemma}
\begin{proof}
    Let $f\in\mathcal{F}^1$ denote an non-imposing ABC voting rule, and consider candidates $a_1, b_1, \dots, a_t, b_t$ and committees $W^i$, $W^j$ as defined in the lemma. We assume for contradiction that the vectors in $U$ are linearly independent. Our goal is to find a vector $v^*$ such that $v^*\not\in \bar R_i^f$ for every $i\in \{1,\dots, |\mathcal{W}_k|\}$. This contradicts one of our basic insights, namely that $\bigcup_{i\in \{1,\dots, |\mathcal{W}_k|\}}R_i^f=\mathbb{Q}^{|\mathcal{A}|}$ as this requires that $\bigcup_{i\in \{1,\dots, |\mathcal{W}_k|\}} \bar R_i^f=\mathbb{R}^{|\mathcal{A}|}$. 

    To his end, we first consider the the matrix $M$ that contains the vectors $u\in U$ as rows. Since $U$ is by assumption linear independent, the image of $M$ has full dimension $\mathbb{R}^{2t}$. This means that there is a vector $v\in\mathbb{R}^{|\mathcal{A}|}$ such that $v \hat u^{a_x, b_x}=1$ for all $x\in \{1,\dots,t\}$, $v \hat u^{a_x, a_{x+1}}=0$ for all $x\in \{1,\dots, t-1\}$, and $\hat u^{i,j}=-1$. First, we note that, by repeatedly applying Claim (4) of \Cref{lem:uabProperties}, it is easy to infer that $v \hat u^{a_x, a_y}=v\sum_{\ell=x}^{y-1}\hat u^{a_\ell, a_{\ell+1}}=0$ for all $x,y\in \{1,\dots, t\}$ with $x<y$. Moreover, by the symmetry of these vectors (see \Cref{lem:hyperplanes}), the same holds for $\hat u^{a_y, a_x}$. By applying again claim (4) of \Cref{lem:uabProperties}, we thus infer that $v \hat u^{a_x, b_y}=v(\hat u^{a_x,a_y}+\hat u^{a_y,b_y})=1$ for all $x,y\in \{1,\dots,t\}$. This insight implies that for every committee $W^{i'}\in [W^{i}$, $W^{j}]$, there is another committee $W^{j'}\in[W^{i}$, $W^{j}]$ such that $v \hat u^{i',j'} < 0$. In more detail, if $W^{i'}\neq W^i$, there are candidates $a_x\not\in W^{i'}$ and $b_y\in W^{i'}$. Then, it holds for the committee $W^{j'}=(W^{i'}\setminus\{b_y\})\cup \{a_i\}$ that $v \hat u^{i',j'}=v\hat u^{b_y, a_x}=-v\hat u^{a_x, b_y}=-1$. On the other side, for $W^i$, it holds by definition of $v$ that $v \hat u^{i,j}=-1$. 

    For the second step, let $r$ denote a ballot size such that $f(A)\neq \mathcal{W}_k$ for some ballot $A$ with $|A|=r$. This means that $f(A^{x,r})=\{W\in\mathcal{W}_k\colon x\in W\}$ and $f(A^{-x,r})=\{W\in\mathcal{W}_k\colon x\not\in W\}$ by \Cref{lem:weakEfficientTriviality}. Now, consider the profile $A^1$ containing one copy $A^{x,r}$ for each $x\in  W^{i}\cap W^{j}$ and one copy of each $A^{-x,r}$ for $x\not\in  W^{i}\cup W^{j}$. By consistency, $f(A^1)= [W^i$, $W^{j}]$. For $W^{i'}, W^{j'} \in [W^{i}$, $W^{j}]$, Claim (1) of \Cref{lem:hyperplanes} shows now that $\hat u^{i,j} v(A^1) = 0$. On the other hand, we will show that for every committee $W^{i'}\notin [W^{i}$, $W^{j}]$, there is another committee $W^{j'}$ such that $v(A^1)\hat u^{i',j'}<0$. For this, let $a\in (W^i\cup W^j)\setminus W^{i'}$, $b\in W^{i'}\setminus (W^i\cup W^j)$, and define $W^{j'}$ as $W^{j'}=(W^{i'}\setminus \{b\})\cup \{a\}$. By Claims (1) to (3) of \Cref{lem:uabProperties}, it follows that $v(A^1)\hat u^{i',j'}=v(A^1)\hat u^{b,a}=(v(A^{a,r}) + v(A^{-b,r})\hat u^{b,a}=v(A^{-b,r})\hat u^{b,a} - v(A^{a,r})\hat u^{a,b}<0$. 
    
    Now, let $v^n = n v(A^1) + v$, where $n$ is large enough such that still $v^n \hat u^{i',j'} < 0$ for all 
    $W^{i'}\notin [W^{i}$, $W^{j}]$ and their corresponding $W^{j'}$. It holds that $v^n\notin \bar R^f_i$ for all $W^i\in \mathcal W_k$. Firstly, by definition of $v_n$, $v_n\not\in \bar R_{i'}^f$ for all $W^{i'}\not\in [W^{i}, W^j]$. Next, consider a committee $W^{i'}\in[W^i, W^j]$ and a corresponding $W^{j'}\in [W^{i}, W^{j}]$ with $v \hat u^{i',j'} < 0$. By construction $v^n\hat u^{i',j'}=(n v(A^1) + v) \hat u^{i',j'}<0$ since $v(A^1)\hat u^{i',j'}=0$.
    In total, $v^n\notin \bigcup_{i\in \{1,\dots,|\mathcal{W}_k|\}} \bar R^f_i = \mathbb{R}^{|\mathcal A|}$. Clearly, this is a contradiction, so the vectors in $U$ are linearly dependent.
\end{proof}

As as consequence of \Cref{lem:convexLinearIndependence} and \Cref{lem:convexLinearDependence}, there are unique real coefficients $\delta_{a_1, b_1}, \dots \delta_{a_{t-1}, a_t}$ such that $\hat u^{i,j}=\delta_{a_1, b_1} \hat u^{a_1,b_1} +\dots + \delta_{a_{t-1}, a_t} \hat u^{a_{t-1},a_t}$. In the next lemma, we will determine these coefficient and show that $\hat u^{i,j}$ can be represented as a (scaled) sum of the $\hat u^{a_x, b_x}$. 

\begin{lemma}\label{lem:convexAdditivity}
    Suppose $2\leq k\leq m-2$ and let $f\in\mathcal{F}^1$ denote a non-imposing ABC voting rule. Moreover, consider two committees $W^i, W^j\in\mathcal{W}_k$ such that $|W^i\setminus W^j|=t\geq 1$ and let $W^{i}\setminus W^{j} = \{a_1,\dots a_t\}$ and $W^{j}\setminus W^{i} = \{b_1,\dots b_t\}$. There is $\delta>0$ such that $\hat u^{j_0,j_t}=\delta \sum_{i\leq t} \hat u^{a_i,b_i}$.
\end{lemma}
\begin{proof}
    Let $f\in \mathcal{F}^1$ be a non-imposing ABC voting rule, and let $W^i, W^j, a_1, b_1, \dots a_t, b_t$ be defined as in the lemma. First, we observe that the case $t= 1$ is trivial with $\delta = 1$ and hence suppose that $t>1$.
    By \Cref{lem:convexLinearDependence}, the $2t$ vectors in $U=\{\hat u^{a_1,b_1}, \hat u^{a_2,b_2}, \dots, \hat u^{a_t,b_t}\}\cup \{\hat u^{a_1,a_2}, \hat u^{a_2,a_3}, \dots, \hat u^{a_{t-1},a_t}\}\cup \{\hat u^{i, j}\}$ are linearly dependent, whereas \Cref{lem:convexLinearIndependence} shows that the vectors in $U\setminus \{\hat u^{i,j}\}$ are linearly independent. Thus, there are unique real coefficients $\delta_{a_1, b_1}, \dots \delta_{a_{t-1}, a_t}$ such that $\hat u^{i,j}=\delta_{a_1, b_1} \hat u^{a_1,b_1} +\dots + \delta_{a_{t-1}, a_t} \hat u^{a_{t-1},a_t}$. 

    Our goal is to determine these coefficients. For this, we let $r$ denote a ballot size such that $f(A)\neq\mathcal{W}_k$ for some ballot $A$ of size $r$, and define $v^x=v(A^{x,r})$ for every profile $A^{x,r}$. Next, we proceed in three steps to show the lemma.\medskip

    \textbf{Step 1: $v^{x_1} \hat u^{x_1,y_1}=v^{x_2} \hat u^{x_2,y_2}>0$}

    First, we show that $v^{x_1} \hat u^{x_1,y_1}=v^{x_2} \hat u^{x_2,y_2}>0$ for all $x_1,x_2,y_1,y_2\in\mathcal{C}$. For this, let $\tau$ denote the permutation defined by $\tau(x_1)=x_2$, $\tau(x_2)=x_1$, $\tau(y_1)=y_2$, $\tau(y_2)=y_1$, and $\tau(z)=z$ for all $z\in\mathcal{C}\setminus\{x_1,x_2,y_1,y_2\}$. By Claim (3) in \Cref{lem:hyperplanes}, it holds that $\tau(\hat u^{x_1, y_1})=\hat u^{x_2,y_2}$, and by the symmetry of the profiles $A^{x,r}$, we infer that $\tau(v^{x_1})=v^{x_2}$. \sloppy{Hence, we can now compute that $v^{x_1}\hat u^{x_1,y_1}=\tau(v^{x_1})(\hat u^{x_1,y_1})=v^{x_2}\hat u^{x_2,y_2}$. Finally, Claim (3) in \Cref{lem:uabProperties} shows that $v^{x_2}\hat u^{x_2,y_2}>0$, thus proving the claim.}\medskip

    \textbf{Step 2: $\delta_{a_1, b_1}=\dots= \delta_{a_t, b_t}\geq 0$.}
    
    For this step, we consider two indices $x,y\in \{1,\dots, t\}$ and the vectors $v^{b_x}$ and $v^{b_y}$. Moreover, let $\tau$ denote the permutation with $\tau(b_x)=b_y$, $\tau(b_y)=b_x$, $\tau(a_x)=a_t$, $\tau(a_y)=a_x$ and $\tau(z)=z$ for all candidates $z\in\mathcal{C}\setminus \{a_x, a_y, b_x, b_y\}$. 
    By Claim (1) in \Cref{lem:uabProperties}, $v^{b_x} \hat u^{i,j}=v^{b_x} (\delta_{a_1, b_1}\hat u^{a_1, b_1}+\dots+\delta_{a_{t-1}, a_t}\hat u^{a_{t-1}, a_t})=\delta_{a_x, b_x} v^{b_x} \hat u^{a_x, b_x}$ and $v^{b_y} \hat u^{i,j}=\delta_{a_y, b_y} v^{b_y} \hat u^{a_y, b_y}$ follows from an analogous reasoning. Finally, we note that $\tau(\hat u^{i,j})=\hat u^{i,j}$ since $\tau(W^i)=W^i$ and $\tau(W^j)=W^j$. Hence, we can compute that $\delta_{a_x, b_x} v^{b_x} \hat u^{a_x, b_x}=v^{b_x} \hat u^{i,j}=\tau(v^{b_x})\tau(\hat u^{i,j})=v^{b_j} \hat u^{i,j}=\delta_{a_y, b_y} v^{b_y} \hat u^{a_y, b_y}$. Since $v^{b_x} \hat u^{a_x, b_x} = v^{b_y} \hat u^{a_y, b_y}< 0$ (by Step 1 and Claim (3) of \Cref{lem:uabProperties}, this proves that $\delta_{a_x, b_x}=\delta_{a_y, b_y}$. Moreover, since $v^{b_x}\in \bar R_j^f$ (as $W^j\in f(A^{b_x,r})$), we infer from Claim (1) of \Cref{lem:hyperplanes} $v^{b_x} \hat u^{i,j}=-v^{b,x}\hat u^{j,i}\leq 0$. Since $v^{b_x} \hat u^{a_x, b_x}<0$, this means that $\delta_{a_x,b_x}\geq 0$.\medskip

    \textbf{Step 3: $\delta_{a_1, a_2}=\dots= \delta_{a_{t-1}, a_t}= 0$.} 
   
    For this step, suppose first that $t=2$. Then, we only need to show that $\hat u^{a_1, a_2}=0$. To do so, we consider the vector $v^{a_1}$. First, we note that $v^{a_1}\hat u^{a_2, b_2}=0$ by Claim (1) of \Cref{lem:uabProperties}. Hence, $v^{a_1} \hat u^{i, j} =\delta_{a_1, b_1} v^{a_1} \hat u^{a_1, b_1}  +\delta_{a_1, a_2} v^{a_1} \hat u^{a_1, a_2}$. Now, we define $\Delta=v^{a_1} \hat u^{a_1, b_1}>0$ and note that by Step 1, $v^{a_1} \hat u^{i, j} = \Delta (\delta_{a_1, b_1}+\delta_{a_1,a_2})$. Analogously, $v^{a_2} \hat u^{i, j} =\delta_{a_2 b_2}  v^{a_2} \hat u^{a_2, b_2} + \delta_{a_1, a_2} v^{a_2} \hat u^{a_1, a_2} = \Delta(\delta_{{a_2 b_2}}-\delta_{a_1,a_2})$. To derive our claim, we consider the permutation $\tau$ with $\tau(a_1)=a_2$, $\tau(a_2)=a_1$, and $\tau(x)=x$ for all candidates. Just as in the last step, we can now infer that $\Delta (\delta_{a_1, b_1}+\delta_{a_1,a_2})=v^{a_1} \hat u^{i, j}=\tau(v^{a_1}) \tau(\hat u^{i, j})=v^{a_2} \hat u^{i, j}=\Delta (\delta_{a_2, b_2}-\delta_{a_1,a_2})$. Since $\Delta>0$ and $\delta_{a_1, b_1}=\delta_{a_2, b_2}$, this equality can only be true if $\delta_{a_1, a_2}=0$, thus proving our claim.
    
    Finally, consider the case that $t>2$ and consider an index $i\in \{2,\dots, t-1\}$. Moreover, let $\bar v^i=(v^{a_i}+v^{b_i})$. Next, we note for all $\hat u\in U\setminus \{\hat u^{i,j}, \hat u^{a_i,b_i}, \hat u^{a_{i-1}, a_i}, \hat u^{a_{i}, a_{i+1}}\}$ that $\bar v^i \hat u=0$ by Claim (1) of \Cref{lem:uabProperties}. Furthermore, $\bar v^i\hat u^{a_i,b_i}=v^{a_i}\hat u^{a_i,b_i} - v^{b_i}\hat u^{b_i, a_i}=0$ by Step 1. On the other hand, $\bar v^i\hat u^{a_i, a_{i+1}}=v^{a_i}\hat u^{a_i, a_{i+1}}=-v^{a_i} \hat u^{a_{i-1}, a_i}>0$ by the symmetry of the $\hat u^{x,y}$ and Step 1. Hence, we conclude that $\bar v^i\hat u^{i,j}=\bar v^i (\delta_{a_1, b_1} \hat u^{a_1, b_1}+\dots+\delta_{a_{t-1}, a_t} \hat u^{a_{t-1}, a_t})=\delta_{a_{i}, a_{i+1}} v^{a_i} \hat u^{a_i, a_{i+1}} + \delta_{a_{i-1}, a_{i}} v^{a_i} \hat u^{a_{i-1}, a_{i}}=\delta_{a_{i}, a_{i+1}} v^{a_i} \hat u^{a_i, a_{i+1}} - \delta_{a_{i-1}, a_{i}} v^{a_i} \hat u^{a_{i}, a_{i+1}}$. On the other hand, $\bar v^i\hat u^{i,j}=-\bar v^i \hat u^{i,j}=0$ by using the symmetry of $\bar v^i$ with respect to $a_i$ and $b_i$. By using Step 1, we thus infer now that $\delta_{a_{i-1}, a_{i}}=\delta_{a_{i}, a_{i+1}}$ for all $i\in \{2,\dots, t-1\}$. 

    Finally, we will show that all $\delta_{a_{i-1}, a_{i}}$ are $0$. For this consider the vectors $v^{a_1}$ and $v^{a_2}$. For the permutation $\tau$ which mirrors $a_1$ and $a_2$ while fixing all remaining candidates, we obtain $v^{a_1} \hat u^{i, j} = \tau(v^{a_1}) \tau(\hat u^{i, j}) =v^{a_2} \hat u^{i, j}$. On the other hand, we can derive from \Cref{lem:uabProperties} that $v^{a_1} \hat u^{i, j}=\delta_{a_1, b_1} v^{a_1} \hat u^{a_1, b_1} + \delta_{a_1, a_2} v^{a_1} \hat u^{a_1, a_2}$ and $v^{a_2} \hat u^{i, j}=\delta_{a_2, b_2} v^{a_2} \hat u^{a_2, b_2} + \delta_{a_1, a_2} v^{a_2} \hat u^{a_1, a_2} + \delta_{a_2, a_3} v^{a_2} \hat u^{a_2, a_3}$. By our previous analysis, $\delta_{a_1, b_1} v^{a_1} \hat u^{a_1, b_1}=\delta_{a_2, b_2} v^{a_2} \hat u^{a_2, b_2}$ and $\delta_{a_1, a_2} v^{a_1} \hat u^{a_1, a_2}=\delta_{a_2, a_3} v^{a_2} \hat u^{a_2, a_3}$, so our equations imply that $\delta_{a_1, a_2} v^{a_2} \hat u^{a_1, a_2}=0$. Since $v^{a_2} \hat u^{a_1, a_2}\neq 0$, this means that $\delta_{a_1,a_2}=0$ and thus, all these $\delta$'s are $0$. Finally, since $\hat u^{i,j}$ is a non-zero vector and $\delta_{a_x,b_x}=\delta_{a_y,b_y}\geq 0$ for all $x,y\in \{1,\dots,t\}$, this inequality must be strict and the lemma follows by choosing $\delta= \delta_{a_1, b_1}$.
\end{proof}

Finally, we are now ready to prove \Cref{thm:BSWAV}.

\BSWAV*
\begin{proof}
    First, the direction from left to right has been shown in the main body. Moreover, if $k\in \{1,m-1\}$, then the set of BSWAV rules is equal to the set of ABC scoring rules and choice set convexity becomes trivial. Hence, the theorem follows from \Cref{prop:Bordercase}. Next, assume that $f\in\mathcal{F}^1$ for $k\in \{2,\dots, m-2\}$. If $f$ is trivial, it is clearly the BSWAV rule induced by the weights $\alpha_x=0$ for all $x$. On the other hand, if $f$ is non-trivial, \Cref{lem:weakEfficientTriviality} holds and consistency therefore entails that $f$ is non-imposing. Hence, we can access all our auxiliary lemmas now. In particular, our goal is to show that $f$ is the BSWAV rule described by the weights $\alpha_r$ constructed in \Cref{lem:convexEquidistance}. For doing so, we define the score function $s(|B(\ell)\cap W^i|, |B(\ell)|)=\alpha_{|B(\ell)|} |B(\ell)\cap W^i|$ and extend it to vectors $v\in \mathbb{R}^{|\mathcal{A}|}$ by $\hat s(v, W^i)=\sum_{1,\dots,|\mathcal{A}|\}} v_\ell s(|B(\ell)\cap W^i|, |B(\ell)|)$. Departing from here, our proof proceeds in three steps. First, we show that there is for all committees $W^i$, $W^j$ a constant $\delta>0$ such that $\delta \hat u^{i,j}_\ell=s(|B(\ell)\cap W^i|, |B(\ell)|)-s(|B(\ell)\cap W^j|, |B(\ell)|)$ for all ballots $B(\ell)$. Based on this insight, we show in the second step that $f(A)\subseteq f'(A):=\{W^i\in\mathcal{W}_k\colon \forall W^j\in\mathcal{W}_k\colon \hat s(v(A), W^i)\geq \hat s(v(A), W^j)\}$. Finally, we turn this subset relation in an equality in the last step and prove that $f$ is a BSWAV rule.\medskip

    \textbf{Step 1: There is $\delta>0$ such that $\delta \hat u^{i,j}_\ell=s(|B(\ell)\cap W^i|, |B(\ell)|)-s(|B(\ell)\cap W^j|, |B(\ell)|)$ for all $B(\ell)$.} 

    \sloppy{For this step, consider two arbitrary committees $W^i,W^j\in\mathcal{W}_k$ and let $B(\ell)\in \mathcal{A}$ denote a ballot. Moreover, let $r=|B(\ell)|$ denote size of $B(\ell)$ and define $W^i\setminus W^j=\{a_1,\dots, a_t\}$,  $W^j\setminus W^i=\{b_1, \dots, b_t\}$. By the definition of $s$, we have that $s(|B(\ell)\cap W^i|, |B(\ell)|)-s(|B(\ell)\cap W^j|, |B(\ell)|)=\alpha_r (|B(\ell)\cap W^i|-|B(\ell)\cap W^j|)$. 
    On the other hand, we know by \Cref{lem:convexAdditivity} that $\hat u^{i,j}=\delta' \sum_{x=1}^t \hat u^{a_x, b_x}$ for some $\delta'>0$. Hence, this step follows by showing that $\sum_{x=1}^t \hat u^{a_x, b_x}_\ell=\alpha_r (|B(\ell)\cap W^i|-|B(\ell)\cap W^j|)$.}
    
    For doing so, we partition the the indices $I=\{1,\dots, t\}$ into four sets: $I_1=\{x\in I\colon a_x, b_x\in B(\ell)\}$, $I_2= \{x\in I: a_x,b_x \notin B(\ell)\}$, $I_3= \{x\in I: a_x\in B(\ell), b_x\notin B(\ell)\}$, and 
    $I_4= \{x\in I: a_x\notin B(\ell), b_x\in B(\ell)\}$. Now, by \Cref{lem:s1}, we know that $\hat u^{a_x, b_x}_\ell=0$  for all $x\in I_1\cup I_2$ as $\hat u^{a_x,b_x}_\ell=\hat u^{i',j'}_\ell=s_r^1(|B(\ell)\cap W^{i'}|, |B(\ell)\cap W^{j'}|)=0$ for two arbitrary committees $W^{i'}, W^{j'}\in\mathcal{W}_k$ with $W^{i'}\setminus W^{j'}=\{a_x\}$ and $W^{j'}\setminus W^{i'}=\{b_x\}$. Moreover, a similar reasoning and \Cref{lem:convexEquidistance} show that $\hat u^{a_x, b_x}_\ell=\alpha_r$ for all $x\in I_3$ and $\hat u^{a_x, b_x}_\ell=-\alpha_r$ for all $x\in I_4$. We thus have that $\sum_{x=1}^t \hat u^{a_x, b_x}_\ell=\alpha_r|I_3|-\alpha_r|I_4|=\alpha_r(|B(\ell)\cap \{a_1, \dots, a_t\}|- |B(\ell)\cap \{b_1, \dots, b_t\}|)=\alpha_r(|B(\ell)\cap W^i|-|B(\ell)\cap W^j|)$, thus proving our claim.\medskip

    \textbf{Step 2: $f(A)\subseteq f'(A)$ for all $A\in\mathcal{A}^*$.}

    For proving this claim, we recall the function $\hat g$ and the sets $\bar R_i^f$ defined in and after \Cref{lem:domain}. In particular, by the definition of these objects, we have that $f(A)=\hat g(v(A))=\{W^i\in\mathcal{W}_k\colon v(A)\in R_i^f\}\subseteq \{W^i\in\mathcal{W}_k\colon v(A)\in\bar \bar R_i^f\}$. Hence, the claim follows by showing that $f'(A)=\{W^i\in\mathcal{W}_k\colon \forall W^j\in\mathcal{W}_k\colon \hat s(v(A), W^i)\geq \hat s(v(A), W^j)\}=\{W^i\in\mathcal{W}_k\colon v(A)\in\bar R_i^f\}$. For doing this, we note that $\bar R_i^f$ can be represented as $\bar R_i^f=\{v\in \mathcal{W}_k\colon \forall j\in \{1,\dots, \mathcal{W}_\setminus \{i\}\colon v \hat u^{i,j}\geq 0\}$ (Claim (1) of \Cref{lem:hyperplanes}). Hence, to show our equivalence, we need to prove that $\hat s(v(A), W^i)\geq \hat s(v(A), W^j)$ if and only if $v(A) \hat u^{i,j}\geq 0$ for every profile $A$ and committees $W^i, W^j$. 

    \sloppy{To do so, consider an arbitrary profile $A\in\mathcal{A}^*$ and let $W^i$, $W^j$ denote two arbitrary committees. By the last step, we have a $\delta>0$ such that $\delta \hat u^{i,j}_\ell=s(|B(\ell)\cap W^i|, |B(\ell)|)-s(|B(\ell)\cap W^j|, |B(\ell)|)$ for all ballots $B(\ell)$. Hence, it is easy to see that $\delta v(A)\hat u^{i,j}=\sum_{\ell\in \{1,\dots, |\mathcal{A}|\}} v(A)_\ell (s(|B(\ell)\cap W^i|, |B(\ell)|)-s(|B(\ell)\cap W^j|, |B(\ell)|))=\hat s(v(A), W^i)-\hat s(v(A), W^j)$, which clearly implies our claim.}\medskip

    \textbf{Step 3: $f(A)$ is a BSWAV rule.}

    For this step, we show that $f(A)=f'(A)$ and that $f'(A)$ is a BSWAV rule. For the latter point, we only need to observe that all $\alpha_r$ are non-negative. Assume for contradiction that this is not the case. Then, there is a ballot size $r$ such that $\alpha_r<0$ and $f'(A)=\{W\in\mathcal{W}_k\colon \forall W'\in\mathcal{W}_k\colon |A\cap W|\leq |A\cap W'|\}$ for an arbitrary ballot $A$ with $|A|=r$. Since $f$ chooses a subset of $f'(A)$, $f$ violates weak efficiency as we cannot decrease the number of unapproved candidates in any committee. This contradicts our assumptions, so $\alpha_r\geq 0$ and $f'$ is a BSWAV rule by definition. 

    Next, to show that $f(A)=f'(A)$, we assume for contradiction that there is a profile $A'$ for which this is not the case. This means that there is a committee $W\in f'(A)\setminus f(A)$. Moreover, since $f'$ is a BSWAV rule, it satisfies all preconditions of \Cref{lem:weakEfficientTriviality}. So, we can infer analogously as for $f$ that $f'$ is non-imposing and consistent. Now, let $A'$ denote a profile with $f'(A')=\{W\}$. By consistency of $f'$ and the subset relation of the last step, we have that $f(\lambda A+A')=f'(\lambda A+A')=\{W\}$ for all $\lambda \in\mathbb{N}$. However, this contradicts the continuity of $f$ and thus our assumption that $f'(A)\neq f(A')$ is wrong. Hence, $f$ is indeed the BSWAV rule induced by $\alpha_r$.
\end{proof}

\subsection{Proof of \Cref{thm:Thiele}}\label{app:Thiele}

In this section, we will prove our characterization of Thiele rule (\Cref{thm:Thiele}). We focus here on showing that every rule that satisfies anonymity, neutrality, consistency, continuity, and independence of losers is indeed a Thiele; the other direction can be found in the main body. To prove this claim, we essentially follow the same steps as for the proof of \Cref{thm:BSWAV} and a detailed proof sketch is given in the main body. Moreover, for a short notation, we define $\mathcal{F}^2$ as the set of all ABC voting rules that satisfy anonymity, neutrality, consistency, continuity, and independence of losers (i.e., the axioms of \Cref{thm:Thiele}).

For the first step in our analysis, we consider again the profiles $A^{x,r}$ and $A^{-x,r}$ in which all ballots $A$ of size $r$ with $x\in A$ and $x\not\in A$, respectively, are reported once. 

\begin{lemma}\label{lem:iolProfiles}
    Let $f\in\mathcal{F}^2$ denote a non-trivial ABC voting rule. There is a ballot size $r$ such that $f(A^{x,r})=\{W\in\mathcal{W}_k\colon x\in W\}$ and $f(A^{-x,r})=\{W\in\mathcal{W}_k\colon x\not\in W\}$ for all $x\in\mathcal{C}$.
\end{lemma}
\begin{proof}
    Let $f\in\mathcal{F}^2$ denote a non-imposing ABC voting rule. To show this lemma, we will find a ballot size $r$ such that $f(A^{-x,r})=\{W\in\mathcal{W}_k\colon x\not\in W\}$ for all $x\in\mathcal{C}$. This implies that $f(A^{x,r})=\{W\in\mathcal{W}_k\colon x\in W\}$ because of the following reasoning. First, anonymity and neutrality entail that $f(A^{x,r}+A^{-x,r})=\mathcal{W}_k$ as the profile $A^{x,r}+A^{-x,r}$ contains all ballots of size $r$. Hence, consistency requires that either $f(A^{x,r})=f(A^{-x,r})=\mathcal{W}_k$ or $f(A^{x,r})\cap f(A^{-x,r})=\emptyset$. In particular, if $f(A^{-x,r})=\{W\in\mathcal{W}_k\colon x\not\in W\}$, then $f(A^{x,r})\subseteq \{W\in\mathcal{W}_k\colon x\in W\}$. Finally, by applying again anonymity and neutrality to $A^{x,r}$, it is easy to see that this subset relation must be equal. 
    
   Now, for finding the index $r$ such that $f(A^{-x,r})=\{W\in\mathcal{W}_k\colon x\not\in W\}$, we first investigate the choice of $f$ for single ballots. By neutrality, it holds that $f(\mathcal{C})=\mathcal{W}_k$ for the ballot $\mathcal{C}$ in which all candidates are approved. Next, consider an arbitrary candidate $x\in\mathcal{C}$. By independence of losers, we know that $\{W\in\mathcal{W}_k\colon x\not\in W\}\subseteq f(\mathcal{C}\setminus \{x\})$. Invoking again neutrality, we derive that there are only two possible outcomes for the ballot $\mathcal{C}\setminus \{x\}$: either $f(\mathcal{C}\setminus \{x\})=\{W\in\mathcal{W}_k\colon x\not\in W\}$ or $f(\mathcal{C}\setminus \{x\})=\mathcal{W}_k$. If the former is the case, the profile only containing the ballot $\mathcal{C}\setminus \{x\}$ constitutes $A^{-x,r}$ for $r=m-1$ and we are done.

    Hence, suppose that $f(\mathcal{C}\setminus \{x\})=\mathcal{W}_k$. By neutrality, the same holds for every ballot of size $m-1$. In this case, we can simply repeat deleting candidates from $\mathcal{C}\setminus \{x\}$ until we arrive at a ballot $\mathcal{C}\setminus X$ such that $f(\mathcal{C}\setminus X)\neq\mathcal{W}_k$; such a ballot must exist as consistency otherwise implies that $f$ is trivial. Moreover, we suppose that the set of deleted candidates $X$ is minimal. By neutrality, this means that $f(\mathcal{C}\setminus Y)=\mathcal{W}_k$ for every set of candidates $Y$ with $|Y|<|X|$. Now, let $y$ denote an arbitrary candidate in $X$ and define $X^y=X\setminus \{y\}$. By our previous insight, $f(\mathcal{C}\setminus X^y)=\mathcal{W}_k$ and independence of losers then shows that $\{W\in\mathcal{W}_k\colon y\not\in W\}\subseteq f(\mathcal{C}\setminus X)$. Since $y$ is chosen arbitrarily, we can apply this argument for every $y\in X$ and derive that $\{W\in\mathcal{W}_k\colon X\not\subseteq W\}\subseteq f(\mathcal{C}\setminus X)$. Moreover, since $f(\mathcal{C}\setminus X)\neq\mathcal{W}_k$ and $f$ is neutral, we infer that this subset relation must actually be an equality. Also, since $f(\mathcal{C}\setminus X)\neq\mathcal{W}_k$, we get that $|X|\leq k$ because otherwise the set $\{W\in\mathcal{W}_k\colon X\not\subseteq W\}$ contains all committees. 
    
    Now, fix a candidate $x\in\mathcal{C}$ and consider $A^{-x,r}$ for $r=m-|X|$. This profile contains each ballot $C\setminus Y$ with $|Y|=|X|$ and $x\in Y$ once. By neutrality, we have for each of these ballots that $f(\mathcal{C}\setminus Y)=\{W\in\mathcal{W}_k\colon Y\not\subseteq W\}$. In particular, since $x\in Y$ for all considered ballots, it holds that $\{W\in\mathcal{W}_k\colon x\not\in W\}\subseteq f(\mathcal{C}\setminus Y)$ for all ballots in $A^{-x,r}$. Hence, consistency applies for $A^{-x,r}$ and shows that $f(A^{-x,r})=\bigcap_{Y\subseteq \mathcal{C}\colon |Y|=|X|\land x\in Y} f(\mathcal{C}\setminus Y)$. This means that $W\not\in f(A^{-x,r})$ for all committees $W\in\mathcal{W}_k$ with $x\in W$ as there is a set $Y$ with $x\in Y$ and $|Y|=|X|$ such that $Y\subseteq W$. Conversely, it immediately follows from consistency that $W\in f(A^{-x,r})$ for all committees $W\in\mathcal{W}_k$ with $x\not\in W$. Hence, we have that $f(A^{-x,r})=\{W\in\mathcal{W}_k\colon x\not \in W\}$, which concludes the proof of the lemma.
\end{proof}

Since the ballot size $r$ for the profiles $A^{x,r}$ and $A^{-x,r}$ plays no role in our subsequent analysis, we will omit it and mean by $A^x$ and $A^{-x}$ the profiles $A^{x,r}$ and $A^{-x,r}$ for the ballot size $r$ given by \Cref{lem:iolProfiles}.

Moreover, based on \Cref{lem:iolProfiles}, it follows straightforwardly that every non-trivial ABC voting rule $f\in\mathcal{F}^2$ is non-imposing. For this, it suffices to consider the profile $A^W$ that consists of a copy of $A^{x}$ for every candidate $x\in W$ as consistency then ensures that $f(A^W)=\{W\}$. Since the trivial rule is clearly the Thiele rule defined by $s(x)=0$ for all $x$, we therefore focus from now on non-imposing rules. This allows us to use the vectors $\hat u^{i,j}$ constructed in \Cref{lem:hyperplanes} and the functions $s_r^1$ constructed in \Cref{lem:s1}. As the next step, we use independence of losers to remove the dependence on the ballot size of the functions $s_r^1$. 

\begin{lemma}\label{lem:IoLCommitteeSymmetry}
    Let $f\in\mathcal{F}^2$ denote a non-imposing ABC voting rule. There is a function $s^1(x,y)$ such that $s^1(|W^i\cap B(\ell)|, |W^j\cap B(\ell)|)=\hat u^{i,j}_\ell$ for all committees $W^i, W^j\in\mathcal{W}_k$ and ballots $B(\ell)\in\mathcal{A}$ with $|W^i\setminus W^j|=1$. 
\end{lemma}
\begin{proof}
    Let $f\in\mathcal{F}^2$ denote a non-imposing ABC voting rule and consider two committees $W^i, W^j\in\mathcal{W}_k$ with $|W^i\setminus W^j|=1$. 
    Now, by \Cref{lem:s1}, there are functions $s^1_r$ such that $\hat u^{i,j}_\ell=s_{|B(\ell)|}(|W^i\cap B(\ell)|, |W^j\cap B(\ell)|)$ for all $B(\ell)\in\mathcal{A}$. For proving this lemma, it thus suffices to show that $s_{|B(\ell)|}(|W^i\cap B(\ell)|, |W^j\cap B(\ell)|)=s_{|B(\ell')|}(|W^i\cap B(\ell')|, |W^j\cap B(\ell')|)$ for any two ballots $B(\ell), B(\ell')$ with $|W^i\cap B(\ell)|=|W^i\cap B(\ell')|$ and $|W^j\cap B(\ell)|=|W^j\cap B(\ell')|$. If $|B(\ell)|=|B(\ell')|$, this follows from the definition of the functions $s^1_r$. Moreover, if $|W^i\cap B(\ell)|=|W^i\cap B(\ell')|=|W^j\cap B(\ell)|=|W^j\cap B(\ell')|$, then Claim (2) of \Cref{lem:s1} shows that $s_{|B(\ell)|}(|W^i\cap B(\ell)|, |W^j\cap B(\ell)|)=s_{|B(\ell')|}(|W^i\cap B(\ell')|, |W^j\cap B(\ell')|)=0$. 
    
    Hence, we suppose that $|B(\ell)|\neq|B(\ell')|$ and $|W^i\cap B(\ell)|\neq |W^j\cap B(\ell)|$. Without loss of generality, we make this more precise by assuming that $|B(\ell)|>|B(\ell')|$ and $|W^i\cap B(\ell)|=|W^j\cap B(\ell)|+1$. In particular, the latter observation means that $a\in B(\ell)\cap B(\ell')$ and $b\not\in B(\ell)\cup B(\ell')$ for the candidates $\{a\}=W^i\setminus W^j$ and $\{b\}=W^j\setminus W^i$. By this insight, it it easy to see that there is a permutation $\tau$ such that $B(\ell'')=\tau(B(\ell'))\subseteq B(\ell)$, $|W^i\cap B(\ell)|=|W^i\cap B(\ell'')|$ and $|W^j\cap B(\ell)|=|W^j\cap B(\ell'')|$. Moreover, it holds that $s_{|B(\ell')|}(|W^i\cap B(\ell')|, |W^j\cap B(\ell')|)=s_{|B(\ell'')|}(|W^i\cap B(\ell'')|, |W^j\cap B(\ell'')|)$, so it suffices to show that $s_{|B(\ell'')|}(|W^i\cap B(\ell'')|, |W^j\cap B(\ell'')|)=s_{|B(\ell)|}(|W^i\cap B(\ell)|, |W^j\cap B(\ell)|)$. 

    For this, consider the profile $A$ in which all ballots are reported once. Clearly, $f(A)=\mathcal{W}_k$ by anonymity and neutrality. Next, let $A'$ denote the profile derived from $A$ by replacing the ballot $B(\ell)$ with $B(\ell'')$. Since $B(\ell'')\subseteq B(\ell)$, $|W^i\cap B(\ell)|=|W^i\cap B(\ell'')|$, and $|W^j\cap B(\ell)|=|W^j\cap B(\ell'')|$, this means that we only disapprove candidates $x\in\mathcal{C}\setminus (W^i\cup W^j)$. So, independence of losers implies that $W^i, W^j\in f(A')$. Now, consider the vector $\hat u^{i,j}$ given by \Cref{lem:hyperplanes}. By Claim (1) of this lemma, we have that $v(A)\hat u^{i,j}=v(A')\hat u^{i,j}=0$ because $v(A), v(A')\in \bar R_i^f$ and $v(A), v(A')\in \bar R_j^f$. Hence, $v(A)\hat u^{i,j}- v(A')\hat u^{i,j}=\hat u^{i,j}_\ell-\hat u^{i,j}_{\ell''}=0$. This implies that $s_{|B(\ell)|}(|W^i\cap B(\ell)|, |W^j\cap B(\ell)|)=\hat u^{i,j}_\ell=\hat u^{i,j}_{\ell''}=s_{|B(\ell'')|}(|W^i\cap B(\ell'')|, |W^j\cap B(\ell'')|)$, thus proving the lemma. 
\end{proof}

We note that the function $s^1$ clearly inherits the symmetry properties of the functions $s_r^1$ discussed in Claims (2) and (3) of \Cref{lem:s1}. 

Now, it follows essentially from \Cref{lem:iolProfiles,lem:IoLCommitteeSymmetry} as well as the proof of \Cref{prop:Bordercase} that all rules in $\mathcal{F}^2$ are Thiele rule if $k=1$ or $k=m-1$. We thus focus on the case that $2\leq k\leq m-2$. To show this case, we need to relate the vectors $\hat u^{i,j}$ and $\hat u^{i',j'}$ to each other for committees $W^i, W^j, W^{i'}, W^{j'}$ with $|W^i\setminus W^j|\neq |W^{i'}\setminus W^{j'}|$. For doing so, we consider a sequence of sequence of committees $W^{j_0}, \dots, W^{j_t}$ such that $|W^{j_0}\setminus W^{j_t}|=t$ and $|W^{j_{x-1}}\setminus W^{j_x}|=1$ for every $x\in \{1,\dots,t\}$. Our goal is then to show that  $\hat u^{j_{0},j_t}$ is the (scaled) sum over the vectors $\hat u^{j_{x-1}, j_x}$. To this end, we proceed analogously as in \Cref{app:BSWAV} and first show that the vectors $\{\hat u^{j_0, j_1},\dots, \hat u^{j_{t-1}, j_t}\}$ are linearly independent.

\begin{lemma}\label{lem:iolLinearIndependence}
    Suppose $2\leq k\leq m-2$ and let $f\in\mathcal{F}^2$ denote a non-imposing ABC voting rule. Moreover, consider a sequence of committees $W^{j_0},\dots, W^{j_t}$ for $t\geq 2$ such that $|W^{j_0}\setminus W^{j_t}|=t$ and $|W^{j_{x-1}}\setminus W^{j_{x}}|=1$ for all $x\in \{1,\dots, t\}$. The vectors $\hat u^{j_0,j_1}, \hat u^{j_1,j_2}, \dots, \hat u^{j_{t-1},j_t}$ are linearly independent.
\end{lemma}
\begin{proof}
    Let $f\in\mathcal{F}^2$ denote a non-imposing ABC voting rule and consider a sequence of committees $W^{j_0},\dots, W^{j_t}$ as specified by the lemma. The conditions that $|W^{j_0}\setminus W^{j_t}|=t$ and $|W^{j_{x-1}}\setminus W^{x}|=1$ for $x\in \{1,\dots, t\}$ means that when moving from $W^{j_{x-1}}$ to $W^{j_x}$, we need to exchange a candidate $a\in W^{j_{x-1}}\cap (W^{j_0}\setminus W^{j_t})$ with a candidate $b\in W^{j_x}\cap (W^{j_t}\setminus W^{j_0})$. Hence, we can write each committee in this sequence as $W^{j_x}=\{b_1,\dots, b_x, a_{x+1}, \dots, a_t, c_{t+1}, \dots, c_k\}$. In particular, $W^{j_0}\setminus W^{j_t}=\{a_1,\dots a_t\}$, $W^{j_t}\setminus W^{j_0}=\{b_1, \dots, b_t\}$, and $W^{j_0}\cap W^{j_t}=\{c_{t+1}, \dots, c_k\}$.

    For showing that the vectors $\hat u^{j_{x-1},j_x}$ for $x\in \{1,\dots, t\}$ are linearly independent, we consider the matrix $M$ whose $x$-th row corresponds to the vector $\hat u^{j_{x-1},j_x}$ for $x\in \{1,\dots, t\}$. In more detail, we will show that the image of $M$ has full dimension, i.e., $\{w\in \mathbb{R}^t\colon \exists v\in \mathbb{R}^{|\mathcal{A}|}\colon Mv=w\}=\mathbb{R}^t$. This suffices to prove the lemma because the image dimension of a matrix is equivalent to its rank, which is equivalent to the number of linearly independent rows. To this end, we consider the profiles $A^{x}$ constructed in \Cref{lem:iolProfiles}. In more detail, we claim that the vectors $w=M v(A^{a_x})$ satisfy $w_x\neq 0$ and $w_y=0$ for $y\neq x$, which implies that the image of $M$ has full dimension. 

    To prove this claim, we consider first two indices $x,y\in \{1,\dots, t\}$ with $x\neq y$. Now, let $\tau$ denote the permutation defined by $\tau(a_y)=b_y$, $\tau(b_y)=a_y$, and $\tau(z)=z$ for all other candidates. It is easy to see $\tau(v(A^{a_x}))=v(A^{a_x})$ due to the symmetry of this profile and $\tau(\hat u^{j_{y-1}, j_y})=\hat u^{j_{y}, j_{y-1}}=-\hat u^{j_{y-1}, j_y}$ because of \Cref{lem:hyperplanes}. Hence, it follows that $v(A^x)\hat u^{j_{y-1}, j_y}=\tau(v(A^x))\tau(\hat u^{j_{y-1}, j_y})=-v(A^x)\hat u^{j_{y-1}, j_y}$, which is only possible if $v(A^x)\hat u^{j_{y-1}, j_y}=0$.

    Finally, for showing that $v(A^{a_x})\hat u^{j_{x-1}, j_x}\neq 0$, we consider a committee $W^i$ with $a_x\in W^i$, $b_x\not\in W^i$ and define $A$ as the profile that contains a copy of $A^{z}$ for every $z\in W^i$. By an analogous argument as in the last paragraph, we infer that $v(A)\hat u^{j_{x-1}, j_x}=\sum_{z\in W^i} v(A^{z})\hat u^{j_{x-1}, j_x}=v(A^{a_x})\hat u^{j_{x-1}, j_x}$. On the other hand, consistency and \Cref{lem:iolProfiles} show that $f(A)=\{W^i\}$. Moreover, continuity implies that there is a $\lambda\in\mathbb{N}$ such that $f(\lambda A+ B(\ell))=\{W^i\}$ for all ballots $B(\ell)$, so we can infer that $v(A)\in \text{int} \bar R_i^f$. This means that $v(A)\hat u^{j_{x}, j_{x-1}}>0$ by Claim (1) of \Cref{lem:hyperplanes} and thus, $v(A)\hat u^{j_{x}, j_{x-1}}=v(A^{a_x})\hat u^{j_{x-1}, j_x}<0$. 
\end{proof}

Our next goal is to show that the linear independence observed in \Cref{lem:iolLinearIndependence} turns into a linear dependence if we add the vector $\hat u^{j_0, j_t}$ to the set. For this, we first show an auxiliary claim. 

\begin{lemma}\label{lem:iolPathIndependence}
    Assume $2\leq k\leq m$ and let $f\in\mathcal{F}^2$ denote a non-imposing ABC voting rule. Moreover, consider an arbitrary sequence of committees $W^{j_0},\dots, W^{j_t}\in\mathcal{W}_k$ such that $|W^{j_0}\setminus W^{j_t}|=t$ and $|W^{j_{x-1}}\setminus W^{j_x}|=1$ for all $x\in \{1,\dots, t\}$. 
    Let $B(\ell)$ be any ballot such that $|W^{j_0}\cap B(\ell)|  \leq |W^{j_t}\cap B(\ell)|$. It holds that $\sum_{x=1}^t s^1(|B(\ell)\cap W^{j_{x-1}}|, |B(\ell)\cap W^{j_{x}}|) = \sum_{x=|W^{j_0}\cap B(\ell)|+1}^{|W^{j_t}\cap B(\ell)|} s^1(x-1,x)$.
\end{lemma}
\begin{proof}
Let $f\in\mathcal{F}^2$ denote an non-imposing ABC voting rule and let the committees $W^{j_0},\dots, W^{j_t}$ and the ballot $B(\ell)$ be defined as in the lemma. Furthermore, let $i^x = |B(\ell)\cap W^{j_x}|$ be the $x$-th intersection size and consider the sum $\sum_{x=1}^{t} s^1(i^{x-1}, i^{x})$. If $i^{x-1}=i^{x}$, then $s^1(i^{x-1}, i^{x})=0$. Hence, we can shorten the sum by removing all such terms without affecting the sum. More rigorously, we define $y_x=0$ and $\hat i^{x}=i^x$ for $x=0$, and $\hat i^{x+1}=i^{y_{x+1}}$ where $y_{x+1}$ is the smallest integer $y>y_x$ such that $i^y\neq i^{y_x}$. Moreover, let $\hat t$ denote the length of this new sequence and observe that $\hat i^x=\hat i^{x+1}-1$ or $\hat i^x=\hat i^{x+1}+1$ for all $x\in \{0,\dots, \hat t-1\}$. Furthermore, by definition, $\sum_{x=1}^{t} s^1(i^{x-1}, i^{x})= \sum_{x=1}^{\hat t} s^1(\hat i^{x-1}, \hat i^{x})$. 

Next, we suppose that $\hat t>|W^{j_t}\cap B(\ell)|-|W^{j_0}\cap B(\ell)|$. In this case, it is straightforward to see that there must be an index $\hat i^{x}$ such that $\hat i^{x}=\hat i^{x+2}$. By \Cref{lem:s1}, we thus have that $s_1(i^x, i^{x+1})=-s_1(i^{x+1}, i^{x})=-s_1(i^{x+1}, i^{x+2})$ and we can hence remove these two terms from our sum. Clearly, we can then compress our indices again and repeat the argument until we have only $\bar t\leq |W^{j_t}\cap B(\ell)|-|W^{j_0}\cap B(\ell)|$ intersection sizes left. Let $\bar i^{0},\dots, \bar i^{\bar t}$ denote this reduced set and note that $|\bar i^{x}-\bar i^{x+1}|=1$ for all $x$. Moreover, it is not difficult to see that $\bar i^{0}=i^0$ and $\bar i^{t}=i^t$, so we have that $\bar i^{x-1}=\bar i^{x}-1$ for all $x\in \{1,\dots, \bar t\}$ and $\bar t=|W^{j_t}\cap B(\ell)|-|W^{j_0}\cap B(\ell)|$. \sloppy{Finally, since we only remove terms that sum up to $0$, it clearly holds that $\sum_{x=1}^t s^1(i^{x-1}, i^x)=\sum_{x=1}^{\bar t} s^1(\bar i^{x-1}, i^x)s^1(\bar i^{x-1}, i^x)=\sum_{x=|W^{j_0}\cap B(\ell)|+1}^{|W^{j_t}\cap B(\ell)|} s^1(x-1,x)$, thus proving the lemma.}
\end{proof}

\sloppy{Based on \Cref{lem:iolLinearIndependence,lem:iolPathIndependence}, we will now show that the vector $\hat u^{j_0, j_t}$ can be represented as (scaled) sum of the vectors $\hat u^{j_0,j_1}, \hat u^{j_1,j_2}, \dots, \hat u^{j_{t-1},j_t}, \hat u^{j_0,j_t}$ are linearly dependent.}

\begin{lemma}\label{lem:iolLinearDependence}
Suppose $2\leq k\leq m-2$ and let $f\in\mathcal{F}^2$ denote a non-imposing ABC voting rule. Moreover, consider an arbitrary sequence of committees $W^{j_0},\dots, W^{j_t}\in\mathcal{W}_k$ such that $|W^{j_0}\setminus W^{j_t}|=t$ and $|W^{j_{x-1}}\setminus W^{j_x}|=1$ for all $x\in \{1,\dots, t\}$. There is a $\delta>0$ such that $\hat u^{j_0, j_t}=\delta\sum_{x=1}^t \hat u^{j_{x-1}, j_x}$.
\end{lemma}
\begin{proof}
    Let $f\in\mathcal{F}^2$ denote an non-imposing ABC voting rule and consider a sequence of committees $W^{j_0}, \dots, W^{j_t}\in\mathcal{W}_k$ as stated by the lemma. Moreover, we define $W^{j_0}\setminus W^{j_t}=\{a_1,\dots, a_t\}$, $W^{j_t}\setminus W^{j_0}=\{b_1,\dots,b_t\}$, and $W^{j_x}=(W^{j_0}\cap W^{j_t})\cup \{b_1,\dots, b_x, a_{x+1}, \dots, b_t\}$. Next, we consider the function $s^1(x+1,x)$ derived in \Cref{lem:IoLCommitteeSymmetry}. First, if this function is constant, then we can use the same arguments as in \Cref{app:BSWAV}. In particular, note here that we show in \Cref{lem:convexEquidistance} an analogous claim and that the subsequent lemmas (\Cref{lem:uabProperties,lem:convexLinearIndependence,lem:convexLinearDependence,lem:convexAdditivity}) do not rely on choice set convexity or weak efficiency. 
    
    We hence suppose that there is a index $p\geq 1$ such $s(p+1,p)\neq s(p,p-1)$. Moreover, we define $p$ as minimal such value and let $\alpha=s^1(1,0)$. In particular, we have that $s^1(x+1,x)=\alpha$ for all $x< p$. In this case, we will prove the lemma by an induction on the length of the considered sequence. First, if $t=1$, the statement is trivial and the induction basis thus holds. Hence, we aim to show the lemma for $t>1$ and suppose that there is a $\delta'$ with $\hat u^{i_0, i_{t'}}=\delta'\sum_{x=1}^{t'} \hat u^{i_{x-1}, i_x}$ for all sequences of committees $W^{i_0}, \dots, W^{i^{t'}}$ with $|W^{i_0}\setminus W^{i_{t'}}|=t'<t$ and $|W^{i_{x-1}}\setminus W^{i_x}|=1$ for all $x\in \{1,\dots, t'\}$. To prove the induction step for our sequence $W^{j_0}, \dots, W^{j_t}$, we proceed in multiple steps. In our first four steps (Steps 1.1 to 1.4), we will show that the vectors $\{\hat u^{j_0,j_1}, \hat u^{j_1,j_2}, \dots, \hat u^{j_{t-1},j_t}, \hat u^{j_0,j_t}\}$ are linearly dependent. Based on \Cref{lem:iolLinearIndependence}, it thus follows that there are coefficients $\delta_x$, not all of which are $0$, such that $\hat u^{j_0,j_t}=\sum_{x=1}^t \delta_x \hat u^{j_{x-1}, j-x}$. In the last step (Step 2), we then show that all $\delta_x$ are the same and greater $0$, thus proving the lemma.\medskip
     
    \textbf{Step 1: The vectors $\{\hat u^{j_0,j_1}, \hat u^{j_1,j_2}, \dots, \hat u^{j_{t-1},j_t}, \hat u^{j_0,j_t}\}$ are linearly independent.}

    Assume for contradiction that the given vectors are linearly independent. To derive a contradiction to this assumption, we will construct a vector $v^4$ that is not contained in $\bar R^f_i$ for any $W^i\in\mathcal{W}_k$. Consequently, $\bigcup_{x\in  \{1,\dots,|\mathcal{W}_k|\}} \bar R_i^f\neq \mathbb{R}^{|\mathcal{A}|}$. Just as for \Cref{lem:convexLinearDependence}, this contradicts the insight that $\bigcup_{x\in  \{1,\dots,|\mathcal{W}_k|\}} R_i^f=\mathbb{Q}^{|\mathcal{A}|}$ and therefore $\bigcup_{x\in  \{1,\dots,|\mathcal{W}_k|\}} \bar R_i^f=\mathbb{R}^{|\mathcal{A}|}$. So, the assumption that the set $\{\hat u^{j_0,j_1}, \hat u^{j_1,j_2}, \dots, \hat u^{j_{t-1},j_t}, \hat u^{j_0,j_t}\}$ is linearly independent must have been wrong. For constructing $v^4$, we will step by step narrow down the choice set.\medskip

    \textbf{Step 1.1:} For our first step, let $\mathcal{W}^1=\{W\in\mathcal{W}^k\colon W^i\cap W^j\subseteq W\subseteq W^{i}\cup W^j\}$, i.e., $\mathcal{W}^1$ is the convex hull of $W^i$ and $W^j$. We will construct a vector $v^1$ such that for every committee $W^i\not\in \mathcal{W}^1$, there is another committee $W^j$ such that $v^1\hat u^{i,j}<0$. This shows that $v^1\not\in \bar R_i^f$ for these committees by Claim (1) of \Cref{lem:hyperplanes}. 
    
    For constructing this vector, we recall the profiles $A^x$ and $A^{-x}$ constructed in \Cref{lem:iolProfiles}. First, we note that $v(A^x)\hat u^{i,j}=v(A^{-x})\hat u^{',j}=0$ for all committees $W^{i},W^{j}$ with $|W^{i}\setminus W^{j}|=1$, $\{x\}\neq \{a\}=W^{i}\setminus W^{j}$, and $\{x\}\neq \{b\}=W^{j}\setminus W^{i}$. For showing this claim, considering the permutation $\tau$ with $\tau(a)=b$, $\tau(b)=a$, and $\tau(z)=z$ for all remaining candidates. It is now easy to verify that $v(A^x) \hat u^{i,j}=\tau (v(A^x)) \tau(\hat u^{i,j})=v(A^x)\hat u^{j,i}=-v(A^x) \hat u^{i,j}$ due to the symmetry of $A^x$ and \Cref{lem:hyperplanes}. This is only possible if $v(A^x) \hat u^{i,j}=0$ and an analogous argument also shows our claim for $A^{-x}$. Furthermore, it holds that $v(A^x)\hat u^{i,j}>0$ and $v(A^{-x})\hat u^{i,j}<0$ for all committees $W^{i},W^{j}$ with $W^{i}\setminus W^{j}=\{x\}$. For showing this, consider the profile $A$ that consists of a copy of $A^{x}$for every $x\in W^{i}$ (the claim for $A^{-x}$ works analogously by considering the profile consisting of $A^{-x}$ for $x\not\in W^{i}$). By consistency, $f(A)=\{W^i\}$, and by continuity, we infer that $v(A)\in \text{int } \bar R_{i}^f$. Hence, by Claim (1) of \Cref{lem:hyperplanes}, $v(A)\hat u^{i,j}>0$. On the other hand, we have that $v(A)\hat u^{i,j}=\sum_{z\in W^{i}} v(A^z)\hat u^{i,j}=v(A^x)\hat u^{i,j}$ as $v(A^z) \hat u^{i,j}=0$ for all $z\in W^{i}\cap W^{j}$. Combining these insights shows that $v(A^x)\hat u^{i,j}>0$.

    Now, for completing this step, we define $A^1$ as the profile that contains a copy of $A^{x}$ for $x\in W^{j_0}\cap W^{j_t}$ and a copy of $A^{-x}$ for $x\in \mathcal{C}\setminus (W^{j_0}\cup W^{j_t})$. By consistency, it is easy to infer that $f(A^1)=\mathcal{W}^1$. Hence, Claim (1) of \Cref{lem:hyperplanes} shows for $v^1=v(A^1)$ and $W^i, W^j\in \mathcal{W}^1$ that $v^1 \hat u^{i,j}=0$. Next, consider a committee $W^i\not\in \mathcal{W}^1$. This means that there is a pair of candidates $a\in W^i, b\not\in W^i$ such that $a\in \mathcal{C}\setminus (W^{j_0}\cap W^{j_t})$ and $b\in W^{j_0}\cap W^{j_t}$, or $a\in \mathcal{C}\setminus (W^{j_0}\cup W^{j_t})$ and $b\in W^{j_0}\cup W^{j_t}$. In both cases, it follows from our previous analysis that $v^1 \hat u^{i,j}<0$ for the committee $W^j$ defined by $W^j=(W^i\setminus \{a\})\cup \{b\}$. For instance, if $a\in \mathcal{C}\setminus (W^{j_0}\cup W^{j_t})\subseteq \mathcal{C}\setminus (W^{j_0}\cap W^{j_t})$, $b\in W^{j_0}\cap W^{j_t}$, then $v^1\hat u^{i,j}=v(A^{-a}) \hat u^{i,j}+v(A^{b})\hat u^{i,j}=v(A^{-a}) \hat u^{i,j}-v(A^{b})\hat u^{j,i}<0$. This completes this step.\medskip

    \textbf{Step 1.2:} For our second step, let $\mathcal{W}^2=\{W\in \mathcal{W}^1\colon \forall x\in \{1,\dots, t\}\colon \{a_x,b_x\}\}$. As second step, we will construct a vector $v^2$ such that for each $W^i\not\in\mathcal{W}^2$, there is a committee $W^j$ such that $v^2\hat u^{i,j}<0$. 
    
    For constructing this vector, recall that $s^1(x+1,x)=\alpha$ for all $x< p$ and $s^1(p+1,p)\neq \alpha$. Moreover, consider two arbitrary committees $W^{i}, W^{j}\in \mathcal{W}^1$ with $\{W^{i}, W^j\}\neq \{W^{j_0}, W^{j_t}\}$. By this assumption, it holds that $|W^i\setminus W^j|=t'<t$, so we can use our induction hypothesis to construct $\hat u^{i,j}$. For doing so, let $W^{i_0},\dots W^{i_{t'}}$ denote a sequence of committees from $W^{i}$ to $W^j$. By the induction hypothesis, $\hat u^{i,j}=\delta \sum_{x=1}^{t'} \hat u^{i_{x+1}, i_x}=$ for some $\delta>0$. In turn, \Cref{lem:iolPathIndependence} shows that $\hat u^{i,j}_\ell=\delta \sum_{x=|B(\ell)\cap W^i|+1}^{|B(\ell)\cap W^j|} s^1(x-1,x)$ for all ballots $B(\ell)$ with $|B(\ell)\cap W^i|\leq |B(\ell)\cap W^j|$. Hence, if additionally $|B(\ell)\cap W^i|\leq p$ and $|B(\ell)\cap W^j|\leq p$, then $\hat u^{i,j}_\ell=-\delta\alpha (|B(\ell)\cap W^j|- |B(\ell)\cap W^i|)=\delta\alpha (|B(\ell)\cap W^i|- |B(\ell)\cap W^j|)$. On the other hand, if $|B(\ell)\cap W^i|\leq p$ and $|B(\ell)\cap W^i|=p+1$, then $\hat u^{i,j}_\ell=\delta\alpha (|B(\ell)\cap W^i|- p)-\delta s(p+1,p)=\delta\alpha (|B(\ell)\cap W^i|- |B(\ell)\cap W^j|)+\delta\alpha-\delta s(p+1,p)$.

    Now for constructing the vector $v^2$, we consider first the profiles $\bar A^{x}$ that contain ballot $A$ with $|A|=p+1$ and $a_x, b_x\in A$ once. Moreover, we define $\bar A$ as the profile that consists of a copy of $\bar A^{x}$ for all $x\in \{1,\dots, t\}$ and let $\bar v=v(\bar A)$. 
    
    First, we observe that all candidates in $(W^{j_0}\cup W^{j_t})\setminus (W^{j_0}\cap W^{j_t})$ are approved by the same number of voters in $\bar A$. Hence all committees in $\mathcal{W}^1$ have the same total number of approvals in $\bar A$, i.e.,  $\sum_{\ell\leq  |\mathcal{A}|} \bar v_\ell |B(\ell)\cap W^i|=\sum_{\ell\leq  |\mathcal{A}|} \bar v_\ell |B(\ell)\cap W^j|$ for all committees $W^i,W^j\in\mathcal{W}^1$. Moreover, note that for every ballot $A$ in $\bar A$ and every committee $W\in\mathcal{W}^2$, it holds that $|W\cap A|\leq p$ because $|A|=p+1$ and $\{b_x,a_x\}\subseteq A$ for some $x\in \{1,\dots, t\}$ but $\{b_x,a_x\}\not\subseteq W$. Hence, for all $W^i, W^j\in \mathcal{W}^2$ with $\{W^i, W^j\}\neq \{W^{j_0}, W^{j_t}\}$, the following equation holds due to our previous analysis. In this equation, we define $I_1=\{\ell\in \{1,\dots, |\mathcal{A}|\}\colon |B(\ell)\cap W^i|\leq |B(\ell)\cap W^j|\}$ and $I_1=\{\ell\in \{1,\dots, |\mathcal{A}|\}\colon |B(\ell)\cap W^i|> |B(\ell)\cap W^j|\}$.
    \begin{align*}
        \bar v \hat u^{i,j}&=\sum_{\ell\in \{1,\dots, |\mathcal{A}|\}} \bar v_\ell \hat u^{i,j}_\ell\\
        &=\sum_{\ell\in I_1|} \bar v_\ell \hat u^{i,j}_\ell - \sum_{\ell\in I_2} \bar v_\ell \hat u^{j,i}_\ell\\
        &=\sum_{\ell\in I_1} \bar v_\ell \delta\alpha(|B(\ell)\cap W^i|- |B(\ell)\cap W^j|) \\
        &\qquad- \sum_{\ell\in I_2} \bar v_\ell \delta\alpha(|B(\ell)\cap W^j|- |B(\ell)\cap W^i|)\\
        &=\delta\alpha \sum_{\ell\in \{1,\dots, |\mathcal{A}|\}} \bar v_\ell |B(\ell)\cap W^i| \\
        &\qquad- \delta\alpha\sum_{\ell\in \{1,\dots, |\mathcal{A}|\}} \bar v_\ell |B(\ell)\cap W^j|\\
        &=0. 
    \end{align*}

    Moreover, it also holds that $\bar v \hat u^{j_0, j_t}=0$. For explaining this, we consider the permutation $\tau$ with $\tau(a_x)=b_x$, $\tau(b_x)=a_x$ for all $x\in \{1,\dots, t\}$ and $\tau(z)=z$ for all remaining candidates. It is easy to see that $\tau(v(\bar A^{x}))=v(\bar A^x)$ for all $x\in \{1,\dots, t\}$ and hence $\tau(\bar v)=\bar v$. By our usual permutation arguments, it thus follows that $\bar v \hat u^{j_0, j_t}=0$. 

    Next, consider a committee $W^i\in\mathcal{W}^1\setminus \mathcal W^2$. Then $a_x,b_x\in W^i$ for some $x\in \{1,\dots, t\}$. Moreover, let $\gamma$ denote the number of ballots $A$ in $\bar A$ such that $|A\cap W^i|=p+1$. 
    Since $p+1\leq k$ and $a_x,b_x\in W^y$, there is at least one ballot $A$ in $\bar A^{x}$ such that $A\subseteq W$ and thus $\gamma\geq 1$. We claim that $\bar v \hat u^{j_0,i} >0$ if $\alpha>s^1(p,p-1)$ and $\bar v \hat u^{j_0,i}  <0$ if $\alpha< s^1(p+1,p)$. 
    Note for this first that $|W^{i_0}\cap W^y|<t$ since $W^i\in \mathcal{W}^1$ and $a_x\in W^i$. Moreover, it holds for every ballot $B(\ell)$ with $\bar v_\ell\neq 0$ that $|W^{j_0}\cap B(\ell)|\leq p$ since $|B(\ell)|=p+1$ and $b_x\in B(\ell)\setminus W^{j_0}$ for some $x\in \{1,\dots, t\}$. On the other hand, $|W^{i}\cap B(\ell)|\leq p+1$. Using our initial insights, we thus derive the following equations, where $I_1=\{\ell\in \{0,\dots,|\mathcal{A}|\} \colon |B(\ell)\cap W^i|=p+1\}$ and $I_2=\{\ell\in \{0,\dots,|\mathcal{A}|\}\colon |B(\ell)\cap W^i|\leq p\}$. 
    \begin{align*}
        \bar v \hat u^{j_0,i} &=\sum_{\ell\in I_1} \bar v_\ell  \hat u^{j_0, i}_\ell +
        \sum_{\ell\in I_2} \bar v_\ell \hat u^{j_0, i}_\ell  \\
        & = \sum_{\ell\in I_1} 
        \bar v_\ell \delta\alpha (|B(\ell)\cap W^{j_0}|-|B(\ell)\cap W^i|)\\
        &\qquad + \sum_{\ell\in I_1} 
        \bar v_\ell(\delta\alpha - \delta s^1(p+1,p))\\
        &\qquad+ \sum_{\ell\in I_2} \bar v_\ell \delta\alpha  (|B(\ell)\cap W^{j_0}|-
        |B(\ell)\cap W^i|)\\
        &=\delta\gamma (\alpha-s^1(p+1, p))\\
        &\qquad+\delta\alpha \sum_{\ell\in \{1,\dots,|\mathcal{A}|\}} \bar v_\ell |B(\ell)\cap W^{j_0}| \\
        &\qquad -\delta\alpha \sum_{\ell\in \{1,\dots,|\mathcal{A}|\}} \bar v_\ell |B(\ell)\cap W^{i}|\\
        &=\delta\gamma (\alpha-s^1(p+1, p)).
    \end{align*}
    In particular, we use in the last steps that all committees in $\mathcal{W}^1$ have the same total number of approvals. Since $\delta>0$ and $\gamma>0$, this shows that $\bar v \hat u^{j_0,i} < 0$ if $\alpha<s(p+1, p)$ and $\bar v \hat u^{j_0,i}>0$ if $\alpha> s(p+1, p)$. Hence, with the right sign in front of $\bar v$, all $W^y\in\mathcal{W}^1\setminus \mathcal W^2$ are dominated.

    Finally, let $v^2=\lambda v^1+\bar v$ if $\alpha>s(p+1, p)$ and $v^2=\lambda v^1-\bar v$ if $\alpha < s(p+1, p)$. In this definition, $\lambda>0$ is so large that for all committees $W^i\in \mathcal{W}_k\setminus \mathcal{W}^1$, there is another committee $W^j\in\mathcal{W}_k$ such that $v^2 \hat u^{i,j}<0$. 
    Now, note that for all $W^i\in \mathcal{W}^1\setminus \mathcal{W}^2$, we have that $v^2 \hat u^{i,j_0}=-v^2 \hat u^{j_0,x}<0$ since $v^1\hat u^{j_0,i}=0$ and we choose the sign of $\bar v$ such that $\pm \bar v \hat u^{j_0,i}>0$. Following a similar reasoning, it is easy to see that $v^2 \hat u^{i,j}=0$ for $W^i, W^j\in \mathcal{W}^2$.\medskip
         
    \textbf{Step 1.3:} For constructing our next vector, we define $\mathcal{W}^3=\{W\in\mathcal{W}^2\colon \forall x\in \{1,\dots, t-1\}\colon \{a_x, b_{x+1}\}\not\in W\}$. Put differently, $\mathcal{W}^3$ consists all committees $W^x$ such that $W^x=W^{j_0}\cap W^{j_t}\cup \{b_1,\dots, b_{x}, a_{x+1}, \dots, b_t\}$, i.e., $\mathcal{W}^3=\{W^{j_0}, \dots, W^{j_t}\}$. We aim to construct a vector $v^3$ such that all committees outside of $\mathcal{W}^3$ are dominated by some other committee. 

    To this end, consider the profile $\hat A^{x}$ for $x\in \{1,\dots, t-1\}$ that contains each ballot $A$ with $|A|=p+1$ and $\{a_x, b_{x+1}\}$ once. Furthermore, define the profile $\hat A$ as follows: $\hat A$ contains a copy of $\hat A^{x}$ for each $x\in \{1,\dots, t-1\}$ and so many copies of the ballots $\{a_t\}$ and $\{b_1\}$ that every candidate in $(W^{j_0}\cup W^{j_t})\setminus (W^{j_0}\cap W^{j_t})$ is approved by the same number of voters.
    Moreover, we define $\hat v=v(\hat A)$. By the definition of $\hat A$, it immediately follows that $\sum_{\ell\in \{1,\dots, |\mathcal{A}|\}} \hat v_\ell |B(\ell)\cap W|=\sum_{\ell\in \{1,\dots, |\mathcal{A}|\}} \hat v_\ell |B(\ell)\cap W'|$ for all $W,W'\in \mathcal{W}^2$. Furthermore, it holds for all committees $W\in\mathcal{W}^3$ and ballots $A$ in $\hat A$ that $|W\cap A|\leq p$. Hence, it follows from exactly the same reasoning as in the last step that $\hat v\hat u^{i,j}=0$ for all $W^i, W^j\in\mathcal{W}^3$ with $\{W^i, W^j\}\neq \{W^{j_0}, W^{j_t}\}$. Moreover, the permutation $\tau$ with $\tau(a_x)=b_{t-x+1}$, $\tau(b_x)=a_{t-x+1}$ for $x\in \{1,\dots, t\}$ and $\tau(z)=z$ for all other candidates maps $\hat A^{x}$ to $\hat A^{t-x+1}$. This means that $\tau(\hat A)=\hat A$ and our permutation argument thus also shows that $\hat v\hat u^{j_0, j_t}=0$. 

    Next, consider a committee $W^i\in\mathcal{W}^2\setminus\mathcal{W}^3$. In particular, this means that $\{a_x, b_{x+1}\}\subseteq W^i$ for some $x\in \{1,\dots, t-1\}$. Now, let $\gamma$ denote the number of ballots $A$ in $\hat A$ such that $|A\cap W^y| = p + 1 $. Since $p + 1\leq k$ and $\{a_x, b_{x+1}\}\subseteq W$, there is clearly a ballot $B(\ell)$ such that $|B(\ell)|=p +1$, $\hat v_\ell\neq 0$, and $B(\ell)\subseteq W$, so $\gamma\geq 1$. Since $a_x\in W^i$, this means that $|W^{j_0}\setminus W^i|<t$. Again, for all ballots $B(\ell)$ with $\hat v_\ell\neq0$, we have $|B(\ell)\cap W^{j_0}|\leq p$ and $|B(\ell)\cap W^{i}|\leq p+1$. So, we can apply the same reasoning as for $\bar v$ to infer that $\hat u^{j_0,i} \hat v > 0$ if $\alpha>s^1(p+1, p)$ and $\hat u^{j_0,i} \hat v < 0$ if $\alpha < s^1(p + 1, p)$.

    Finally, let $v^3=\lambda v^2+\hat v$ if $\alpha>s(p+1, p)$ and $v^3=\lambda v^2-\hat v$ if $\alpha < s(p+1, p)$. Here, we choose $\lambda>0$ again so large that for every committee $W^i\in\mathcal{W}_k\setminus \mathcal{W}^2$, there is another committee $W^j\in\mathcal{W}_k$ such that $v^3 \hat u^{i,j}  < 0$. Moreover, note that for every committee $W^i\in\mathcal{W}^2\setminus \mathcal{W}^3$, we have that $v^3 \hat u^{i,j_0} < 0$ because $\hat u^{i,j_0} v^2 = 0$ and we choose the sign of $\hat v$ so that $\pm \hat v  \hat u^{j_0, i} > 0$. Finally, $v^3 \hat u^{x,y}  = 0$ for all $x,y\in\mathcal{W}^3$ because $ v^2 \hat u^{x,y}= 0$ and $\hat v \hat u^{x,y} = 0$. Hence, for every committee $W^i\in\mathcal{W}_k\setminus \mathcal{W}^3$, there is another committee $W^j$ such that $v^3 \hat u^{i,j}<0$, and for all $W^{i}, W^{j}\in \mathcal{W}^3$, it holds that $v^3\hat u^{i,j}=0$.\medskip 
    
    \textbf{Step 1.4:} As last step, we consider the matrix $M$ that contains the vectors $\hat u^{j_0, j_1}, \dots, \hat u^{j_{x-1}, j-x}, \hat u^{j_0, j_t}$ as rows. More specifically, we assume that the $x$-th row of $M$ is $\hat u^{j_{x-1},j_x}$ for $x\in \{1,\dots, t\}$ and the $t+1$-th row of $M$ is $\hat u^{j_0,j_t}$. Now, by assumption, the rows of $M$ are linearly independent, i.e., the matrix has row rank of $t+1$. This means equivalently that it has a column rank of $t+1$, which in turn implies that the image of $M$ has full dimension. Thus, there is a vector $v^*$ such that $w=Mv^*$ satisfies $w_x=1$ for $x\in \{1,\dots, t\}$ and $w_{t+1}=-1$. 
    
    Next, just as in the previous steps, we define $v^4=\lambda v^3+v^*$, where $\lambda>0$ is so large that for every committee $W^i\in\mathcal{W}_k\setminus \mathcal{W}^3$, there is another committee $W^j\in\mathcal{W}_k$ with $v^4 \hat u^{y,z}<0$. Now, by definition of $v^4$ and Claim (1) of \Cref{lem:hyperplanes}, $v^4\not\in\bar R_{i}^f$ for every $W^i\in\mathcal{W}_k\setminus \mathcal{W}^3$. On the other hand, we have shown in Step 3 that $v^3 \hat u^{i,j}=0$ for all $W^i, W^j\in\mathcal{W}^3$. So $v^4 \hat u^{i,j}=v^*\hat u^{i,j}$ for these committees. This means that $v^4\hat u^{j_0, j_t}=-1<-0$ and $v^4 \hat u^{j_{x-1} j_x}=-v^4\hat u^{j_x, j_{x-1}}=-1<0$ for all $x\in \{1,\dots, t\}$. So, $v^4\not\in \bar R_i^f$ for any committee $W^i\in\mathcal{W}_k$, which gives us the desired contradiction. Hence, the initial assumption is wrong and the vectors $\hat u^{j_0, j_1}, \dots, \hat u^{j_{t-1}, j_t}, \hat u^{j_0, j_t}$ are linearly independent.\medskip

    \textbf{Step 2: There is $\delta>0$ such that $\hat u^{j_0, j_t}=\delta \sum_{x=1}^t \hat u^{j_{x-1}, j-x}$}

    By Step 1, we know that the set $\{\hat u^{j_0, j_1}, \dots, \hat u^{j_{t-1}, j_t}, \hat u^{j_0, j_t}\}$ is linearly dependent, whereas the set $\{\hat u^{j_0, j_1}, \dots, \hat u^{j_{t-1}, j_t}\}$ is linearly independent (\Cref{lem:iolLinearIndependence}). Consequently, there are unique values $\delta_x$ for $x\in \{1,\dots, t\}$, not all of which are $0$, such that $\hat u^{j_0,j_t}=\sum_{x=1}^{t} \delta_x \hat u^{j_{x-1}, j_x}$. Now, consider the profiles $A^x$ constructed in \Cref{lem:iolProfiles}. As discussed before, it holds that $\hat u^{i,j} v(A^c) =0$ for all committees $W^i, W^j\in\mathcal{W}_k$ such that $|W^i\setminus W^j|=1$ and $c\in W^i\cap W^j$ or $c\not\in W^i\cup W^j$. 
    Conversely, $\hat u^{i,j} v(A^c) > 0$ for all committees $W^i, W^j\in\mathcal{W}_k$ with $W^i\setminus W^j=\{c\}$. Moreover, observe that $\hat u^{y_1,z_1} v(A^{c_1}) =\hat u^{y_2,z_2} v( A^{c_2}) \neq 0$ for all committees $W^{y_1}, W^{z_1}, W^{y_2}, W^{z_2}\in\mathcal{W}_k$ with $W^{y_1}\setminus W^{z_1}=\{c_1\}$ and $W^{y_2}\setminus W^{z_2}=\{c_2\}$. This can again be proven by choosing a suitable permutation $\tau$. 
    
    Now, let $x\in \{1,\dots, t\}$ be an arbitrary index. By our previous argument, we have that $v(A^{a_{x}}) \hat u^{j_{x-1},j_x}> 0$ because $a_{x}\in W^{j_{x-1}}\setminus W^{j_x}$. On the other side, $a_{x}\in W^{j_{y-1}}\cap \bar W^{j_y}$ for all $1\leq y<x$, and $a_{x}\not\in \bar W^{j_{y-1}}\cup \bar W^{j_y}$ for all $y>x$. So, we infer that $\hat u^{j_{y-1},j_y} v( A^{a_{x}}) = 0$ for all $y\in \{1,\dots, t\}$ with $y\neq x$. Hence, it follows that $v( A^{a_{x}}) \hat u^{j_0,j_t}=\sum_{y=1}^{t} \delta_y v( A^{a_{x}}) \hat u^{j_{y-1},j_y}=\delta_x v( A^{a_{x}})  \hat u^{j_{x-1}, j_x} $. 

    As next step, we consider two distinct indices $x_1, x_2\in \{1,\dots, t\}$ and the profiles $ A^{a_{x_1}}$ and $ A^{a_{x_2}}$. Moreover, let $\tau:\mathcal{C}\rightarrow\mathcal{C}$ be a permutation such that $\tau(a_{x_1})=a_{x_2}$, $\tau(a_{x_2})=a_{x_1}$, and $\tau(c)=c$ for all other candidates. First, $a_{x_1}, a_{x_2}\in W^{j_0}\setminus W^{j_t}$, so $\tau(W^{j_0})=W^{j_0}$ and $\tau(W^{j_t})=W^{j_t}$. Hence, $v(A^{a_{x_1}})\hat u^{j_0, j_t}=\tau(v(A^{a_{x_1}})) \tau(\hat u^{j_0, j_t})=v(A^{a_{x_2}})\hat u^{j_0, j_t}$ by Claim (3) of \Cref{lem:hyperplanes}. Combining our last two insights thus shows that $\delta_{x_1} \hat u^{j_{x_1-1},j-{x_1}} v( A^{a_{x_1}})=\delta_{x_2} \hat u^{j_{x_2-1},j_{x_2}} v( A^{a_{x_2}})$. Since $v( A^{a_{x_1}}) \hat u^{j_{x_1-1},j_{x_1}}=v( A^{a_{x_2}}) \hat u^{j_{x_2-1},j_{x_2}}\neq 0$, this means that $\delta_{x_1}=\delta_{x_2}$. This proves that $\hat u^{j_0,j_t}=\delta \sum_{x=1}^{t} \hat u^{j_{x-1},j_x}$ for some $\delta\in\mathbb{R}$. Moreover, since there is at least one non-zero $\delta_x$, it follows that $\delta\neq 0$

    Finally, we need to show that $\delta>0$. For this, we note that $W^{j_0}\in f(A^{a_1})$ as $a_1\in W^{j_0}$ and therefore $v(A^{a_1})\in \bar R_{j_0}^f$. Now, by Claim (1) of \Cref{lem:hyperplanes}, this means that $v(A^{a_1}) \hat u^{j_0, j_t}\geq 0$. On the other side, we have already shown that $v(A^{a_1)} \hat u^{i,j}=\delta v(A^{a_1}) \hat u^{j_0,j_1}$ and that $v(A^{a_1}) \hat u^{j_0,j_1}$ and $\delta\neq 0$. Combining these claims then shows that $\delta>0$ and thus proves the lemma.
\end{proof}

Finally, we are now ready to prove \Cref{thm:Thiele}. 

\Thiele*
\begin{proof}
We note that the direction from left to right has been shown in the main body. Hence, we focus on the converse direction and assume that $f\in\mathcal{F}^2$. Now, if $f$ is the trivial rule, it is clearly the Thiele rule defined by $s(0)=0$. On the other hand, we can assume that $f$ is non-imposing if it is non-trivial by \Cref{lem:iolProfiles}. This allows us to access the vectors $\hat u^{i,j}$ constructed in \Cref{lem:hyperplanes} and the function $s^1$ constructed in \Cref{lem:IoLCommitteeSymmetry}. Now, we define the function $s(x)$ as follows: $s(0)=0$ and $s(x)=\sum_{y=1}^x s(y, y-1)$ for all $x\in \{1,\dots, k\}$. Moreover, we extend $s$ to vectors by $\hat s(v,W)=\sum_{\ell\in\{1,\dots,|\mathcal{A}|\}} v_\ell s(|B(\ell)\cap W|)$. We will show this lemma by proving that $f(A)=f'(A):=\{W\in\mathcal{W}_k\colon \forall W'\in\mathcal{W}_k\colon \hat s(v(A), W)\geq \hat s(v(A), W')\}$ and that $f'$ is a Thiele rule. For doing so, we proceed in multiple steps.\medskip

\textbf{Step 1: There is $\delta>0$ such that $\delta\hat u^{i,j}_\ell=s(|B(\ell)\cap W^i|)- s(|B(\ell)\cap W^i|)$ for all ballots $B(\ell)$.}

As the first step, we show that the vectors $\hat u^{i,j}$ can be represented by the function $s$. For this, consider two arbitrary committees $W^i, W^j\in\mathcal{W}_k$ and a ballot $B(\ell)$. First, if $|W^i\setminus W^j|=1$, then this claim follows from \Cref{lem:IoLCommitteeSymmetry}. Hence, suppose that $|W^i\setminus W^j|=t\geq 2$, which requires that $2\leq k\leq m-2$, which means that \Cref{lem:iolLinearDependence} applies. To use this lemma, let $W^{j_0}, \dots, W^{j_t}$ denote a sequence of committees from $W^i$ to $W^j$. Then, we infer that $\hat u^{i,j}=\delta \sum_{x=1}^t \hat u^{j_{x-1}, j_x}$ for some $\delta>0$. Now, suppose that $|B(\ell)\cap W^i|\leq |B(\ell)\cap W^j|$. Then, \Cref{lem:iolPathIndependence} shows that $\sum_{x=1}^t \hat u^{j_{x-1}, j_x}=\sum_{x=|B(\ell)\cap W^i|+1}^{|B(\ell)\cap W^j|} s^1(x-1,x)$. By the definition of $s$ and the fact that $s^1(x-1,x)=-s^1(x,x-1)$, we get that $\sum_{x=|B(\ell)\cap W^i|+1}^{|B(\ell)\cap W^j|} s^1(x-1,x)=-\sum_{x=|B(\ell)\cap W^i|+1}^{|B(\ell)\cap W^j|} s^1(x,x-1)=-(s(|B(\ell)\cap W^j) - s(|B(\ell)\cap W^i))=s(|B(\ell)\cap W^i) - s(|B(\ell)\cap W^j)$. Hence, the claim is proven in this case. 

Next, assume that $|B(\ell)\cap W^i|> |B(\ell)\cap W^j|$. In this case, we can consider the vectors $\hat u^{j,i}$ and our previous argument shows that $\hat u^{j,i}_\ell=\delta (s(|B(\ell)\cap W^j) - s(|B(\ell)\cap W^i))$. Finally, the step follows again since $\hat u^{i,j}_\ell=-\hat u^{j,i}_\ell$.\medskip

\textbf{Step 2: $f(A)\subseteq f'(A)$ for all profiles $A\in\mathcal{A}^*$}

For showing this step, consider an arbitrary profile $A$. By our lemmas, we have that $f(A)=\hat g(v(A))=\{W^i\in\mathcal{W}_k\colon v(A)\in R_i^f\}\subseteq \{W^i\in\mathcal{W}_k\colon v(A)\in \bar R_i^f\}$. Hence, our goal is to show that $v(A)\in \bar R_i^f$ if and only if $\hat s(v(A), W^i)\geq s(v(A), W^j)$ for all committees $W^j\in\mathcal{W}_k$. For doing so, we recall that $\bar R_i^f=\{v\in\mathbb{R}^{|\mathcal{A}|}\colon \forall j\in \{1,\dots, |\mathcal{W}_k|\}\setminus \{i\}\colon v\hat u^{i,j}\geq 0\}$. Hence, it clearly suffices to show that $v(A)\hat u^{i,j}\geq 0$ if and only if $\hat s(v(A), W^i)\geq \hat s(v(A), W^j)$. For this, we observe that by Step 1, $v(A) \hat u^{i,j}=\sum_{\ell\in \{1,\dots, |\mathcal{A}|\}} v(A)_\ell \delta (s(|B(\ell)\cap W^i|)- s(|B(\ell)\cap W^i|))=\delta (\hat s(v(A), W^i) - \hat s(v(A), W^j))$ for some $\delta>0$. This shows that our claim holds and thus this step follows.\medskip

\textbf{Step 3: $f(A)\subseteq f'(A)$ for all profiles $A\in\mathcal{A}^*$ and $f'$ is a Thiele rule}

First, we show that $f'(A)$ is a Thiele rule. For this, we note first that $s(0)=0$ by definition and it thus only remains to prove that $s$ is non-decreasing. Now, assume for contradicting that there is an index $p\in \{1,\dots, k\}$ such that $s(p)<s(p-1)$. First, suppose that $p>1$. In this case, consider the profile $A$ in which every ballot of size $p$ is reported once. By anonymity and neutrality, it follows that $f(A)=f'(A)=\mathcal{W}_k$. Next, we consider two arbitrary committees $W,W'\in\mathcal{W}_k$ and let $c\in\mathcal{W}\setminus W'$. Moreover, let $B(\ell)$ denote a ballot such that $B(\ell)\subseteq W$ and $c\in B(\ell)$. Finally, let $A'$ denote the profile in which we replace $B(\ell)$ with $B(\ell)\setminus \{c\}$. It is easy to see that $s(A',W)=s(A,W)-s(p)+s(p-1)>s(A,W)=s(A,W')=s(A',W')$. Hence, $W'\not\in f'(A')$ and therefore also $W'\not\in f(A')$. However, this contradicts independence of losers as $W'\in f(A)$ and $c\not\in W'$. Hence, we infer that $s(p)\geq s(p-1)$ for all $p\in \{2,\dots, k\}$. As second case suppose that $p=1$, which means that $s(1)<s(0)=0$. In this case, let $A$ denote the profile consisting of all ballots of size $2$. Now, consider a ballot $\{x,y\}$ and let $W$ and $W'$ denote committees such that $x\in W$, $y\not\in W$ and $x\not\in W'$, $y\in W'$. Finally, let $A'$ denote the profile derived from $A$ by replacing the ballot $\{x,y\}$ with the ballot $\{y\}$. It is easy to see that $s(A',W)=s(A,W)-s(1)+s(0)>s(A,W)=s(A,W')=s(A',W')$. Hence, $W'\in f(A)$ and $W'\not\in f(A)$ as $W'\not\in f'(A')$. This contradicts independence of losers as $x\not\in W'$. Both cases combined show that $s$ is non-decreasing, so $f'$ is indeed a Thiele rule. 

Finally, we show that $f(A)=f'(A)$ for all profiles $A\in\mathcal{A}^*$. Assume that this is not the case, which means that there is a profile $A$ and a committee $W$ such that $W\in f'(A)\setminus f(A)$ because of Step 2. Moreover, note that $f'$ is a Thiele rule and hence satisfies consistency and all the other axioms of \Cref{thm:Thiele}. Next, since $s$ is non-zero (as the vectors $\hat u^{i,j}$ are non-zero), $f'$ is not the trivial rule. Hence, we can use \Cref{lem:iolProfiles} to show that $f'$ is non-imposing. In particular, there is a profile $A'$ such that $f(A')=f'(A')=\{W\}$. By consistency of $f'$, this means that $f(\lambda A+ A')=f'(\lambda A+ A')=\{W\}$ for every $\lambda\in\mathbb{N}$. However, this contradicts the continuity of $f$ and therefore, our initial assumption must have been wrong. This shows that $f$ is the Thiele rule defined by $s$. 
\end{proof}

\section{Proof of \Cref{prop:PAV,prop:SAV}}

Finally, we discuss the proofs of \Cref{prop:PAV,prop:SAV}.

\PAV*
\begin{proof}
    First, we show that \texttt{PAV} satisfies party-propor\-tio\-na\-lity. To this end, let $A$ denote a party-list profile with parties $\mathcal{P}_A$ and consider two parties $P_i, P_j$ such that $\frac{n_i}{|P_i|}<\frac{n_j}{|P_j|}$. Moreover, we assume for contradiction that there is a committee $W\in \texttt{PAV}(A)$ such that $P_i\subseteq W$ and $P_j\not\subseteq W$. Now, let $x\in P_i\cap W$ and $y\in P_j\setminus W$ and consider the committee $W'=(W\cup \{y\})\setminus \{x\}$. It is easy to compute that 
    $\hat s_{\texttt{PAV}}(A,W')=\hat s_{\texttt{PAV}}(A,W)-\frac{n_i}{|P_i|}+\frac{n_j}{|W\cap P_j|+1}\geq \hat s_{\texttt{PAV}}(A,W)-\frac{n_i}{|P_i|}+\frac{n_j}{|P_j|}>\hat s_{\texttt{PAV}}(A,W)$. However, this contradicts that $W\in \texttt{PAV}(A)$ and thus shows that our assumption that $P_j\not\subseteq W$ was wrong. So, \texttt{PAV} satisfies party-proportionality. 

    For the other direction, let $f$ denote a Thiele rule satisfying party-proportionality and let $s$ denote its Thiele scoring function. First, we show that $s(1)>0$ and consider to this end the profile $A$ in which two voters approves party $P_1=\{c_1\}$ and one voter approves party $P_2=\{c_2,\dots, c_{k+1}\}$. Now, if $s(0)=0$, it holds for the committee $W=P_2$ that $W\in f(A)$ because $s$ is non-decreasing. However, we have that $\frac{n_1}{|P_1|}>\frac{n_2}{|P_2|}$ and $P_1\not\subseteq W$. Hence, this violates party-proportionality and therefore $s(1)>0$ must be true. Since Thiele rules are invariant under scaling $s$, we subsequently assume that $s(1)=1$. 

    As the next step, we suppose for contradiction that there is an index $\ell\in \{2,\dots,k\}$ such that $s(\ell)\neq \sum_{x=1}^{\ell} \frac{1}{x}$. Moreover, we assume that $\ell$ is minimal, i.e., $s(y)= \sum_{x=1}^{y} \frac{1}{x}$ for all $x<\ell$. Now, we proceed with a case distinction with respect to $s(\ell)$ and first consider the case that $s(\ell)>\sum_{x=1}^{\ell} \frac{1}{x}$. In this case, we define $\Delta=s(\ell)-\sum_{x=1}^{\ell} \frac{1}{x}$ and let $t\in\mathbb{N}$ such that $\Delta t>1$ and $\frac{\ell}{\ell-1}t>t+1$. Furthermore, let $A$ denote the party-list profile where $n_1=\ell t$ voters approve the party $P_1=\{c_1,\dots,c_\ell\}$ and every other candidate is in a singleton party that is approved by $t+1$ voters. Now, it can be checked that for all committees $W,W'\in\mathcal{W}_k$ with $P_1\subseteq W$ and $\ell'=|P_1\cap W'|<\ell$, the following inequality holds:
    \begin{align*}
        \hat s(A,W)&= t\ell s(\ell)+(k-\ell)(t+1)\\
        &=t\ell \sum_{x=1}^\ell \frac{1}{x} + t\ell\Delta + (k-\ell)(t+1)\\
        &>t\ell \sum_{x=1}^{\ell'} \frac{1}{x} +t\ell (\ell-\ell')\frac{1}{\ell} + \ell + (k-\ell)(t+1)\\
        &\geq t\ell \sum_{x=1}^{\ell'} \frac{1}{x} + (\ell-\ell')(t+1)+(k-\ell)(t+1)\\
        &= t\ell \sum_{x=1}^{\ell'} \frac{1}{x} + (k-\ell')(t+1)\\
        &=\hat s(A,W').
    \end{align*}
    This proves that $P_1\subseteq W$ for all $W\in f(A)$. However, this means that there is a party $P_j=\{c\}\neq P_1$ with $c\not\in W$ for some $W\in f(A)$. This contradicts party-proportionality as $\frac{n_j}{|P_j|}=t+1>t=\frac{n_1}{|P_1|}$ and $P_1\subseteq W$, so the assumption that $s(\ell)>\sum_{x=1}^\ell \frac{1}{x}$ must have been wrong. 

    For the second case, we suppose that $s(\ell)<\sum_{x=1}^\ell \frac{1}{x}$, define $\Delta=\sum_{x=1}^\ell \frac{1}{x}-s(\ell)$, and let $t\geq 2$ such that $t\ell\Delta>1$. Moreover, we consider the profile $A$ in which $t\ell$ voters approve the party $P_1=\{c_1,\dots,c_{\ell}\}$ and all other candidates are in singleton parties that are approved by $t-1$ voters. Now, we can compute for all committees $W,W'\in\mathcal{W}_k$ with $P_1\subseteq W$ and $|W'\cap P_1|=\ell-1$ that 
    \begin{align*}
        \hat s(A,W)&=t\ell s(\ell)+(k-\ell)(t-1)\\
        &=t\ell \sum_{x=1}^{\ell}\frac{1}{x} - t\ell\Delta + (k-\ell)(t-1)\\
        &>t\ell \sum_{x=1}^{\ell-1} \frac{1}{x} + t -1 + (k-\ell)(t-1)\\
        &=t\ell \sum_{x=1}^{\ell-1} \frac{1}{x} + (k-\ell+1)(t-1)\\
        &=\hat s(A,W').
    \end{align*}
    Hence, it holds for all committees $W\in f(A)$ that $P_1\not\subseteq W$. In turn, this means that there is a candidates $c\in W\setminus P_1$. However, this contradicts party-proportionality since $\frac{n_1}{|P_1|}=t>t-1=\frac{n_j}{|P_j|}$ for $P_j=\{c\}$ and thus, the assumption that $s(\ell)<\sum_{x=1}^{\ell} \frac{1}{x}$ is wrong. Hence, we now conclude that $s(\ell)=\sum_{x=1}^{\ell} \frac{1}{x}$ for all $\ell\in \{1,\dots,k\}$, so $f$ is \texttt{PAV}.
\end{proof} 

\SAV*
\begin{proof}
    First, it follows immediately from the definition of \texttt{SAV} that this rule satisfies both party-proportionality and aversion to unanimous committees. We thus focus on the converse direction and therefore consider a BSWAV rule $f$ that is party-proportional and has aversion to unanimous committees. Moreover, let $\alpha\in \mathbb{R}^m_{\geq 0}$ denote the weight vector of $f$. First, we show that $\alpha_1>0$ and consider to this end the profile $A$ where $2$ voters approve party $P_1=\{c_1\}$ and $1$ voter approves party $P_2=\{c_2,\dots, c_{k+1}\}$. Now, if $\alpha_1=0$, it is easy to see that $P_2\in f(A)$. However, this contradicts party-proportionality as $\frac{n_1}{|P_1|}>\frac{n_2}{|P_2|}$ and no member of $P_1$ is chosen in the committee  $W=P_2$. Hence, it must hold that $\alpha_1>0$ and we can rescale our weight vector such that $\alpha_1=1$. Moreover, we note that the entry $\alpha_m$ has no effect on $f$ as voters who approve all candidates increase the score of each committee by the same. We thus also suppose that $\alpha_m=\frac{1}{m}$. 

    Next, we assume for contradiction that there is an index $\ell\in \{2,\dots, m-1\}$ such that $\alpha_\ell\neq \frac{1}{\ell}$. Just as for \Cref{prop:AV}, we use a case distinction with respect to $\alpha_\ell$ and first assume that $\alpha_\ell>\frac{1}{\ell}$. In this case, we define $\Delta=\alpha_\ell-\frac{1}{\ell}$ and let $t\in\mathbb{N}$ such that $t\ell\Delta>1$. Moreover, we consider the profile $A$ in which $\ell t$ voters approve the party $P_1=\{c_1,\dots,c_\ell\}$ and each candidate $c\in\mathcal{C}\setminus P_i$ is uniquely approved by $t+1$ voters. We can now compute for every committee $W$ with $|W\cap P_1|=\ell'$ that $\hat s(A,W)=t\ell \ell' \alpha_\ell + (k-\ell')(t+1)=t\ell \ell' (\frac{1}{\ell}+\Delta)+(k-\ell')(t+1)=tk+k+\ell'(t\ell\Delta-1)$. Since $t\ell\Delta>1$, this means that the committees $W$ that maximize $|W\cap P_1|$ have maximal score. Now, if $\ell=|P_1|>k$, this means that $W\subseteq P_1$ for all $W\in f(A)$. However, $f$ then fails aversion to unanimous committees since $n_j=t+1>t=\frac{n_1}{|P_1|}$ for all other parties $P_j$. On the other hand, if $\ell\leq k$, then $f$ fails party-proportionality since $P_1\subseteq W$ for every $W\in f(A)$. This implies that there is another party $P_j=\{c\}$ with $P_j\cap W=\emptyset$. Since $\frac{n_j}{|P_j|}=t+1>t=\frac{n_1}{|P_1|}$, party-proportionality is violated. In summary, this means that the assumption that $\alpha_\ell> \frac{1}{\ell}$ has been wrong. 

    For the second case, suppose that $\alpha_\ell<\frac{1}{\ell}$. Moreover, we define $\Delta=\frac{1}{\ell}- \alpha_\ell$ and let $t\geq 2$ such that $t\ell\Delta>1$. Finally, consider the profile $A$ in which $t\ell$ voters approve $P_1=\{c_1,\dots, c_\ell\}$ and every other candidate is in a singleton party approved by $t-1$ voters. In this profile, every committee $W$ with $|W\cap P_1|=\ell'$ gets a score of $\hat s(A,W)=t\ell \ell' \alpha_\ell + (k-\ell')(t-1)=t\ell \ell' (\frac{1}{\ell}-\Delta)+(k-\ell')(t-1)=tk-k-\ell'(t\ell\Delta-1)$. Since $t\ell\Delta>1$, this means that the committees $W\in f(A)$ minimize $|W\cap P_1|$, so $P_1\not\subseteq W$. On the other hand, this implies that for every $W\in f(A)$ that there is another party $P_j=\{c\}$ such that $c\in W$. This, however, violates party-proportionality: $P_j\subseteq W$ and $P_1\not\subseteq W$ even though $\frac{n_j}{|P_j|}=t-1<t=\frac{n_1}{|P_1|}$. Hence, we also have in this case a contradiction and can now infer that $\alpha_\ell=\frac{1}{\ell}$, which means that $f$ is \texttt{SAV}. 
\end{proof}